\documentclass[11pt]{article} 
\pdfoutput=1
\usepackage[utf8]{inputenc}
\usepackage[T3,T1]{fontenc}
\usepackage[english]{babel}
\usepackage{xcolor}
\makeatletter
\newcommand{\note}[1]{{\textcolor{red}{[#1]}}\@latex@warning{Note: #1}}
\makeatother
\usepackage{amsmath,amsfonts,amssymb,amsthm}
\usepackage{enumitem}
\usepackage{mathtools}
\usepackage{booktabs}

\DeclareSymbolFont{tipa}{T3}{cmr}{m}{n}
\DeclareMathAccent{\invbreve}{\mathalpha}{tipa}{16}

\usepackage[mono=false]{libertine}
\usepackage[cmintegrals,libertine]{newtxmath}
\usepackage[cal=euler, scr=boondoxo]{mathalfa}
\usepackage{physics}

\useosf
\usepackage[a4paper,vmargin={3.5cm,3.5cm},hmargin={2.5cm,2.5cm}]{geometry}
\usepackage[font={small,sf}, labelfont={sf,bf}, margin=1cm]{caption}
\usepackage[colorlinks=true]{hyperref}
\usepackage{color,graphicx}

\usepackage{stackrel}
\usepackage{soul} 

\DeclareMathOperator*{\residue}{Res}

\theoremstyle{plain}
\newtheorem{theorem}{Theorem}
\newtheorem{corollary}[theorem]{Corollary}
\newtheorem{proposition}[theorem]{Proposition}
\newtheorem{lemma}[theorem]{Lemma}
\theoremstyle{definition}

\newcommand{\ind}{\mathbf{1}}
\newcommand{\R}{\mathbb{R}}

\newcommand{\Z}{\mathbb{Z}}

\newcommand{\rmd}{\mathrm{d}}
\newcommand{\breakline}{\nonumber\\&\quad\quad\quad}

\begin{document}

\title{\bf Topological recursion of the Weil--Petersson volumes of hyperbolic surfaces with tight boundaries}

\author{\textsc{Timothy Budd} \, and \, \textsc{Bart Zonneveld} \\[5mm]
{\small IMAPP, Radboud University, Nijmegen, The Netherlands.}\\
{\small \texttt{\href{mailto:t.budd@science.ru.nl}{t.budd@science.ru.nl}, \href{mailto:b.zonneveld@science.ru.nl}{b.zonneveld@science.ru.nl}}}}
\date{\today}
\maketitle
\begin{abstract}
	The Weil--Petersson volumes of moduli spaces of hyperbolic surfaces with geodesic boundaries are known to be given by polynomials in the boundary lengths. 
	These polynomials satisfy Mirzakhani's recursion formula, which fits into the general framework of topological recursion. 
	We generalize the recursion to hyperbolic surfaces with any number of special geodesic boundaries that are required to be \emph{tight}.
	A special boundary is tight if it has minimal length among all curves that separate it from the other special boundaries.
	The Weil--Petersson volume of this restricted family of hyperbolic surfaces is shown again to be polynomial in the boundary lengths.
    This remains true when we allow conical defects in the surface with cone angles in $(0,\pi)$ in addition to geodesic boundaries. 
	Moreover, the generating function of Weil--Petersson volumes with fixed genus and a fixed number of special boundaries is polynomial as well, and satisfies a topological recursion that generalizes Mirzakhani's formula. 

	This work is largely inspired by recent works by Bouttier, Guitter \& Miermont on the enumeration of planar maps with tight boundaries. 
	Our proof relies on the equivalence of Mirzakhani's recursion formula to a sequence of partial differential equations (known as the Virasoro constraints) on the generating function of intersection numbers. 

	Finally, we discuss a connection with JT gravity.
    We show that the multi-boundary correlators of JT gravity with defects (cone points or FZZT branes) are expressible in the tight Weil--Petersson volume generating functions, using a tight generalization of the JT trumpet partition function.

    \vspace{-3mm}
\end{abstract}

\begin{figure}[h!]
    \centering 
    \includegraphics[width=.39\linewidth]{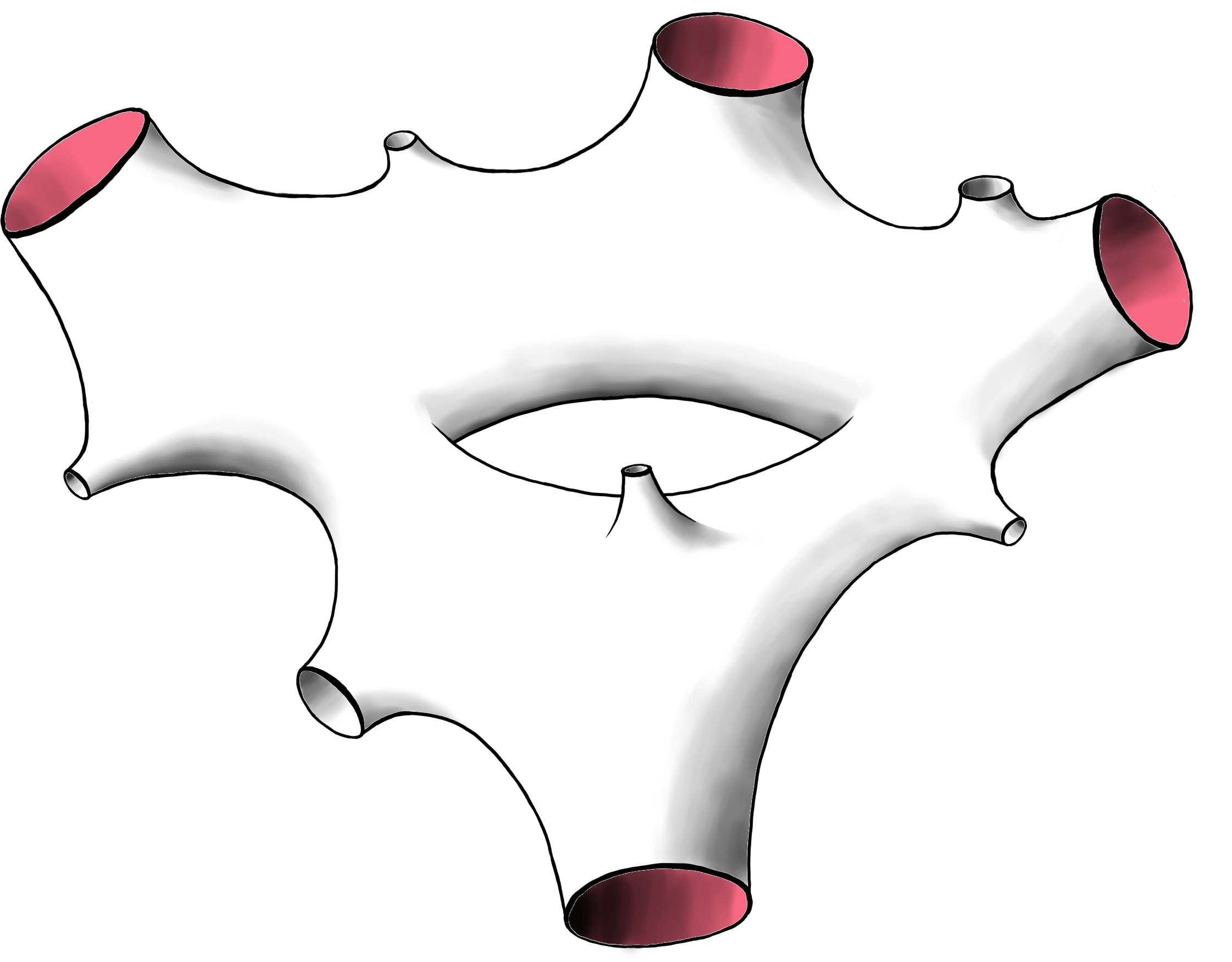}
    \caption{An example of a genus-$1$ hyperbolic surface with four tight boundaries (in pink).}
\end{figure}

\clearpage
\section{Introduction}

\subsection{Topological recursion of Weil--Petersson volumes}

In the celebrated work \cite{Mirzakhani2007} Mirzakhani established a recursion formula for the Weil--Petersson volume $V_{g,n}(\mathbf{L})$ of the moduli space of genus-$g$ hyperbolic surfaces with $n$ labeled boundaries of lengths $\mathbf{L} = (L_1, \ldots, L_n) \in \R_{>0}^n$.
Denoting $[n] = \{1,2,\ldots,n\}$ and using the notation $\mathbf{L}_I = (L_i)_{i\in I}$, $I\subset[n]$, for a subsequence of $\mathbf{L}$ and $\mathbf{L}_{\widehat{I}} = \mathbf{L}_{[n]\setminus I}$, the recursion can be expressed for $(g,n)\notin \{(0,3),(1,1)\}$ as
\begin{align}
V_{g,n}(\mathbf{L})&=
\frac{1}{2L_1} \int_0^{L_1}\dd{t}\int_0^\infty\dd{x}\int_0^\infty\dd{y}\,xy K_0(x+y,t) \Bigg[ V_{g-1,n+1}(x,y,\mathbf{L}_{\widehat{\{1\}}})\nonumber\\
&\mkern370mu +\!\!\!\!\sum_{\substack{g_1+g_2=g\\I\amalg J=\{2,\ldots,n\}}} V_{g_1,1+|I|}(x,\mathbf{L}_{I}) V_{g_2,1+|J|}(y,\mathbf{L}_{J})\Bigg]\nonumber\\
&+\frac{1}{2L_1}\int_0^{L_1}\dd{t}\int_0^\infty\dd{x}\sum_{j=2}^n\,x \left(K_0(x,t+L_j) + K_0(x,t-L_j)\right) V_{g,n-1}(x,\mathbf{L}_{\widehat{\{1,j\}}}),\label{eq:Mirzakhanirecursion}
\end{align}
where
\begin{align}
K_0(x,t)&=\frac{1}{1+\exp(\frac{x+t}{2})}+\frac{1}{1+\exp(\frac{x-t}{2})}.
\end{align}
Together with $V_{0,3}(\mathbf{L}) = 1$ and $V_{1,1}(\mathbf{L}) = \frac{L_1^2}{48} + \frac{\pi^2}{12}$ this completely determines $V_{g,n}$ as a symmetric polynomial in $L_1^2, \ldots, L_n^2$ of degree $3g-3+n$.

This recursion formula remains valid \cite{Tan_Generalizations_2006,Mirzakhani2007} when we replace one or more of the boundaries by cone points with cone angle $\alpha_i \in (0,\pi)$ if we assign to it an imaginary boundary length $L_i = i \alpha_i$.
Cone points with angles in $(0,\pi)$ are called \emph{sharp}, as opposed to blunt cone points that have angle in $(\pi,2\pi)$.
The Weil--Petersson volume of the moduli space of genus-$g$ surfaces with $n$ geodesic boundaries or sharp cone points is thus correctly computed by the polynomial $V_{g,n}(\mathbf{L})$. 

It was recognized by Eynard \& Orantin \cite{eynard2007weil} that Mirzakhani's recursion (in the case of geodesic boundaries) fits the general framework of topological recursion. 
To state this result explicitly one introduces for any $g,n\geq0$ satisfying $3g-3+n \geq 0$ the Laplace transformed\footnote{Note that, due to the extra factors $L_i$ in the integrand, $\mathcal{W}_{g,n}(\mathbf{z})$ is $(-1)^n$ times the partial derivative in each of the variables $z_1,\ldots,z_n$ of the Laplace transforms of $V_{g,n}(\mathbf{L})$, but we will refer to $\mathcal{W}_{g,n}(\mathbf{z})$ as the Laplace-transformed Weil--Petersson volumes nonetheless.} Weil--Petersson volumes
\begin{align}
    \omega_{g,n}^{(0)}(\mathbf{z})= \int_{0}^\infty \left[\prod_{i=1}^n\rmd L_i\, L_i e^{-z_i L_i} \right] V_{g,n}(\mathbf{L}),\label{eq:WPLaplace}
\end{align}
which are even polynomials in $z_1^{-1}, \ldots,z_n^{-1}$ of degree $6g-4+2n$, while setting 
\begin{equation}
    \omega_{0,1}^{(0)}(\mathbf{z})= 0, \qquad \omega_{0,2}^{(0)}(\mathbf{z}) = \frac{1}{(z_1-z_2)^2}.
\end{equation}
Then Mirzakhani's recursion \eqref{eq:Mirzakhanirecursion} translates into the recursion \cite[Theorem 2.1]{eynard2007weil} 
\begin{equation}
    \omega_{g,n}^{(0)}(\mathbf{z}) = \residue_{u\to0} \frac{\pi}{(z_1^2-u^2) \sin2\pi u} \Bigg[\omega^{(0)}_{g-1,n+1}(u,-u,\mathbf{z}_{\widehat{\{1\}}})+
\sum_{\substack{g_1+g_2=g\\I\amalg J=\{2,\ldots,n\}}}\omega^{(0)}_{g_1,1+|I|}(u,\mathbf{z}_{I})\omega^{(0)}_{g_2,1+|J|}(-u,\mathbf{z}_{J})\Bigg] 
\end{equation}
valid when $g,n\geq 0$ and $3g-3+n \geq 0$, which one may recognize as the recursion for the invariants $\omega^{(0)}_{g,n}(\mathbf{z})$ of the complex curve
\begin{equation}\label{eq:WPcurve}
\begin{dcases}
x(z) = z^2 \\
y(z) = \frac{1}{\pi}\sin(2\pi z).
\end{dcases}
\end{equation}
The main purpose of the current work is to generalize these recursion formulas to hyperbolic surfaces with so-called tight boundaries, which we introduce now.

\subsection{Hyperbolic surfaces with tight boundaries}

Let $S_{g,n}$ be a fixed topological surface of genus $g$ with $n$ boundaries and $\mathcal{T}_{g,n}(\mathbf{L})$ the Teichm\"uller space of hyperbolic structures on $S_{g,n}$ with geodesic boundaries of lengths $\mathbf{L} = (L_1,\ldots,L_n) \in \R_{>0}^n$.
Denote the boundary cycles\footnote{Our constructions will not rely on an orientation of the boundary cycles, but for definiteness we may take them clockwise (keeping the surface on the left-hand side when following the boundary).} by $\partial_1,\ldots,\partial_n$ and the free homotopy class of a cycle $\gamma$ in $S_{g,n}$ by $[\gamma]_{S_{g,n}}$.
For a hyperbolic surface $X \in \mathcal{T}_{g,n}(\mathbf{L})$ and a cycle $\gamma$, we denote by $\ell_\gamma(X)$ the length of $\gamma$, in particular $\ell_{\partial_i}(X) = L_i$.
The mapping class group of $S_{g,n}$ is denoted $\mathrm{Mod}_{g,n}$ and the quotient of $\mathcal{T}_{g,n}(\mathbf{L})$ by its action leads to the moduli space 
\begin{align*}
	\mathcal{M}_{g,n}(\mathbf{L}) = \mathcal{T}_{g,n}(\mathbf{L})/\mathrm{Mod}_{g,n}.
\end{align*}

\begin{figure}
	\centering
    \includegraphics[height=.35\linewidth]{images/Tight_surface_normal_color1b.png}%
	\includegraphics[height=.35\linewidth]{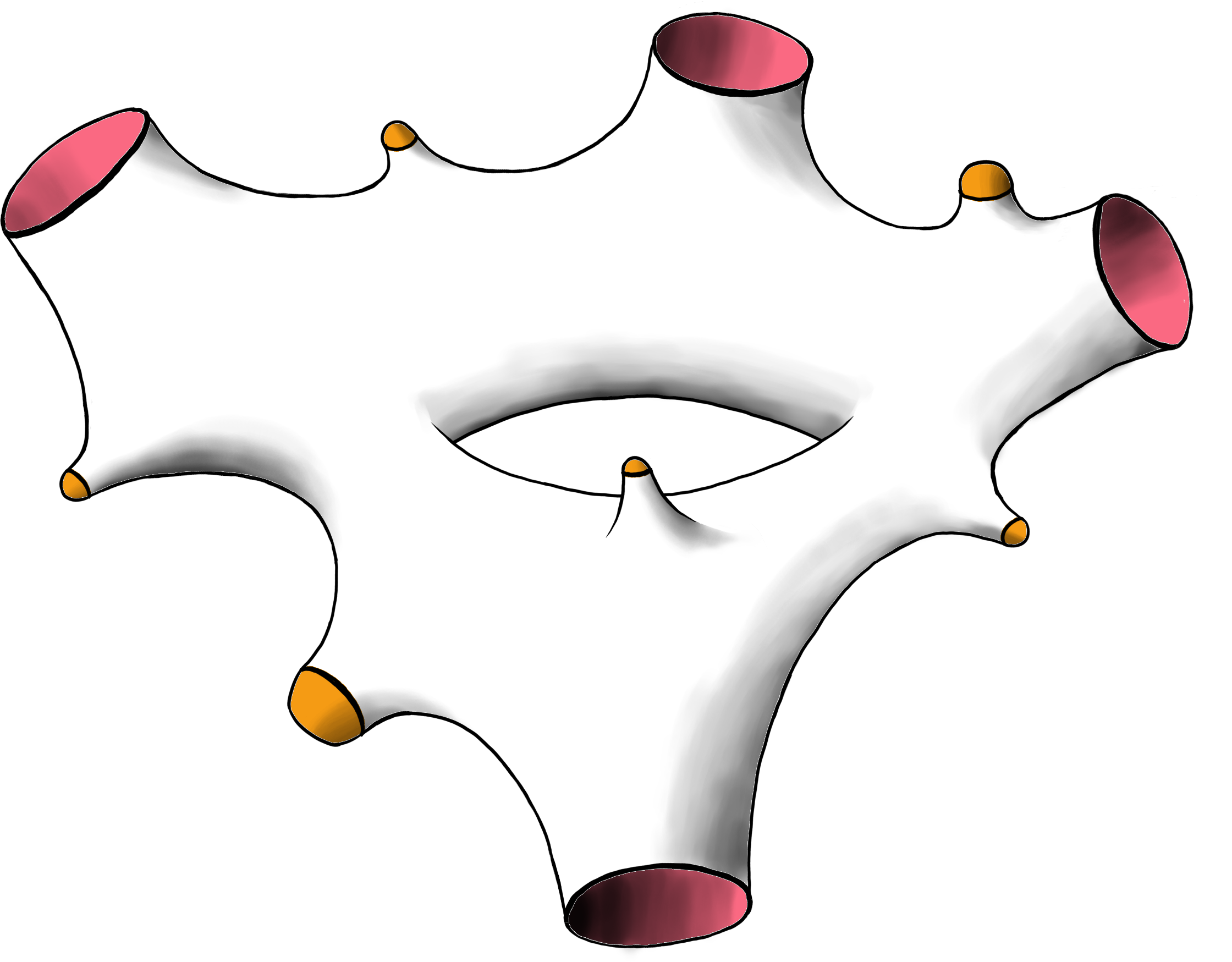}
	\caption{A surface $S_{g,n+p}$ (left) and the result $\invbreve{S}_{g,n,p}$ (right) of capping off the last $p$ boundaries, with in this case $g=1$, $n=4$ and $p=6$. If the original surface was equipped with a hyperbolic structure, it is natural to think of the resulting surface as the Riemannian manifold obtained by attaching constantly curved hemispheres of appropriate diameter to the $p$ geodesic boundaries.  \label{fig:capping}}
\end{figure}

Let us denote by $\invbreve{S}_{g,n,p} \supset S_{g,n+p}$ the topological surface obtained from $S_{g,n+p}$ by capping off the last $p$ boundaries with disks (Figure~\ref{fig:capping}). Note that the free homotopy classes $[\cdot]_{S_{g,n+p}}$ of $S_{g,n+p}$ are naturally partitioned into the free homotopy classes $[\cdot]_{\invbreve{S}_{g,n,p}}$ of $\invbreve{S}_{g,n,p}$. 
In particular, $[\partial_{j}]_{S_{g,n+p}}$ for $j=n+1,\ldots,n+p$ are all contained in the null-homotopy class of $\invbreve{S}_{g,n,p}$.
For $i=1,\ldots,n$ the boundary $\partial_i$ of $X \in \mathcal{M}_{g,n+p}(\mathbf{L})$ is said to be \emph{tight in }$\invbreve{S}_{g,n,p}$ if $\partial_i$ is the only simple cycle $\gamma$ in $[\partial_i]_{\invbreve{S}_{g,n,p}}$ of length $\ell_\gamma(X) \leq L_i$.
Remark that both $[\partial_i]_{S_{g,n+p}}$ and $[\partial_i]_{\invbreve{S}_{g,n,p}}$ for $i=1,\ldots,n$ are $\mathrm{Mod}_{g,n+p}$-invariant, so these classes are well-defined at the level of the moduli space.
This allows us to introduce the \emph{moduli space of tight hyperbolic surfaces} 
\begin{align}
	\mathcal{M}^{\text{tight}}_{g,n,p}(\mathbf{L}) &= \{ X \in \mathcal{M}_{g,n+p}(\mathbf{L}) : \partial_{1},\ldots,\partial_{n}\text{ are tight in }\invbreve{S}_{g,n,p}\} \subset \mathcal{M}_{g,n+p}(\mathbf{L}).
\end{align}
Note that $\mathcal{M}^{\text{tight}}_{g,n,0}(\mathbf{L}) = \mathcal{M}^{\text{tight}}_{g,0,n}(\mathbf{L}) = \mathcal{M}_{g,n}(\mathbf{L})$, while $\mathcal{M}^{\text{tight}}_{0,1,p}(\mathbf{L}) = \emptyset$ because $\partial_{1}$ is null-homotopic and $\mathcal{M}^{\text{tight}}_{0,2,p}(\mathbf{L}) = \emptyset$ because $[\partial_{1}]_{\invbreve{S}_{0,2,p}}=[\partial_{2}]_{\invbreve{S}_{0,2,p}}$ and therefore $\partial_{1}$ and $\partial_{2}$ can never both be the unique shortest cycle in their class.
In general, it is an open subset of $\mathcal{M}_{g,n+p}(\mathbf{L})$ and therefore it inherits the Weil--Petersson symplectic structure and Weil--Petersson measure $\rmd \mu_{\mathrm{WP}}$ from $\mathcal{M}_{g,n+p}(\mathbf{L})$.
The corresponding \emph{tight} Weil--Petersson volumes are denoted 
\begin{align}
	T_{g,n,p}(\mathbf{L}) = \int_{\mathcal{M}^{\text{tight}}_{g,n,p}(\mathbf{L})} \rmd \mu_{\mathrm{WP}} \quad\leq\quad V_{g,n+p}(\mathbf{L}),
\end{align}
such that $T_{g,n,0}(\mathbf{L}) = T_{g,0,n}(\mathbf{L}) = V_{g,n}(\mathbf{L})$ and $T_{0,1,p}(\mathbf{L})=T_{0,2,p}(\mathbf{L})=0$.

We can extend this definition to the case in which one or more of the boundaries $\partial_{n+1},\ldots,\partial_{n+p}$ is replaced by a sharp cone point with cone angle $\alpha_i \in (0,\pi)$.
In this case we make the usual identification $L_i = i \alpha_i$, and still denote the corresponding Weil--Petersson volume by $T_{g,n,p}(\mathbf{L})$.
Our first result is the following.

\begin{proposition}\label{prop:Tpolynomial}
    For $g,n,p\geq 0$ such that $3g-3 + n \geq 0$, the tight Weil--Petersson volume $T_{g,n,p}(\mathbf{L})$ of genus $g$ surfaces with $n$ tight boundaries and $p$ geodesic boundaries or sharp cone points is a polynomial in $L_1^2, \ldots, L_{n+p}^2$ of degree $3g-3+n+p$ that is symmetric in $L_1,\ldots, L_n$ and symmetric in $L_{n+1},\ldots,L_{n+p}$.
\end{proposition}

For most of the upcoming results we maintain the intuitive picture that the tight boundaries are the ``real'' boundaries of the surface, whose number and lengths we specify, while we allow for an arbitrary number of other boundaries or cone points that we treat as defects in the surface.
To this end we would like to encode the volume polynomials in generating functions that sum over the number of defects with appropriate weights.
A priori it is not entirely clear what is the best way to organize such generating functions, so to motivate our definition we take a detour to a natural application of Weil--Petersson volumes in random hyperbolic surfaces.

\subsection{Intermezzo: Random (tight) hyperbolic surfaces}\label{sec:probintermezzo}

If we fix $g$, $n$ and $\mathbf{L} \in ([0,\infty) \cup i (0,\pi))^n$, then upon normalization by $1/V_{g,n}(\mathbf{L})$ the Weil--Petersson measure $\rmd\mu_{\mathrm{WP}}$ provides a well-studied probability measure on $\mathcal{M}_{g,n}(\mathbf{L})$ defining the \emph{Weil--Petersson random hyperbolic surface}, see e.g. \cite{Mirzakhani_Growth_2013,Guth_Pants_2011,Mirzakhani_Lengths_2019,Gilmore_Short_2021,Monk_Benjamini_2022}.
A natural way to extend the randomness to the boundary lengths or cone angles is by choosing a (Borel) measure $\mu$ on $[0,\infty) \cup i (0,\pi)$ and first sampling $\mathbf{L} \in ([0,\infty) \cup i (0,\pi))^n$ from the probability measure
\begin{equation}\label{eq:wplengthmeasure}
    \frac{1}{\mu^{\otimes n}(V_{g,n})} V_{g,n}(\mathbf{L})\rmd\mu(L_1)\cdots\rmd\mu(L_n), \quad \mu^{\otimes n}(V_{g,n}) \coloneqq \int V_{g,n}(\mathbf{L})\rmd\mu(L_1)\cdots\rmd\mu(L_n),
\end{equation}
and then sampling a Weil--Petersson random hyperbolic surface on $\mathcal{M}_{g,n}(\mathbf{L})$.
If the genus-$g$ \emph{partition function}\footnote{We use the physicists' convention of writing the argument $\mu$ in square brackets to signal it is a functional dependence (in the sense of calculus of variations).}
\begin{equation}\label{eq:wppartitionfunction}
    F_g[\mu] = \sum_{n\geq 0} \frac{\mu^{\otimes n}(V_{g,n})}{n!}
\end{equation}
converges, we can furthermore make the size $n\geq 0$ random by sampling it with the probability $\mu^{\otimes n}(V_{g,n})/(n! F_g[\mu])$.
The resulting random surface (of random size) is called the genus-$g$ \emph{Boltzmann hyperbolic surface} with weight $\mu$.
See the upcoming work \cite{Budd_Statistics_} for some of its statistical properties.

A natural extension is to consider the genus-$g$ Boltzmann hyperbolic surface with $n$ tight boundaries of length $\mathbf{L}=(L_1,\ldots,L_n)$, where the number $p$ of defects and their boundary lengths/cone angles $\mathbf{K}=(K_1,\ldots,K_p)$ are random.
The corresponding partition function is 
\begin{align}
    T_{g,n}(\mathbf{L};\mu] = \sum_{p\geq 0} \frac{\mu^{\otimes p}(T_{g,n,p})(\mathbf{L})}{p!}, \quad \mu^{\otimes p}(T_{g,n,p})(\mathbf{L}) \coloneqq \int T_{g,n,p}(\mathbf{L},\mathbf{K})\rmd\mu(K_1)\cdots\rmd\mu(K_p). 
\end{align}
If it is finite, we can sample $p$ with probability $\mu^{\otimes p}(T_{g,n,p})(\mathbf{L})/(p!T_{g,n}(\mathbf{L};\mu])$ and then $\mathbf{K}$ from the probability measure
\begin{align}\label{eq:tightwplengthmeasure}
    \frac{1}{\mu^{\otimes p}(T_{g,n,p})(\mathbf{L})} T_{g,n,p}(\mathbf{L},\mathbf{K})\rmd\mu(K_1)\cdots\rmd\mu(K_n)
\end{align}
and then finally a random tight hyperbolic surface from the probability measure $\rmd \mu_{\mathrm{WP}}/T_{g,n,p}(\mathbf{L},\mathbf{K})$ on $\mathcal{M}^{\text{tight}}_{g,n,p}(\mathbf{L},\mathbf{K})$.
Note that for $\mu=0$ the genus-$g$ Boltzmann hyperbolic surface with $n$ tight boundaries reduces to the Weil--Petersson random hyperbolic surface we started with.

The important observation for the current work is that the partition functions $F_g[\mu]$ and $T_{g,n}(\mathbf{L};\mu]$ of these random surfaces can be thought of as (multivariate, exponential) generating functions of the volumes $V_{g,n}(\mathbf{L})$ and $T_{g,n,p}(\mathbf{L},\mathbf{K})$ if we treat $\mu$ as a formal generating variable.
Since we will not be concerned with the details of the measures \eqref{eq:wplengthmeasure} and \eqref{eq:tightwplengthmeasure} and $F_g[\mu]$ and $T_{g,n}(\mathbf{L};\mu]$ only depend on the even moments $\int L^{2k} \rmd\mu(L)$, we can instead take these moments as the generating variables.

\subsection{Generating functions}

To be precise, we let a \emph{weight} $\mu$ be a real linear function on the ring of even, real polynomials (i.e. $\mu \in \R[K^2]^*$).
For an even real polynomial $f$ we use the suggestive notation 
\begin{align}
    \mu(f) = \int \rmd \mu(K) f(K),
\end{align}
making it clear that the notion of weight generalizes the Borel measure described in the intermezzo above.
For $L \in [0,\infty)\cup i(0,\pi)$, the Borel measure given by the delta measure $\delta_L$ at $L$ gives a simple example of a weight $\mu=\delta_L$ satisfying $\delta_L(f) = f(L)$.
The choice of weight $\mu$ is clearly equivalent to the choice of a sequence of \emph{times} $(t_0,t_1,\ldots) \in \R^{\Z_{\geq0}}$ recording the evaluations of $\mu$ on the even monomials, up to a conventional normalization, 
\begin{align}
    t_k[\mu]= \frac{2}{4^k k!} \mu(K^{2k}) = \int\rmd\mu(K)\,\frac{2K^{2k}}{4^kk!}.\label{eq:times}
\end{align}
Naturally we can interpret $\mu^{\otimes p} \in (\R[K^2]^*)^{\otimes p}$ as an element of $(\R[K^2]^{\otimes p})^* \cong \R[K_1^2,\ldots,K_p^2]^*$ by setting $\mu^{\otimes p}(f_1(K_1)\cdots f_p(K_p)) = \mu(f_1)\cdots\mu(f_p)$ for even polynomials $f_1,\ldots,f_p$ and extending by linearity.
More generally, we can view $\mu^{\otimes p}$ as a linear map $\R[L_1^2,\ldots,L_n^2,K_1^2,\ldots,K_p^2] \cong \R[L_1^2,\ldots,L_n^2][K_1^2,\ldots,K_p^2] \to \R[L_1^2,\ldots,L_n]$.
We use the notation
\begin{align}
    \mu^{\otimes p}(f) = \int f(\mathbf{K}) \rmd\mu(K_1)\cdots\rmd\mu(K_p), \qquad \mu^{\otimes p}(f)(\mathbf{L}) = \int f(\mathbf{L},\mathbf{K}) \rmd\mu(K_1)\cdots\rmd\mu(K_p).
\end{align} 
One can then naturally introduce the generating function $F[\mu]$ of a collection of symmetric, even polynomials $f_1(L_1),f_2(L_1,L_2),\ldots$ via $F[\mu] = \sum_{p\geq 0} \frac{\mu^{\otimes p}(f_p)}{p!}$.
Then the generating function of tight Weil--Petersson volumes is defined to be
\begin{align}
    T_{g,n}(\mathbf{L};\mu] &= \sum_{p=0}^\infty \frac{\mu^{\otimes p}(T_{g,n,p})(\mathbf{L})}{p!} = \sum_{p=0}^\infty \frac{1}{p!} \int \rmd\mu(K_{1})\cdots \int\rmd\mu(K_{p})T_{g,n,p}(\mathbf{L},\mathbf{K}),\label{eq:tightgenfun}
\end{align}
which we interpret in the sense of a formal power series, so we do not have to worry about convergence.
We could make this more precise by fixing a weight $\mu$ and considering $T_{g,n}(\mathbf{L};x \mu] \in \R[\![x]\!]$ as a univariate formal power series in $x$.
Or we could view $T_{g,n}(\mathbf{L};\mu]$ as a multivariate formal power series in the times $(t_0,t_1,\ldots)$ defined in \eqref{eq:times}.

What is important is that we can make sense of the \emph{functional derivative} $\frac{\delta}{\delta\mu(L)}$ on these types of series defined by
\begin{align}
    \frac{\delta}{\delta\mu(L)} P[\mu] = \frac{\partial}{\partial x} P[\mu + x \delta_L] \Big|_{x=0}.
\end{align}
In particular, if $f(\mathbf{L},\mathbf{K})$, with $\mathbf{L} = (L_1,\ldots,L_n)$ and $\mathbf{K} = (K_1,\ldots,K_p)$, is an even polynomial that is symmetric in $K_1,\ldots,K_p$ then
\begin{align}
    \frac{\delta}{\delta\mu(L)} \mu^{\otimes p}(f)(\mathbf{L}) = p\, \mu^{\otimes p-1}(f)(\mathbf{L},L).
\end{align}
At the level of the generating function we thus have
\begin{align}
    \frac{\delta}{\delta\mu(L)}T_{g,n}(\mathbf{L};\mu] = \sum_{p=0}^\infty \frac{1}{p!} \int \rmd \mu(K_{1})\cdots \int\rmd \mu(K_{p}) T_{g,n,p+1}(\mathbf{L},L,\mathbf{K}).
\end{align}
In terms of formal power series in the times we may instead identify the functional derivative in terms of the formal partial derivatives as
\begin{align}
    \frac{\delta}{\delta \mu(L)} = \sum_{k=0}^\infty \frac{2 L^{2k}}{4^{k}k!} \frac{\partial}{\partial t_k}, \qquad \frac{\delta}{\delta \mu(0)} = 2\frac{\partial}{\partial t_0}.
\end{align}

\subsection{Main results}

To state our main results about $T_{g,n}(\mathbf{L};\mu]$, we need to introduce the generating function $R[\mu]$ as the unique formal power series solution satisfying $R[\mu] = \int\rmd\mu(L) + O(\mu^2)$ to
\begin{equation}\label{eq:Zdef}
Z(R[\mu];\mu] = 0,\qquad Z(r;\mu]\coloneqq \frac{\sqrt{r}}{\sqrt{2}\pi}J_1(2\pi\sqrt{2r}) -  \int\rmd\mu(L)\,I_0(L\sqrt{2r})
\end{equation}
where $I_0$ and $J_1$ are (modified) Bessel functions. 
Let also the \emph{moments} $M_k[\mu]$ be the defined recursively via
\begin{equation}\label{eq:momentdef}
M_0[\mu] = \frac{1}{\,\fdv{R}{\mu(0)}\,},\quad M_{k}[\mu] = M_0[\mu] \,\fdv{M_{k-1}[\mu]}{\mu(0)}, \qquad k\geq 1,
\end{equation}
where the reciprocal in the first identity makes sense because $\fdv{R}{\mu(0)} = 1 + O(\mu)$.
Alternatively, for $k\geq0$ we may express $M_k$ as
\begin{equation}\label{eq:Mexplicit}
M_k[\mu] = Z^{(k+1)}(R[\mu];\mu]=\qty(\frac{-\sqrt{2}\pi}{\sqrt{R[\mu]}})^{k}J_k(2\pi\sqrt{2R[\mu]}) -  \int\rmd\mu(L)\,\qty(\frac{L}{\sqrt{2R[\mu]}})^{k+1}I_{k+1}(L\sqrt{2 R[\mu]}) ,
\end{equation}
where $Z^{(k+1)}(r;\mu]$ denotes the $(k+1)$th derivative of $Z(r;\mu]$ with respect to $r$.

We further consider the series
\begin{align}\label{eq:etadef}
\eta(u;\mu]=\sum_{p=0}^\infty  M_p[\mu] \frac{u^{2p}}{(2p+1)!!},\qquad \hat X(u;\mu] = \frac{ \sin(2\pi u)}{2\pi u \,\eta(u;\mu]},
\end{align}
which we both interpret as formal power series in $u$ with coefficients that are formal power series in $\mu$.
The reciprocal in the second definition is well-defined because $\eta(0) = M_0[\mu] = 1 + O(\mu)$.
We can now state our main result that generalizes Mirzakhani's recursion formula.

\begin{theorem} \label{the:tight_top_recursion}
    The tight Weil--Petersson volume generating functions $T_{g,n}(\mathbf{L})$ satisfy 
    \begin{align}
    T_{g,n}(\mathbf{L})&=
    \frac{1}{2L_1} \int_0^{L_1}\dd{t}\int_0^\infty\dd{x}\int_0^\infty\dd{y}\,xy K(x+y,t;\mu] \Bigg[ T_{g-1,n+1}(x,y,\mathbf{L}_{\widehat{\{1\}}})\nonumber\\
    &\mkern370mu +\!\!\!\!\sum_{\substack{g_1+g_2=g\\I\amalg J=\{2,\ldots,n\}}} T_{g_1,1+|I|}(x,\mathbf{L}_{I}) T_{g_2,1+|J|}(y,\mathbf{L}_{J})\Bigg]\nonumber\\
    &+\frac{1}{2L_1}\int_0^{L_1}\dd{t}\int_0^\infty\dd{x}\sum_{j=2}^n\,x \left(K(x,t+L_j;\mu] + K(x,t-L_j;\mu]\right) T_{g,n-1}(x,\mathbf{L}_{\widehat{\{1,j\}}}),
    \end{align}
    which is the same recursion formula \eqref{eq:Mirzakhanirecursion} as for the Weil--Petersson volumes $V_{g,n}(\mathbf{L})$ except that the kernel $K_0(x,t)$ is replaced by the ``convolution''
    \begin{align}
    K(x,t,\mu]&= \int_{-\infty}^\infty \rmd X(z)\, K_0(x+z,t),
    \end{align}
    where $X(z) = X(z;\mu]$ is a measure on $\R$ determined by its two-sided Laplace transform
    \begin{align}
    \int_{-\infty}^\infty\rmd X(z)\, e^{-uz} = \hat X(u;\mu] = \frac{ \sin(2\pi u)}{2\pi u \,\eta(u;\mu]}.\label{eq:Xhat}
    \end{align}
Furthermore, we have
\begin{align}
T_{0,3}(L_1,L_2,L_3;\mu]&=\frac{1}{M_0[\mu]},\label{eq:T03}\\
T_{1,1}(L;\mu]&=-\frac{M_1[\mu]}{24M_0[\mu]^2}+\frac{L^2}{48M_0[\mu]}.
\end{align}
\end{theorem}

\noindent
We will not specify precisely what it means to have a measure $X(z;\mu]$ that itself is a formal power series in $\mu$.
Importantly its moments $\int_{-\infty}^\infty \rmd X(z) \,z^p = (-1)^p p![u^p]\hat{X}(u;\mu]$ are formal power series in $\mu$, so for any $x,t\in\R$
\begin{align}
    K(x,t;\mu] = \sum_{p=0}^\infty (-1)^p \frac{\partial^p}{\partial x^p}K_0(x,t) [u^p]\hat{X}(u;\mu]
\end{align}
is a formal power series in $\mu$ as well.

In the case $\mu=0$, it is easily verified that 
\begin{align}
    M_k[0] &= \frac{(-2\pi^2)^k}{k!},\\
    \eta(u,0] &= \frac{\sin(2\pi u)}{2\pi u},
\end{align}
so $\hat{X}(u;0] = 1$ and $X(z;0] = \delta_0(z)$ and therefore one retrieves Mirzakhani's kernel $K(x,t) = K_0(x,t)$.
Given that the form of Mirzakhani's recursion is unchanged except for the kernel, this strongly suggests that the Laplace transforms 
\begin{align}
    \omega_{g,n}(\mathbf{z})=\omega_{g,n}(\mathbf{z};\mu]\coloneqq \int_{0}^\infty \left[\prod_{i=1}^n\rmd L_i\, L_i e^{-z_i L_i} \right] T_{g,n}(\mathbf{L};\mu]\label{eq:tightlaplace}
\end{align}
of the tight Weil--Petersson volumes can be obtained as invariants in the framework of topological recursion as well.
When $\mu=0$ this reduces to the Laplace-transformed Weil--Petersson volumes $\omega_{g,n}(\mathbf{z};0] = \omega_{g,n}^{(0)}(\mathbf{z})$ defined in \eqref{eq:WPLaplace}.
The following theorem shows that this is the case in general.

\begin{theorem}\label{the:tight_spectral_curve}
    Setting $\omega_{0,2}(\mathbf{z})=(z_1-z_2)^{-2}$ and $\omega_{0,0}(\mathbf{z}) = \omega_{0,1}(\mathbf{z})=0$, the Laplace transforms \eqref{eq:tightlaplace} satisfy for every $g,n\geq 0$ such that $3g-3+n \geq 0$ the recursion
    \begin{align}
        \omega_{g,n}(\mathbf{z})
        &=\residue_{u\to0}\frac{1}{2u\qty(z_1^2-u^2)\eta(u;\mu]}\Bigg[\omega_{g-1,n+1}(u,-u,\mathbf{z}_{\widehat{\{1\}}})+
        \sum_{\substack{g_1+g_2=g\\I\amalg J=\{2,\ldots,n\}}}\omega_{g_1,|I|}(u,\mathbf{z}_{I})\omega_{g_2,|J|}(-u,\mathbf{z}_{J})\Bigg].
    \end{align}
    These correspond precisely to the invariants of the curve 
    \begin{align}
        \begin{dcases}
        x=z^2\\
        y=2z\,\eta(z;\mu]. 
        \end{dcases}
    \end{align}
\end{theorem}

Another consequence of Theorem~\ref{the:tight_top_recursion} is that the tight Weil--Petersson volumes $T_{g,n}(\mathbf{L};\mu]$ for all $g,n\geq 0$, such that $n\geq 3$ for $g=0$ and $n\geq 1$ for $g=1$, are expressible as a rational polynomial in $L_1^2,\ldots, L_n^2$ and $M_0^{-1}, M_1 , M_2, \ldots$.

Besides satisfying a recursion in the genus $g$ and the number of tight boundaries $n$, these also satisfy a recurrence relation in $n$ only.

\begin{theorem}\label{the:recursion_n} 
    For all $g,n\geq 0$, such that $n\geq 3$ for $g=0$ and $n\geq 1$ for $g=1$, we have that
    \begin{align}
        T_{g,n}(\mathbf{L};\mu] = \frac{1}{M_0^{2g-2+n}} \mathcal{P}_{g,n}\left(\mathbf{L},\frac{M_1}{M_0},\ldots,\frac{M_{3g-3+n}}{M_0}\right),
    \end{align}
    where $\mathcal{P}_{g,n}(\mathbf{L},\mathbf{m})$ is a rational polynomial in $L_1^2,\ldots,L_n^2,m_1,\ldots,m_{3g-3+n}$.
    This polynomial is symmetric and of degree $3g-3+n$ in $L_1^2,\ldots,L_n^2$, while $\mathcal{P}_{g,n}(\sqrt\sigma \mathbf{L},\sigma m_1,\sigma^2 m_2, \sigma^3 m_3, \ldots)$ is homogeneous of degree $3g-3+n$ in $\sigma$.
    For all $g\geq 0$, $n\geq 1$ such that $2g-3+n>0$ the polynomial $\mathcal{P}_{g,n}(\mathbf{L},\mathbf{m})$ can be obtained from $\mathcal{P}_{g,n-1}(\mathbf{L},\mathbf{m})$ via the recursion relation
    \begin{align}
        \mathcal{P}_{g,n}(\mathbf{L},\mathbf{m}) &= \sum_{p=1}^{3g-4+n} \left(m_{p+1} - \frac{L_1^{2p+2}}{2^{p+1}(p+1)!}-m_1 m_p + \frac{1}{2}L_1^2 m_p\right) \frac{\partial \mathcal{P}_{g,n-1}}{\partial m_p}(\mathbf{L}_{\widehat{\{1\}}},\mathbf{m}) \nonumber\\
        &\quad  + (2g-3+n)(-m_1+\tfrac12 L_1^2) \mathcal{P}_{g,n-1}(\mathbf{L}_{\widehat{\{1\}}},\mathbf{m})
        + \sum_{i=2}^n \int \dd{L_i} L_i \, \mathcal{P}_{g,n-1}(\mathbf{L}_{\widehat{\{1\}}},\mathbf{m}),\label{eq:Pgnrecursion}
    \end{align}
    where we use the shorthand notation $\int \dd{L} L f(L,\ldots) = \int_0^L \dd{x} x f(x,\ldots)$.
    Furthermore, we have 
    \begin{align}
    \mathcal{P}_{0,3}(\mathbf{L})&=1\\
    \mathcal{P}_{1,1}(L_1,m_1)&=\frac{1}{24}(-m_1+\tfrac12 L_1^2)
    \end{align}
    and $\mathcal{P}_{g,0}$ for $g\geq2$ is given by
    \begin{equation}
    \mathcal{P}_{g,0}(m_1,\ldots,m_{3g-3}) = \sum_{\substack{d_2,d_3,\ldots\geq 0\\ \sum_{k\geq 2} (k-1)d_k = 3g-3}}\!\!\!\! \langle \tau_2^{d_2}\tau_3^{d_3}\cdots \rangle_g \prod_{k\geq 2} \frac{(-m_{k-1})^{d_k}}{d_k!},
    \end{equation}
    where $\langle \tau_2^{d_2}\tau_3^{d_3}\cdots \rangle_g$ are the $\psi$-class intersection numbers on the moduli space $\mathcal{M}_{g,n}$ with $n = \sum_k d_k \leq 3g-3$ marked points.
\end{theorem}

\noindent For instance, the first few applications of the recursion yield
\begin{align*}
    \mathcal{P}_{0,4}(\mathbf{L},\mathbf{m}) &= \frac{1}{2} (L_1^2+\cdots+L_4^2) - m_1,\\
    \mathcal{P}_{0,5}(\mathbf{L},\mathbf{m}) &= \frac{1}{8} (L_1^4+\cdots+L_5^4) + \frac{1}{2}(L_1^2L_2^2+\cdots+L_4^2L_5^2) - \frac{3}{2}(L_1^2+\cdots+L_5^2)m_1 + 3 m_1^2 - m_2,\\
    \mathcal{P}_{1,2}(\mathbf{L},\mathbf{m}) &= \frac{1}{192}(L_1^4+L_2^4) + \frac{1}{96}L_1^2L_2^2 -\frac{1}{24}(L_1^2+L_2^2)m_1 + \frac{1}{12}m_1^2-\frac{1}{24}m_2^2.
\end{align*}
Note that this provides a relatively efficient way of calculating the Weil--Petersson volumes $V_{g,n}(\mathbf{L})$ from the polynomial $\mathcal{P}_{g,0}$, since
\begin{align*}
    V_{g,n}(\mathbf{L}) = T_{g,n}(\mathbf{L};0] = \mathcal{P}_{g,n}(\mathbf{L},\mathbf{m}) \Big|_{m_k = (-2\pi^2)^k/k!}.
\end{align*} 
A simple corollary of Theorem~\ref{the:recursion_n} is that the volumes satisfy \emph{string} and \emph{dilaton} equations generalizing those for the Weil--Petersson volumes derived by Do \& Norbury in \cite[Theorem~2]{Do_Weil_2009}.
\begin{corollary}\label{cor:stringdilaton}
    For all $g \geq 0$ and $n \geq 1$, such that $n\geq 4$ when $g=0$ and $n\geq 2$ when $g=1$, we have the identities 
    \begin{align}
        \sum_{p=0}^\infty 2^p p!\, M_p[\mu] \, [L_{1}^{2p}] T_{g,n}(\mathbf{L};\mu] &= \sum_{j=2}^n \int \rmd L_j \,L_j T_{g,n-1}(\mathbf{L}_{\widehat{\{1\}}};\mu]+\ind_{\{g=0,n=3\}},\\
        \sum_{p=1}^\infty 2^p p!\, M_{p-1}[\mu]\, [L_{1}^{2p}] T_{g,n}(\mathbf{L};\mu] &= (2g-3+n)T_{g,n-1}(\mathbf{L}_{\widehat{\{1\}}};\mu],
    \end{align}
    where the notation $[L_1^{2p}]T_{g,n}(\mathbf{L};\mu]$ refers to the coefficient of $L_1^{2p}$ in the polynomial $T_{g,n}(\mathbf{L};\mu]$.
\end{corollary}

As explained in \cite{Do_Weil_2009}, the string and dilaton equations for symmetric polynomials, in particular for the Weil--Petersson volumes, give rise to a recursion in $n$ for genus $0$ and $1$. Using Theorem \ref{the:recursion_n}, we also get such a recursion for higher genera in the case of tight Weil--Petersson volumes.

\subsection{Idea of the proofs}

This work is largely inspired by the recent work \cite{Bouttier_Bijective_2022} of Bouttier, Guitter \& Miermont.
There the authors consider the enumeration of planar maps with three boundaries, i.e.\ graphs embedded in the triply punctured sphere, see the left side of Figure~\ref{fig:tightmap}.
Explicit expressions for the generating functions of such maps, also known as pairs of pants, with controlled face degrees were long known, but they show that these generating functions become even simpler when restricting to \emph{tight} pairs of pants, in which the three boundaries are required to have minimal length (in the sense of graph distances) in their homotopy classes.
They obtain their enumerative results on tight pairs of pants in a bijective manner by considering a canonical decomposition of a tight pair of pants into certain triangles and diangles, see Figure~\ref{fig:tightmap}. 

\begin{figure}
    \centering
    \includegraphics[width=.8\linewidth]{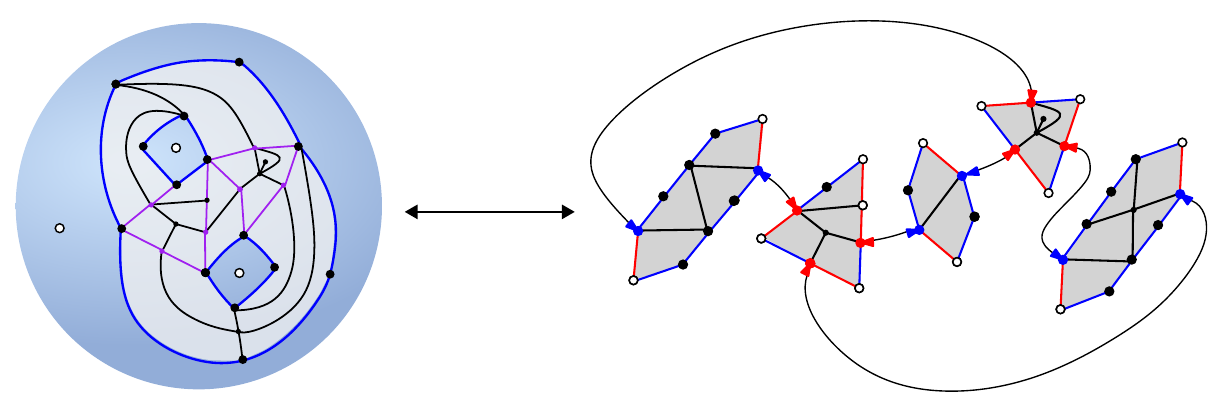}
    \caption{An example of a tight pair of pants, i.e.\ a planar map with three boundaries (in blue) of minimal length in their homotopy class in the triply punctured sphere. On the right the canonical decomposition described in \cite{Bouttier_Bijective_2022}. Figure adapted from \cite[Fig.~1.1]{Bouttier_Bijective_2022}.\label{fig:tightmap}}
\end{figure}

Our result \eqref{eq:T03} for the genus-$0$ tight Weil--Petersson volumes with three distinguished boundaries can be seen as the analogue of \cite[Theorem~1.1]{Bouttier_Bijective_2022}, although less powerful because our proof is not bijective.
Instead, we derive generating functions of tight Weil--Petersson volumes from known expressions in the case of ordinary Weil--Petersson volumes.
The general idea is that a genus-$0$ hyperbolic surface with two distinguished (but not necessarily tight) boundaries can unambiguously be cut along a shortest geodesic separating those two boundaries, resulting in a pair of certain half-tight cylinders (Figure~\ref{fig:cylinderdecomposition}).
Also a genus-$g$ surface with $n$ distinguished (not necessarily tight) boundaries can be shown to decompose into a tight hyperbolic surface and $n$ half-tight cylinders.
The first decomposition uniquely determines the Weil--Petersson volumes of the moduli spaces of half-tight cylinders, while the second determines the tight Weil--Petersson volumes.
This relation is at the basis of Proposition~\ref{prop:Tpolynomial}. 

To arrive at the recursion formula of Theorem~\ref{the:tight_top_recursion} we follow the line or reasoning of Mirzakhani's proof \cite{Mirzakhani2007a} of Witten's conjecture \cite{Witten_Two_1991} (proved first by Kontsevich \cite{Kontsevich1992}).
She observes that the recursion equation \eqref{eq:Mirzakhanirecursion} implies that the generating function of certain intersection numbers satisfies an infinite family of partial differential equations, the Virasoro constraints. 
Mulase \& Safnuk \cite{mulase2006mirzakhanis} have observed that the reverse implication is true as well.
We will demonstrate that the generating functions of tight Weil--Petersson volumes and ordinary Weil--Petersson volumes are related in a simple fashion when expressed in terms of the times \eqref{eq:times} and that the former obey a modified family of Virasoro constraints. 
These constraints in turn are equivalent to the generalized recursion of Theorem~\ref{the:tight_top_recursion}.

\subsection{Discussion}

Mirzakhani's recursion formula has a bijective interpretation \cite{Mirzakhani2007}.
Upon multiplication by $2L_1$ the left-hand side $2L_1 V_{g,n}(\mathbf{L})$ accounts for the volume of surfaces with a marked point on the first boundary.
Tracing a geodesic ray from this point, perpendicularly to the boundary, until it self-intersects or hits another boundary allows one to canonically decompose the surface into a hyperbolic pair of pants (3-holed sphere) and one or two smaller hyperbolic surfaces.
The terms on the right-hand side of \cite{Mirzakhani2007} precisely take into account the Weil--Petersson volumes associated to these parts and the way they are glued.

It is natural to expect that Theorem~\ref{the:tight_top_recursion} admits a similar bijective interpretation, in which the surface decomposes into a tight pair of pants (a sphere with $3+p$ boundaries, three of which are tight) and one or two smaller tight hyperbolic surfaces.
However, Mirzakhani's ray shooting procedure does not generalize in an obvious way.
Nevertheless, working under the assumption that a bijective decomposition exists, one is led to suspect that the generalized kernel $K(x,t,\mu]$ of Theorem~\ref{the:tight_top_recursion} contains important information about the geometry of tight pairs of pants.
Moreover, one would hope that this geometry can be further understood via a decomposition of the tight pairs of pants themselves analogous to the planar map case of \cite{Bouttier_Bijective_2022} described above.

Since a genus-$0$ surface $\mathsf{X}\in\mathcal{M}_{0,3+p}(0,0,0,\mathbf{L})$ with three distinguished cusps is always a tight pair of pants (since the zero length boundaries are obviously minimal), a consequence of a bijective interpretation of Theorem~\ref{the:tight_top_recursion} is a conjectural interpretation of the series $\hat{X}(u;\mu] = \sin(2\pi u)/(2\pi u\eta)$ in \eqref{eq:Xhat} in terms of the hyperbolic distances between the three cusps.
To be precise, let $c_1,c_2,c_3$ be unit-length horocycles around the three cusps and $\Delta(\mathsf{X}) = d_{\mathrm{hyp}}(c_1,c_2)-d_{\mathrm{hyp}}(c_1,c_3)$ the difference in hyperbolic distance between two pairs, then it is plausible that
\begin{align}
    \hat{X}(u;\mu] \stackrel{?}{=} \sum_{p\geq 0} \frac{1}{p!}\int \left(\int_{\mathcal{M}_{0,3+p}(0,0,0,\mathbf{L})} e^{2u \Delta(\mathsf{X})} \rmd \mu_{\mathrm{WP}}(\mathsf{X}) \right) \rmd \mu(L_1)\cdots\rmd \mu(L_p).
\end{align}
Or in the probabilistic terms of Section~\ref{sec:probintermezzo}, the measure $M_0[\mu] X(z;\mu]$ on $\mathbb{R}$, which integrates to $1$ due to \eqref{eq:T03}, is the probability distribution of the random variable $2\Delta(\mathsf{X})$ in a genus-$0$ Boltzmann hyperbolic surface $\mathsf{X}$ with weight $\mu$.
In upcoming work we shall address this conjecture using very different methods.

Another natural question to ask is whether the generalization of the spectral curve \eqref{eq:WPcurve} of Weil--Petersson volumes to the one of tight Weil--Petersson volumes in Theorem~\ref{the:tight_spectral_curve} can be understood in the general framework of deformations of spectral curves in topological recursion \cite{Eynard_Invariants_2007}.

\subsection{Outline}
The structure of the paper is as follows:

In section~\ref{sec:decomp} we introduce the half-tight cylinder, which allows us to do tight decomposition of surfaces, which relates the regular hyperbolic surfaces to the tight surfaces. Using the decomposition we prove Proposition~\ref{prop:Tpolynomial}. 

In section~\ref{sec:generating_func} consider the generating functions of (tight) Weil--Petersson volumes and their relations. Furthermore, we use the Virasoro constraints to prove Theorem \ref{the:tight_top_recursion}, Theorem~\ref{the:recursion_n} and Corollary~\ref{cor:stringdilaton}.

In section~\ref{sec:Laplace} we take the Laplace transform of the tight Weil--Petersson volumes and prove Theorem~\ref{the:tight_spectral_curve}. We also look at the relation between the disk function of the regular hyperbolic surfaces and the generating series of moment $\eta$.

Finally, in section \ref{sec:JT} we briefly discuss how our results may be of use in the study of JT gravity.

\paragraph{Acknowledgments} This work is supported by the START-UP 2018 programme with project
number 740.018.017 and the VIDI programme with project number VI.Vidi.193.048, which are
financed by the Dutch Research Council (NWO).

\newpage
\section{Decomposition of tight hyperbolic surfaces}\label{sec:decomp}

\subsection{Half-tight cylinder}

Recall that a boundary $\partial_i$ of $X \in \mathcal{M}_{g,n+p}(\mathbf{L})$ is said to be \emph{tight in }$\invbreve{S}_{g,n,p}$ if $\partial_i$ is the only simple cycle $\gamma$ in $[\partial_i]_{\invbreve{S}_{g,n,p}}$ of length $\ell_\gamma(X) \leq L_i$ and we defined the \emph{moduli space of tight hyperbolic surfaces} as 
\begin{align}
	\mathcal{M}^{\text{tight}}_{g,n,p}(\mathbf{L}) &= \{ X \in \mathcal{M}_{g,n+p}(\mathbf{L}) : \partial_i\text{ is tight in }\invbreve{S}_{g,n,p}\} \subset \mathcal{M}_{g,n+p}(\mathbf{L}).
\end{align}
We noted before that when $g=0$ and $n=2$ we have $\mathcal{M}^{\text{tight}}_{0,2,p}(\mathbf{L}) = \emptyset$ because $\partial_1$ and $\partial_2$ belong to the same free homotopy class of $\invbreve{S}_{0,2,p}$ and can therefore never both be the unique shortest cycle.
Instead, it is useful for any $p\geq 1$ to consider the \emph{moduli space of half-tight cylinders}
\begin{align}
	\mathcal{H}_{p}(\mathbf{L}) &= \{ X \in \mathcal{M}_{0,2+p}(\mathbf{L}) : \partial_2\text{ is tight in }\invbreve{S}_{0,2,p}\} \subset \mathcal{M}_{0,2+p}(\mathbf{L}),
\end{align}
which is non-empty whenever $L_1 > L_2 > 0$.
We will also consider 
\begin{align}\label{eq:weaklyhalftight}
	\overline{\mathcal{H}_p}(\mathbf{L}) &= \{ X \in \mathcal{M}_{0,2+p}(\mathbf{L}) : \partial_2\text{ has minimal length in }[\partial_2]_{\invbreve{S}_{0,2,p}}\} \subset \mathcal{M}_{0,2+p}(\mathbf{L})
\end{align}
and denote its Weil--Petersson volume by $H_p(\mathbf{L})$.
By construction, it is non-zero for $L_1 \geq L_2 > 0$ and $H_p(\mathbf{L}) \leq V_{0,2+p}(\mathbf{L})$.
\begin{lemma}
	$\mathcal{H}_p(\mathbf{L})$ is an open subset of $\mathcal{M}_{0,2+p}(\mathbf{L})$, and when it is non-empty ($L_1>L_2$) its closure is $\overline{\mathcal{H}_p}(\mathbf{L})$.
	In particular, both have the same finite Weil--Petersson volume $H_p(\mathbf{L})$ when $L_1 > L_2$, but $\mathcal{H}_p(\mathbf{L})$ has $0$ volume and $\overline{\mathcal{H}_p}(\mathbf{L})$ non-zero volume $H_p(\mathbf{L})$ when $L_1 = L_2$.
\end{lemma}
\begin{proof}
	For $L_1 > L_2$, $\mathcal{H}_{p}(\mathbf{L})$ is the intersection of the open sets $\{ \ell_\gamma(X) > L_2 \}$ indexed by the countable set of free homotopy classes $\gamma$ in $[\partial_2]_{\invbreve{S}_{0,2,p}}$. It is not hard to see that in a neighbourhood of any $X \in \mathcal{H}_{p}(\mathbf{L})$ only finitely many of these are important, so the intersection is open. Its closure is given by the countable intersection of closed sets $\{ \ell_\gamma(X) \geq L_2 \}$, which is precisely $\overline{\mathcal{H}_p}(\mathbf{L})$. 
\end{proof}

\subsection{Tight decomposition}

We are now ready to state the main result of this section.

\begin{proposition}\label{prop:VHTrelation}
	The Weil--Petersson volumes $T_{g,n,p}(\mathbf{L})$ and $H_p(\mathbf{L})$ satisfy
	\begin{align*}
		V_{g,n+p}(\mathbf{L}) &= \sum_{I_0 \sqcup \cdots \sqcup I_n = \{n+1,\ldots,n+p\}} \int T_{g,n,|I_0|}(\mathbf{K},\mathbf{L}_{I_0})\prod_{\substack{1\leq i\leq n\\I_i \neq \emptyset}} H_{|I_i|}(L_i,K_i,\mathbf{L}_{I_i})\,K_i\rmd K_i,\qquad (g\geq 1\text{ or }n\geq 3)\\
		V_{0,2+p}(\mathbf{L}) &= H_{p}(\mathbf{L}) + \sum_{\substack{I_1 \sqcup I_2 = \{3,\ldots,2+p\}\\I_1,I_2\neq\emptyset}} \int_0^{L_2} H_{|I_1|}(L_1,K,\mathbf{L}_{I_1}) H_{|I_2|}(L_2,K,\mathbf{L}_{I_2}) K\rmd K, \qquad (L_1 \geq L_2)
	\end{align*}
	where in the first equation it is understood that $K_i = L_i$ whenever $I_i = \emptyset$.
\end{proposition}

The remainder of this section will be devoted to proving this result.
But let us first see how it implies Proposition~\ref{prop:Tpolynomial}.

\begin{proof}[Proof of Proposition~\ref{prop:Tpolynomial}]
Clearly $H_1(\mathbf{L}) = V_{0,3}(\mathbf{L}) = 1$ for $L_1 \geq L_2$ and $T_{g,n,0}(\mathbf{L}) = V_{g,n}(\mathbf{L})$.
Rewriting the equations as 
\begin{align*}
	T_{g,n,p}(\mathbf{L}) &= V_{g,n+p}(\mathbf{L}) - \sum_{\substack{I_0 \sqcup \cdots \sqcup I_n = \{n+1,\ldots,n+p\}\\|I_0| < p}} \int T_{g,n,|I_0|}(\mathbf{K},\mathbf{L}_{I_0})\prod_{\substack{1\leq i\leq n\\I_i \neq \emptyset}} H_{|I_i|}(L_i,K_i,\mathbf{L}_{I_i})\,K_i\rmd K_i,\\
	H_{p}(\mathbf{L}) &= V_{0,2+p}(\mathbf{L}) - \sum_{\substack{I_1 \sqcup I_2 = \{3,\ldots,2+p\}\\I_1,I_2\neq\emptyset}} \int_0^{L_2} H_{|I_1|}(L_1,K,\mathbf{L}_{I_1}) H_{|I_2|}(L_2,K,\mathbf{L}_{I_2}) K\rmd K,
\end{align*}
it is clear that they are uniquely determined recursively in terms of $V_{g,n}(\mathbf{L})$.
Moreover, by induction we easily verify that $H_p(\mathbf{L})$ in the region $L_1 \geq L_2$ is a polynomial in $L_1^2,\ldots,L_{2+p}^2$ of degree $p-1$ that is symmetric in $L_3,\ldots,L_{2+p}$, and $T_{g,n,p}$ is a polynomial in $L_1^2,\ldots,L_{n+p}^2$ of degree $3g-3+n+p$ that is symmetric in $L_1,\ldots,L_n$ and symmetric in $L_{n+1},\ldots,L_{n+p}$.
\end{proof}

\subsection{Tight decomposition in the stable case}

\subsubsection{Shortest cycles}

The following parallels the construction of shortest cycles in maps described in \cite[Section~6.1]{Bouttier_Bijective_2022}.

\begin{lemma}\label{lem:innermostshortestcycles}
	Given a hyperbolic surface $X \in \mathcal{M}_{g,n+p}$ for $g\geq 1$ or $n\geq 2$, then for each $i=1,\ldots,n$ there exists a unique \emph{innermost shortest cycle} $\sigma^i_{\invbreve{S}_{g,n,p}}(X)$ on $X$, meaning that it has minimal length in $[\partial_i]_{\invbreve{S}_{g,n,p}}$ and such that all other cycles of minimal length (if they exist) are contained in the region of $X$ delimited by $\partial_i$ and $\sigma^i_{\invbreve{S}_{g,n,p}}(X)$. 
	Moreover, if $g\geq 1$ or $n \geq 3$, the curves $\sigma^1_{\invbreve{S}_{g,n,p}}(X), \ldots,\sigma^n_{\invbreve{S}_{g,n,p}}(X)$ are disjoint.
\end{lemma}
\begin{proof}
	First note that if a shortest cycle exists, it is a simple closed geodesic. 
    As a consequence of \cite[Theorem 1.6.11]{buser1992geometry}, there are only finitely many closed geodesics with length $\leq L_i$ in $[\partial_i]_{\invbreve{S}_{g,n,p}}$.
    Since $\partial_i\in [\partial_i]_{\invbreve{S}_{g,n,p}}$ has length $L_i$, this proves the existence of at least one cycle in $[\partial_i]_{\invbreve{S}_{g,n,p}}$ with minimal length.
	
    \begin{figure}
        \centering
        \includegraphics[height=.35\linewidth]{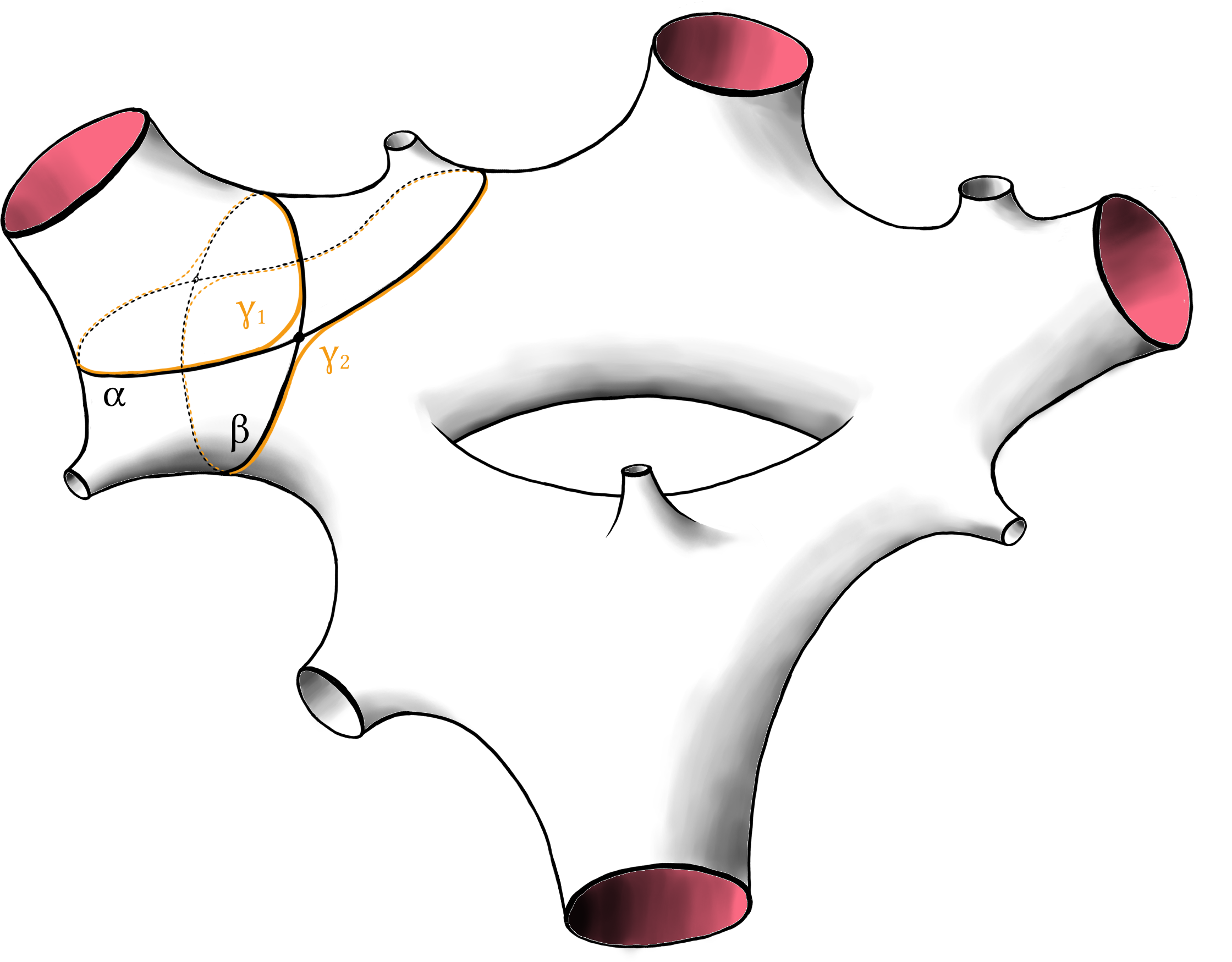}%
        \includegraphics[height=.35\linewidth]{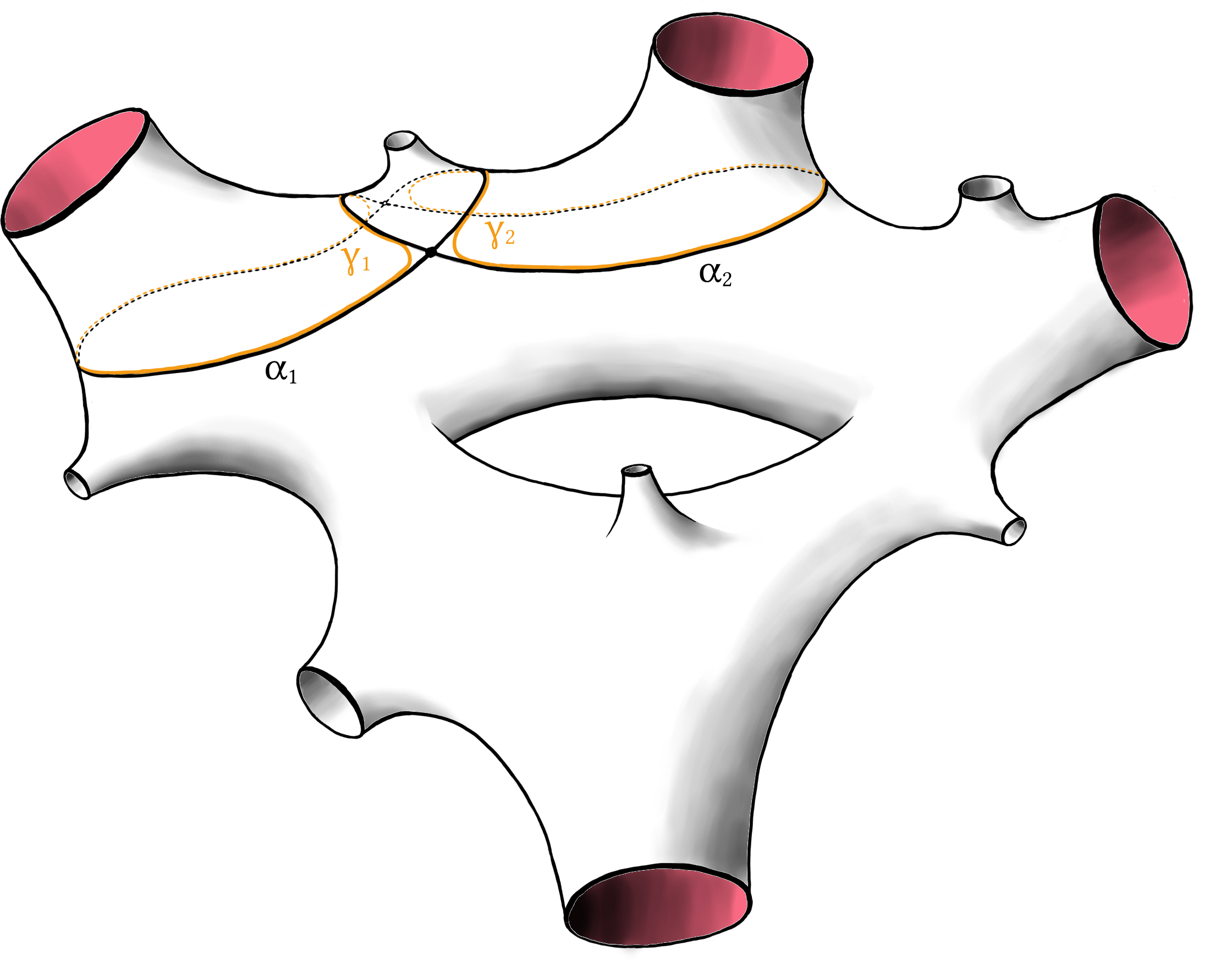}
        \caption{Illustrations of why two intersecting geodesics $\alpha$ and $\beta$ cannot both have minimal length in their free-homotopy classes on $\invbreve{S}_{g,n,p}$. \label{fig:intersectinggeodesics}}
    \end{figure}

	Regarding the existence and uniqueness of a well-defined innermost shortest cycle, suppose $\alpha,\beta \in [\partial_i]_{\invbreve{S}_{g,n,p}}$ are two distinct simple closed geodesics with minimal length $\ell$ (see left side of Figure~\ref{fig:intersectinggeodesics}).
    Since $\alpha \in [\partial_i]_{\invbreve{S}_{g,n,p}}$, cutting along $\alpha$ separates the surface in two disjoint parts.
    Therefore, $\alpha$ and $\beta$ can only have an even number of intersections. 
    If the number of intersections is greater than zero, we can choose two distinct intersections and combine $\alpha$ and $\beta$ to get two distinct cycles $\gamma_1$ and $\gamma_2$ by switching between $\alpha$ and $\beta$ at the chosen intersections, such that $\gamma_1$ and $\gamma_2$ are still in $[\partial_i]_{\invbreve{S}_{g,n,p}}$. 
    Since the total length is still $2\ell$, at least one of the new cycles has length $\leq \ell$. 
    This cycle is not geodesic, so there will be a closed cycle in $[\partial_i]_{\invbreve{S}_{g,n,p}}$ with length $<\ell$, which contradicts that $\alpha$ and $\beta$ have minimal length. 
    We conclude that $\alpha$ and $\beta$ are disjoint. 
    Since all cycles in $[\partial_i]_{\invbreve{S}_{g,n,p}}$ with minimal length are disjoint and separating, the notion of being innermost is well-defined.
	
	Consider $\alpha_i=\sigma^i_{\invbreve{S}_{g,n,p}}(X)$ and $\alpha_j=\sigma^j_{\invbreve{S}_{g,n,p}}(X)$ for $i\neq j$ (see right side of Figure~\ref{fig:intersectinggeodesics}). Just as before, since $\alpha_i$ is separating and $\alpha_i$ and $\alpha_j$ are simple, the number of intersections is even. If $\alpha_i$ and $\alpha_j$ are not disjoint, we can choose two distinct intersections and construct two distinct cycles $\gamma_i$ and $\gamma_j$ by switching between $\alpha_i$ and $\alpha_j$ at the chosen intersections, such that $\gamma_i$ and $\gamma_j$ are in $[\partial_i]_{\invbreve{S}_{g,n,p}}$ and $[\partial_j]_{\invbreve{S}_{g,n,p}}$ respectively. Since the total length of the cycles stays the same, there is at least one $a\in \{i,j\}$ such that $\gamma_a$ has length less or equal than $\alpha_a$. Since $\gamma_a$ is not geodesic, there is a closed cycle in $[\partial_a]_{\invbreve{S}_{g,n,p}}$ with length strictly smaller than $\alpha_a$, which is a contradiction, so the innermost shortest cycles are disjoint.
\end{proof}

In particular the proof implies the following criterions are equivalent: 
\begin{itemize}
    \item A simple closed geodesic $\alpha \in [\partial_i]_{\invbreve{S}_{g,n,p}}$ is the innermost shortest cycle $\sigma^i_{\invbreve{S}_{g,n,p}}(X)$;
    \item For a simple closed geodesic $\alpha \in [\partial_i]_{\invbreve{S}_{g,n,p}}$ we have $\ell(\alpha) \leq L_i$ and each simple closed geodesic $\beta\in\sigma^i_{\invbreve{S}_{g,n,p}}(X)$ that is disjoint from $\alpha$ has length $\ell(\beta)\geq \ell(\alpha)$ with equality only being allowed if $\beta$ is contained in the region between $\alpha$ and $\partial_i$.
\end{itemize}

\subsubsection{Integration on Moduli space}
Let us recap Mirzakhani's decomposition of moduli space integrals in the presence of distinguished cycles \cite[Section 8]{Mirzakhani2007}.
A \emph{multicurve} $\Gamma = (\gamma_1,\ldots,\gamma_k)$ is a collection of disjoint simple closed curves $\Gamma = (\gamma_1,\ldots,\gamma_k)$ in $S_{g,n}$ which are pairwise non-freely-homotopic.
Given a multicurve, in which each curve $\gamma_i$ may or may not be freely homotopic to a boundary $\partial_j$ of $S_{g,n}$, one can consider the stabilizer subgroup
\begin{align*}
	\operatorname{Stab}(\Gamma) = \{ h\in \mathrm{Mod}_{g,n} : h \cdot\gamma_i = \gamma_i \} \subset \mathrm{Mod}_{g,n}.
\end{align*}
Note that if $\gamma_i \in [\partial_j]_{S_{g,n}}$ is freely homotopic to one of the boundaries $\partial_j$ then $h \cdot\gamma_i = \gamma_i$ for any $h\in \mathrm{Mod}_{g,n}$.
The moduli space of hyperbolic surfaces with distinguished (free homotopy classes of) curves is the quotient
\begin{align*}
	\mathcal{M}_{g,n}(\mathbf{L})^\Gamma = \mathcal{T}_{g,n}(\mathbf{L})/\operatorname{Stab}(\Gamma).
\end{align*}
For a closed curve $\gamma$ in $S_{g,n}$ and $X\in \mathcal{M}_{g,n}$, let $\ell_\gamma(X)$ be the length of the geodesic representative in the free homotopy class of $\gamma$. 
For $\mathbf{K} = (K_1,\ldots,K_k) \subset \R_{>0}^k$ we can restrict the lengths of the geodesic representatives of curves in $\Gamma$ by setting 
\begin{align*}
	\mathcal{M}_{g,n}(\mathbf{L})^\Gamma(\mathbf{K}) = \{ X \in \mathcal{M}_{g,n}(\mathbf{L})^\Gamma : \ell_{\gamma_i}(X) = K_i, i=1,\ldots,k\} \subset \mathcal{M}_{g,n}(\mathbf{L})^\Gamma.
\end{align*}
If $\gamma_i \in [\partial_j]_{S_{g,n}}$ then this set is empty unless $K_i = L_j$.
Denote by $\pi^\Gamma : \mathcal{M}_{g,n}(\mathbf{L})^\Gamma \to \mathcal{M}_{g,n}(\mathbf{L})$ the projection.
If there are exactly $p$ cycles among $\Gamma$ that are not freely homotopic to a boundary, then this space admits a natural action of the $p$-dimensional torus $(S^1)^p$ obtained by twisting along each of these $p$ cycles proportional to their length.
The quotient space is denoted
\begin{align*}
	\mathcal{M}_{g,n}(\mathbf{L})^{\Gamma*}(\mathbf{K}) = \mathcal{M}_{g,n}(\mathbf{L})^\Gamma(\mathbf{K}) / (S^1)^p
\end{align*}
and is naturally equipped with a symplectic structure inherited from the Weil--Petersson symplectic structure on $\mathcal{M}_{g,n}(\mathbf{L})^\Gamma$.
If we denote by $S_{g,n}(\Gamma)$ the possibly disconnected surface obtained from $S_{g,n}$ by cutting along all $\gamma_i$ that are not freely homotopic to a boundary and by $\mathcal{M}(S_{g,n}(\Gamma),\mathbf{L},\mathbf{K})$ its moduli space, then according to \cite[Lemma 8.3]{Mirzakhani2007}, the canonical mapping
\begin{align}\label{eq:symplectomorphism}
	\mathcal{M}_{g,n}(\mathbf{L})^{\Gamma*}(\mathbf{K}) \to \mathcal{M}(S_{g,n}(\Gamma),\mathbf{L},\mathbf{K})
\end{align}
is a symplectomorphism. 
Given an integrable function $F : \mathcal{M}_{g,n}(\mathbf{L})^\Gamma \to \R$ that is invariant under the action of $(S^1)^p$, there exists a naturally associated function $\tilde{F} : \mathcal{M}(S_{g,n}(\Gamma),\mathbf{L},\mathbf{K}) \to \R$ such that (essentially \cite[Lemma 8.4]{Mirzakhani2007})
\begin{align}\label{eq:integrationfactorization}
	\int_{\mathcal{M}_{g,n}(\mathbf{L})^\Gamma} F(X) \rmd\mu_{WP}(X) = \int \prod_{\substack{1\leq i\leq n\\\gamma_i \notin [\partial_j]_{S_{g,n}}}} K_i \rmd K_i \int_{\mathcal{M}(S_{g,n}(\Gamma),\mathbf{L},\mathbf{K})}\tilde{F}(X) \rmd\mu_{WP}(X).
\end{align}

\subsubsection{Shortest multicurves} 
Suppose $g\geq 1$ or $n\geq 3$, meaning that we momentarily exclude the cylinder case ($g=0$, $n=2$).
We consider now a special family of multicurves $\Gamma = (\gamma_1,\ldots,\gamma_n)$ on $S_{g,n+p}$ for $n\geq 1, p\geq0$. 
Namely, we require that $\gamma_i \in [\partial_i]_{\invbreve{S}_{g,n,p}}$ is freely homotopic to the boundary $\partial_i$ in the capped-off surface $\invbreve{S}_{g,n,p}$ for $i=1,\ldots,n$. 
Then there exists a partition $I_0 \sqcup \cdots \sqcup I_n = \{n+1,\ldots,n+p\}$ such that $S_{g,n+p}(\Gamma)$ has $n+1$ connected components $s_0,\ldots,s_n$, where $s_0$ is of genus $g$ and is adjacent to all curves $\Gamma$ as well as the boundaries $(\partial_j)_{j\in I_0}$ while for each $i=1, \ldots,n$, $s_i$ is of genus $0$ and contains the $i$th boundary $\partial_i$ as well as $(\partial_j)_{j\in I_i}$ and is adjacent to $\gamma_i$.
Note that $I_i = \emptyset$ if and only if $\gamma_i \in [\partial_i]_{S_{g,n+p}}$.
Finally, we observe that mapping class group orbits $\{\text{Mod}_{g,n+p}\cdot [\Gamma]_{S_{g,n+p}}\}$ of these multicurves $\Gamma$ are in bijection with the set of partitions $\{I_0 \sqcup \cdots \sqcup I_n = \{n+1,\ldots,n+p\}\}$. 

With the help of Lemma~\ref{lem:innermostshortestcycles} we may introduce the restricted moduli space in which we require $\gamma_i$ to be (freely homotopic to) the innermost shortest cycle in $[\partial_i]_{\invbreve{S}_{g,n,p}}$,
\begin{align*}
	\hat{\mathcal{M}}_{g,n,p}(\mathbf{L})^\Gamma = \left\{ X \in \mathcal{M}_{g,n+p}(\mathbf{L})^\Gamma : \gamma_i\in [\sigma^i_{\invbreve{S}_{g,n,p}}(X)]\text{ for }i=1,\ldots,n\right\}.
\end{align*}

\begin{lemma}\label{lem:modulilift}
	The natural projection 
	\begin{align*}
		\bigsqcup_{\Gamma} \hat{\mathcal{M}}_{g,n,p}(\mathbf{L})^\Gamma \longrightarrow \mathcal{M}_{g,n+p}(\mathbf{L}),
	\end{align*}
	where the disjoint union is over (representatives of) the mapping class group orbits of multicurves $\Gamma$, is a bijection. 
\end{lemma}
\begin{proof}
	If $X,X'\in \mathcal{T}_{g,n+p}(\mathbf{L})$ are representatives of hyperbolic surfaces in $\hat{\mathcal{M}}_{g,n,p}(\mathbf{L})^\Gamma$ and $\hat{\mathcal{M}}_{g,n,p}(\mathbf{L})^{\Gamma'}$ respectively, then by definition $[\gamma_i] = [\sigma^i_{\invbreve{S}_{g,n,p}}(X)]$ and $[\gamma_i']=[\sigma^i_{\invbreve{S}_{g,n,p}}(X')]$. If $X$ and $X'$ represent the same surface in $\mathcal{M}_{g,n+p}(\mathbf{L})$, they are related by an element $h$ of the mapping class group, $X' = h\cdot X$, and therefore also $\sigma^i_{\invbreve{S}_{g,n,p}}(X) = h\cdot \sigma^i_{\invbreve{S}_{g,n,p}}(X')$ and $[\gamma_i] = h\cdot[\gamma_i']$.
	So $\Gamma$ and $\Gamma'$ belong to the same mapping class group orbit and, if $\Gamma$ and $\Gamma'$ are freely homotopic, we must have $h \in \operatorname{Stab}(\Gamma)$. Hence, $X$ and $X'$ represent the same element in the set on the left-hand side, and we conclude that the projection is injective.
	It is also surjective since any $X\in \mathcal{T}_{g,n+p}(\mathbf{L})$ is a representative of $\hat{\mathcal{M}}_{g,n,p}(\mathbf{L})^\Gamma$ if we take $\Gamma = (\sigma^1_{\invbreve{S}_{g,n,p}}(X),\ldots,\sigma^n_{\invbreve{S}_{g,n,p}}(X))$, which is a valid multicurve due to Lemma~\ref{lem:innermostshortestcycles}.
\end{proof}
We can introduce the length-restricted version $\hat{\mathcal{M}}_{g,n,p}(\mathbf{L})^\Gamma(\mathbf{K}) \subset \mathcal{M}_{g,n+p}(\mathbf{L})^\Gamma(\mathbf{K})$ as before.
\begin{lemma} \label{lem:tightfactorization}
The subset $\hat{\mathcal{M}}_{g,n,p}(\mathbf{L})^\Gamma(\mathbf{K})\subset \mathcal{M}_{g,n+p}(\mathbf{L})^\Gamma(\mathbf{K})$ is invariant under twisting (the torus-action on $\mathcal{M}_{g,n+p}(\mathbf{L})^\Gamma(\mathbf{K})$ described above).
The image of the quotient $\hat{\mathcal{M}}_{g,n,p}(\mathbf{L})^{\Gamma*}(\mathbf{K})$ under the symplectomorphism \eqref{eq:symplectomorphism} is precisely
\begin{align}\label{eq:tightmodulifactorization}
	\mathcal{M}^{\text{tight}}_{g,n,|I_0|}(\mathbf{K},\mathbf{L}_{I_0})\times \prod_{\substack{1\leq i\leq n\\I_i \neq \emptyset}} \overline{\mathcal{H}_{|I_i|}}(L_i,K_i,\mathbf{L}_{I_i}).
\end{align}
\end{lemma}
\begin{proof}
    Let $X \in \mathcal{M}_{g,n+p}(\mathbf{L})^\Gamma(\mathbf{K})$ be a hyperbolic surface with distinguished multicurve $\Gamma$.
    The lengths of the geodesics associated to $\Gamma$ as well as the lengths of the geodesics that are disjoint from those geodesics are invariant under twisting $X$ along $\Gamma$.
    The criterion explained just below Lemma~\ref{lem:innermostshortestcycles} for $\gamma_i$ to be the innermost shortest cycle $\sigma^i_{\invbreve{S}_{g,n,p}}(X)$ is thus also preserved under twisting, showing that the subset $\hat{\mathcal{M}}_{g,n,p}(\mathbf{L})^\Gamma(\mathbf{K})$ is invariant.
	
    Let $X_0 \in \mathcal{M}_{g,n,|I_0|}(\mathbf{K},\mathbf{L}_{I_0})$ and $X_i \in \mathcal{M}_{0,2+|I_i|}(L_i,K_i,\mathbf{L}_{I_i})$ for those $i=1,\ldots, n$ for which $I_i\neq\emptyset$ be the hyperbolic structures on the connected components $s_0,\ldots,s_n$ of $S_{g,n+p}(\Gamma)$ obtained by cutting $X$ along the geodesics associated to $\Gamma$. 
    For each $i=1,\ldots,n$ the criterion for $\gamma_i$ to be the innermost shortest cycle $\sigma^i_{\invbreve{S}_{g,n,p}}(X)$ is equivalent to the following two conditions holding:
    \begin{itemize}
        \item the $i$th boundary of $X_0$ is tight in the capped-off surface associated to $s_0$;
        \item $I_i = \emptyset$ (meaning $\gamma_i = \partial_i$) or $X_i \in \overline{\mathcal{H}_{|I_i|}}(L_i,K_i,\mathbf{L}_{I_i})$ (recall the definition in \eqref{eq:weaklyhalftight}).
    \end{itemize}
    Hence, we have $X \in \hat{\mathcal{M}}_{g,n,p}(\mathbf{L})^\Gamma(\mathbf{K})$ precisely when $X_0 \in \mathcal{M}^{\text{tight}}_{g,n,|I_0|}(\mathbf{K},\mathbf{L}_{I_0})$ and $X_i \in \overline{\mathcal{H}_{|I_i|}}(L_i,K_i,\mathbf{L}_{I_i})$ when $I_i \neq \emptyset$.
    This proves the second statement of the lemma.
\end{proof}

It follows that the Weil--Petersson volume of $\hat{\mathcal{M}}_{g,n,p}(\mathbf{L})^{\Gamma*}(\mathbf{K})$ is equal to the product of the volumes of the spaces appearing in \eqref{eq:tightmodulifactorization}.
Combining with Lemma~\ref{lem:modulilift} and the integration formula \eqref{eq:integrationfactorization} this shows that
\begin{align*}
	V_{g,n+p}(\mathbf{L}) &= \sum_{\Gamma} \int_{\hat{\mathcal{M}}_{g,n,p}(\mathbf{L})^\Gamma} \rmd\mu_{WP}\\
	&= \sum_{I_0 \sqcup \cdots \sqcup I_n = \{n+1,\ldots,n+p\}} \int T_{g,n,|I_0|}(\mathbf{K},\mathbf{L}_{I_0})\prod_{\substack{1\leq i\leq n\\I_i \neq \emptyset}} H_{|I_i|}(L_i,K_i,\mathbf{L}_{I_i})\,K_i\rmd K_i,
\end{align*}
where it is understood that $K_i = L_i$ whenever $I_i = \emptyset$.
This proves the first relation of Proposition~\ref{prop:VHTrelation}.

\subsection{Tight decomposition of the cylinder}

\begin{figure}
    \centering
    \includegraphics[width=.4\linewidth]{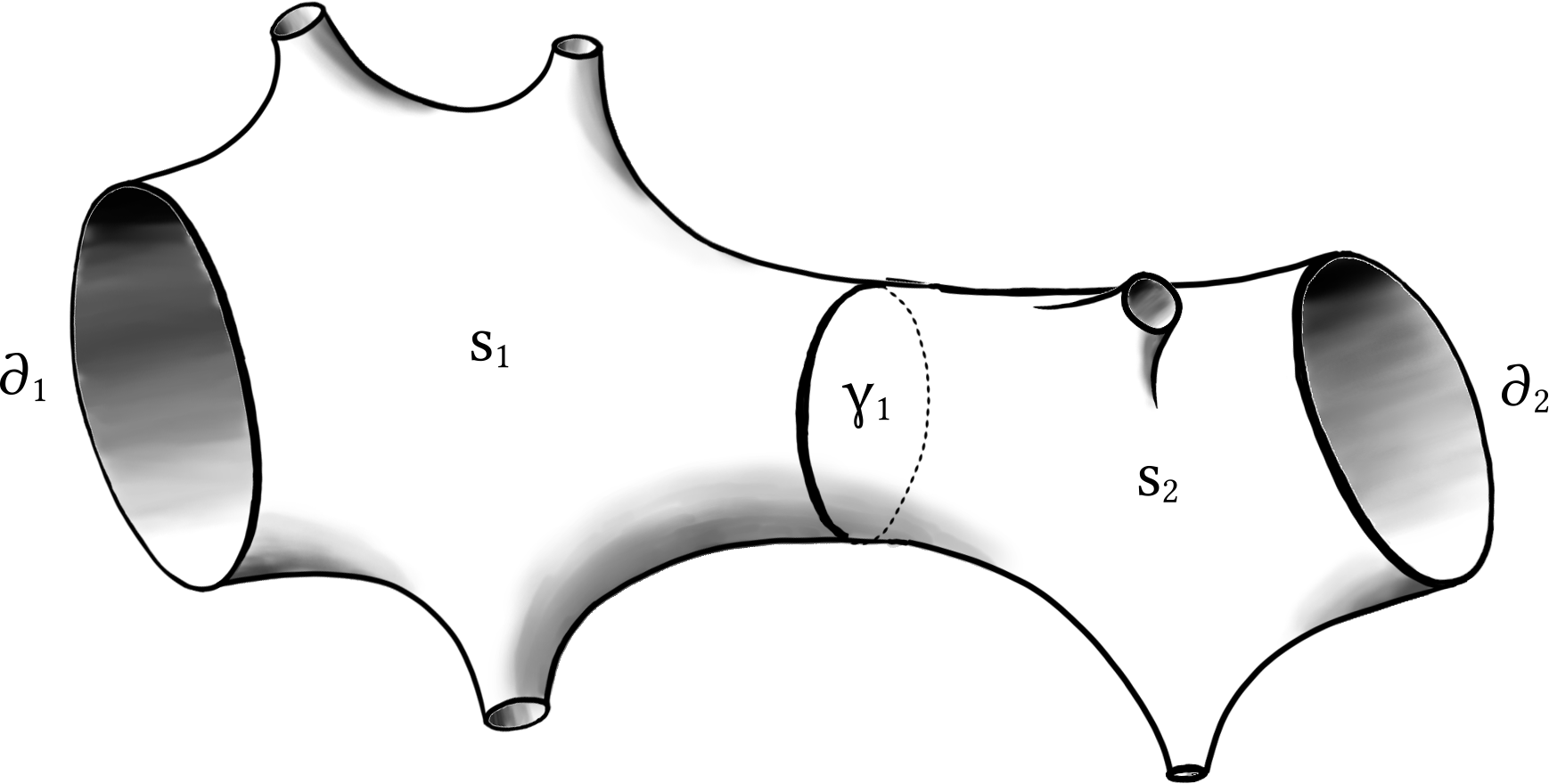}%
    \caption{A surface with two distinguished boundaries $\partial_1$ and $\partial_2$ naturally decomposes into a pair of half-tight cylinders.\label{fig:cylinderdecomposition}}
\end{figure}

The decomposition we have just described does not work well in the case $g=0$ and $n=2$, because $\partial_1$ and $\partial_2$ are in the same free homotopy class of the capped surface $\invbreve{S}_{0,2,p}$.
Instead, we should consider a multicurve $\Gamma=(\gamma_1)$ consisting of a single curve $\gamma_1$ on $S_{0,2+p}$ in the free homotopy class $[\partial_1]_{\invbreve{S}_{0,2,p}}=[\partial_2]_{\invbreve{S}_{0,2,p}}$, see Figure~\ref{fig:cylinderdecomposition}.
In this case there exists a partition $I_1 \sqcup I_2 = \{3,\ldots,p+2\}$ such that $S_{0,2+p}(\Gamma)$ has two connected components $s_1$ and $s_2$, with $s_i$ a genus-$0$ surface with $2+|I_i|$ boundaries corresponding to $\partial_i$, $\gamma_1$ and  $(\partial_{j})_{j\in I_i}$.
We consider the restricted moduli space
\begin{align*}
	\hat{\mathcal{M}}_{0,2,p}(\mathbf{L})^\Gamma = \left\{ X \in \mathcal{M}_{0,2+p}(\mathbf{L})^\Gamma : \gamma_1\in [\sigma^1_{\invbreve{S}_{0,2,p}}(X)]\right\},
\end{align*} 
which thus treats the two boundaries $\partial_1$ and $\partial_2$ asymmetrically, by requiring that $\gamma_1$ is the shortest curve farthest from $\partial_1$.
Lemma~\ref{lem:modulilift} goes through unchanged: the projection
\begin{align*}
	\bigsqcup_{\Gamma} \hat{\mathcal{M}}_{0,2,p}(\mathbf{L})^\Gamma \longrightarrow \mathcal{M}_{0,2+p}(\mathbf{L}),
\end{align*}
where the disjoint union is over the mapping class group orbits of $\Gamma = (\gamma_1)$, is a bijection.
Assuming $L_1 \geq L_2$, we cannot have $\gamma_1 \in [\partial_1]_{S_{0,2+p}}$ so $I_1 \neq \emptyset$. 
There are two cases to consider:
\begin{itemize}
	\item $\gamma_1 \in [\partial_2]_{S_{0,2+p}}$ and therefore $I_2 = \emptyset$: this means that $\partial_2$ has minimal length in $[\partial_2]_{\invbreve{S}_{0,2,p}}$, so $\hat{\mathcal{M}}_{0,2,p}(\mathbf{L})^\Gamma = \overline{\mathcal{H}_p}(\mathbf{L})$.
	\item $I_2 \neq \emptyset$: by reasoning analogous to that of Lemma~\ref{lem:tightfactorization} we have that $\hat{\mathcal{M}}_{0,2,p}(\mathbf{L})^{\Gamma*}(K)$ is symplectomorphic to 
	\begin{align*}
		\overline{\mathcal{H}_{|I_1|}}(L_1,K,\mathbf{L}_{I_1}) \times \mathcal{H}_{|I_2|}(L_2,K,\mathbf{L}_{I_2}).
	\end{align*}
\end{itemize}
Hence, when $L_1 \geq L_2$ we have
\begin{align*}
	V_{0,2+p}(\mathbf{L}) &= H_{p}(\mathbf{L}) + \sum_{\substack{I_1 \sqcup I_2 = \{3,\ldots,2+p\}\\I_1,I_2\neq\emptyset}} \int_0^{L_2} H_{|I_1|}(L_1,K,\mathbf{L}_{I_1}) H_{|I_2|}(L_2,K,\mathbf{L}_{I_2}) K\rmd K.
\end{align*}
This proves the second relation of Proposition~\ref{prop:VHTrelation}.

\section{Generating functions of tight Weil--Petersson volumes}\label{sec:generating_func}
\subsection{Definitions}

Let us define the following generating functions of the Weil--Petersson volumes, half-tight cylinder volumes and tight Weil--Petersson volumes: 
\begin{align}
	F_g[\mu] &= \sum_{n=0}^\infty \frac{1}{n!} \int \dd{\mu(L_1)}\cdots \int\dd{\mu(L_{n})} V_{g,n}(\mathbf{L}),\\
	H(L_1,L_2;\mu] &= \sum_{p=1}^\infty \frac{1}{p!} \int \dd{\mu(L_3)}\cdots \int \dd{\mu(L_{2+p})} H_p(\mathbf{L}),\label{eq:htcylindergenfun}\\
	\tilde F_g[\nu,\mu]&=\sum_{n=0}^\infty \frac{1}{n!} \int \dd{\nu(L_1)}\cdots \int \dd{\nu(L_n)} T_{g,n}(L_1,\ldots,L_n;\mu].\label{eq:bivariategenfun}
\end{align}

Furthermore, for $g\geq 2$ we recall the polynomial
\begin{equation}\label{eq:Pg0expr}
\mathcal{P}_{g,0}(m_1,\ldots,m_{3g-3}) = \sum_{\substack{d_2,d_3,\ldots\geq 0\\ \sum_{k\geq 2} (k-1)d_k = 3g-3}}\!\!\!\! \langle \tau_2^{d_2}\tau_3^{d_3}\cdots \rangle_g \prod_{k\geq 2} \frac{(-m_{k-1})^{d_k}}{d_k!},
\end{equation}
where $\langle \tau_2^{d_2}\tau_3^{d_3}\cdots \rangle_g$ are the $\psi$-class intersection numbers on the moduli space $\mathcal{M}_{g,n}$ with $n = \sum_k d_k \leq 3g-3$ marked points.
Then according to \cite[Theorem~3]{budd2020irreducible}\footnote{Note that there has been a shift in conventions, e.g. regarding factors of $2$. }
\begin{align}
F_0[\mu] &= \frac{1}{2}\int_{0}^{R} \!\!\!\rmd r\, Z(r;\mu]^2,\label{eq:F0expr}\\
F_1[\mu] &= - \frac{1}{24} \log M_0[\mu],\\
F_g[\mu] &= \frac{1}{\left(M_0[\mu]\right)^{2g-2}}\,\mathcal{P}_{g,0}\left(\frac{M_1[\mu]}{M_0[\mu]},\ldots, \frac{M_{3g-3}[\mu]}{M_0[\mu]}\right) \quad \text{for }g\geq 2.\label{eq:Fgexpr}
\end{align}
In the genus-$0$ case we can take successive derivatives to find useful formulas for one, two or three distinguished boundaries of prescribed lengths,
\begin{align}
\fdv{F_0[\mu]}{\mu(L_1)} &= -\int_0^{R[\mu]} I_0(L_1\sqrt{2r})\,Z(r;\mu]\rmd r,  \label{eq:diskfunction},\\
\frac{\delta^2F_0[\mu]}{\delta\mu(L_1)\delta\mu(L_2)} &=  \int_0^{R[\mu]} I_0(L_1\sqrt{2r})\, I_0(L_2\sqrt{2r})\rmd r,\label{eq:cylinderfunction}\\
\frac{\delta^3 F_0[\mu]}{\delta\mu(L_1)\delta\mu(L_2)\delta\mu(L_3)} &= \frac{1}{M_0[\mu]}\left[\prod_{i=1}^3 I_0(L_i\sqrt{2R[\mu]})\right].
\end{align}

\subsection{Volume of half-tight cylinder}

The equations of Proposition~\ref{prop:VHTrelation} turn into the equations
\begin{align}
	\frac{\delta^n F_g[\mu]}{\delta \mu(L_1)\cdots\delta \mu(L_n)} &= \int T_{g,n}(\mathbf{K};\mu]\prod_{i=1}^n (K_i H(L_i,K_i;\mu] + \delta(K_i - L_i)) \rmd K_i,\qquad  (g\geq 1\text{ or }n\geq 3)\label{eq:FTrelation}\\
	\frac{\delta^2 F_0[\mu]}{\delta \mu(L_1)\delta \mu(L_2)}&= H(L_1,L_2;\mu] + \int_0^{L_2} H(L_1,K;\mu] H(L_2,K;\mu] K\rmd K. \qquad (L_1 \geq L_2)\label{eq:Hgenfuneq}
\end{align}
Let us focus on the last equation, which should determine $H(L_1,L_2;\mu]$ uniquely. 
The left-hand side depends on $\mu$ only through the quantity $R[\mu]$, and the dependence on $R$ is analytic,
\begin{align*}
    \frac{\delta^2 F_0[\mu]}{\delta \mu(L_1)\delta \mu(L_2)} = R + \frac{1}{4}(L_1^2+L_2^2)R^2 + \frac{1}{48}(L_1^4 + 4 L_1^2L_2^2 + L_2^4) R^3 + \cdots.
\end{align*} 
Hence, the same is true for $H(L_1,L_2;\mu]$ and one may easily calculate order by order in $R$ that
\begin{align*}
    H(L_1,L_2;\mu] = R + \frac{1}{4}(L_1^2-L_2^2)R^2 + \frac{1}{48}(L_1^2-L_2^2)^2R^3+\cdots.
\end{align*} 
This suggests that $H(L_1,L_2;\mu]$ depends on $L_1$ and $L_2$ only through the combination $L_1^2 - L_2^2$.
Let's prove this. 

\begin{lemma}\label{lem:halftight}
    The half-tight cylinder generating function satisfies
    \begin{align}
        \left(\frac{1}{L_1}\frac{\partial}{\partial L_1}+\frac{1}{L_2}\frac{\partial}{\partial L_2}\right) H(L_1,L_2;\mu] = 0,\label{eq:Hdifferentialidentity}
    \end{align}
    and is therefore given by
    \begin{align}
        H(L_1,L_2;\mu] = \sum_{\ell=0}^\infty \frac{2^{-\ell} R[\mu]^{\ell+1}}{\ell!(\ell+1)!}(L_1^2-L_2^2)^{\ell} = \sqrt{\frac{2R[\mu]}{L_1^2-L_2^2}} I_1\left( \sqrt{L_1^2-L_2^2}\sqrt{2R[\mu]}\right) \qquad (L_1 \geq L_2).\label{eq:Hexpression}
    \end{align}
\end{lemma}
\begin{proof}
    By construction $H(L_1,0;\mu] = \frac{\delta^2F_0[\mu]}{\delta\mu(L_1)\delta\mu(0)}$ and the integral \eqref{eq:cylinderfunction} with $L_2=0$ evaluates to
    \begin{align}
        H(L_1,0;\mu] = \frac{\delta^2F_0[\mu]}{\delta\mu(L_1)\delta\mu(0)} &= \frac{\sqrt{2R}}{L_1} I_1(L_1\sqrt{2R}).\label{eq:Hat0}
    \end{align}
    The identity
    \begin{align*}
        \frac{\partial}{\partial r}\left( \frac{\sqrt{2r}}{L_1} I_1(L_1\sqrt{2r})\frac{\sqrt{2r}}{L_2} I_1(L_2\sqrt{2r})\right) = \left(\frac{1}{L_1}\frac{\partial}{\partial L_1}+\frac{1}{L_2}\frac{\partial}{\partial L_2}\right) I_0(L_1\sqrt{2r})\, I_0(L_2\sqrt{2r}),
    \end{align*}
    which can be easily checked by calculating the derivatives, implies that
    \begin{align*}
        H(L_1,0;\mu] H(L_2,0;\mu] = \left(\frac{1}{L_1}\frac{\partial}{\partial L_1}+\frac{1}{L_2}\frac{\partial}{\partial L_2}\right) \frac{\delta^2F_0[\mu]}{\delta\mu(L_1)\delta\mu(L_2)}. 
    \end{align*}
    Hence, by \eqref{eq:Hgenfuneq} we find that
    \begin{align*}
        \left(\frac{1}{L_1}\frac{\partial}{\partial L_1}+\frac{1}{L_2}\frac{\partial}{\partial L_2}\right) H(L_1,L_2;\mu] &= H(L_1,0;\mu] H(L_2,0;\mu] - \left(\frac{1}{L_1}\frac{\partial}{\partial L_1}+\frac{1}{L_2}\frac{\partial}{\partial L_2}\right) \int_0^{L_2} H(L_1,K;\mu] H(L_2,K;\mu] K\rmd K\\
        &= H(L_1,0;\mu] H(L_2,0;\mu] - H(L_1,L_2;\mu] H(L_2,L_2;\mu] \\
        &\quad - \int_0^{L_2} \left(\frac{1}{L_1}\frac{\partial}{\partial L_1}+\frac{1}{L_2}\frac{\partial}{\partial L_2}\right)H(L_1,K;\mu] H(L_2,K;\mu] K\rmd K\\
        &= -\int_0^{L_2} \frac{\partial}{\partial K}\left(H(L_1,K;\mu] H(L_2,K;\mu]\right)\rmd K \\
        &\quad - \int_0^{L_2} \left(\frac{1}{L_1}\frac{\partial}{\partial L_1}+\frac{1}{L_2}\frac{\partial}{\partial L_2}\right)H(L_1,K;\mu] H(L_2,K;\mu] K\rmd K\\
        &= - \int_0^{L_2} \Bigg[H(L_1,K;\mu]\left(\frac{1}{L_2}\frac{\partial}{\partial L_2}+\frac{1}{K}\frac{\partial}{\partial K}\right) H(L_2,K;\mu] \\
        &\qquad\qquad + H(L_2,K;\mu]\left(\frac{1}{L_1}\frac{\partial}{\partial L_1}+\frac{1}{K}\frac{\partial}{\partial K}\right) H(L_1,K;\mu]\Bigg] K\rmd K.
    \end{align*}
    Since the leading coefficient in $R$ of $H(L_1,L_2;\mu]$ satisfies \eqref{eq:Hdifferentialidentity}, it follows that the same is true for the higher-order coefficients in $R$.

    As a consequence of \eqref{eq:Hdifferentialidentity}, $H(L_1,L_2;\mu] = H\left(\sqrt{L_1^2-L_2^2},0;\mu\right]$ and the claimed expression \eqref{eq:Hexpression} follows from \eqref{eq:Hat0}.
\end{proof}

\subsection{Rewriting generating functions}

Since the work of Mirzakhani \cite{Mirzakhani2007a} it is known that the Weil--Petersson volumes $V_{g,n}(\mathbf{L})$ are expressible in terms of intersection numbers as follows.
The compactified moduli space $\overline{\mathcal{M}}_{g,n}$ of genus-$g$ curves with $n$ marked points comes naturally equipped with the Chern classes $\psi_1,\ldots,\psi_n$ associated with its $n$ tautological line bundles, as well as the cohomology class $\kappa_1$ of the Weil--Petersson symplectic structure (up to a factor $2\pi^2$). 
The corresponding intersection numbers are given by the integrals
\begin{align*}
    \left<\kappa_1^m\tau_{d_1}\cdots\tau_{d_n}\right>_{g,n} = \int_{\overline{\mathcal{M}}_{g,n}} \kappa_1^m \psi_1^{d_1}\cdots \psi_n^{d_n},
\end{align*}
where $d_1,\ldots,d_n \geq 0$ and $n = d_1+\cdots d_n + m + 3 - 3g$.
For $g \geq 0$ we denote the generating function of these intersection numbers by 
\begin{align}
    G_g(s;x_0,x_1,\ldots)=\sum_{n\geq 0}\frac{1}{n!}\sum_{\substack{m,d_1,\ldots,d_n \geq 0\\d_1+\cdots+d_n+m=3g-3+n}}\left<\kappa_1^m\tau_{d_1}\cdots\tau_{d_n}\right>_{g,n} \, \frac{s^m}{m!}\,x_{d_1}\cdots x_{d_n}.\label{eq:Ggdef}
\end{align}
We may sum over all genera to arrive at the generating function 
\begin{align}
    G(s;x_0,x_1,\ldots)=\sum_{g=0}^\infty \lambda^{2g-2} G_g(s;x_0,x_1,\ldots).
\end{align}
In order to lighten the notation we do not write the dependence on $\lambda$ explicitly here, which only serves as a formal generating variable.
Note that $\lambda$ is actually redundant for organizing the series, since any monomial appears in at most one of the $G_g$ as can be seen from \eqref{eq:Ggdef}.
Then the generating function of Weil--Petersson volumes can be expressed as
\begin{align}
    \sum_{g=0}^\infty \lambda^{2g-2}2^{3-3g}F_g[\mu]=G(\pi^2;t_0[\mu],t_1[\mu],\ldots),\label{eq:wpintersectionnum}
\end{align}
where the \emph{times} $t_k[\mu]$ are defined by
\begin{align}
    t_k[\mu]=\int \dd{\mu(L)}\,\frac{2L^{2k}}{4^kk!}.
\end{align}
See \cite[Lemma 11]{budd2020irreducible} based on \cite{Mirzakhani2007a}, where one should be careful that some conventions differ by some factors of two compared to the current work.

We will show that the (bivariate) generating function $\tilde{F}_g[\nu,\mu]$ of \emph{tight} Weil--Petersson volumes, defined in \eqref{eq:bivariategenfun}, is also related to the intersection numbers, but with different times. 

\begin{proposition}\label{prop:GB_link_to_tildeT}
The generating function of the volumes $T_{g,n}$ is related to the generating function of intersection numbers via
\begin{align}
    \tilde{F}_g[\nu,\mu] = 2^{3g-3} G_g(0;\tau_0[\nu,\mu],\tau_1[\nu,\mu],\ldots),
\end{align}
where the shifted times $\tau_k[\nu,\mu]$ are defined by
\begin{align}\label{eq:taushiftedtimes}
    \tau_k[\nu,\mu]=t_k[\nu]+\delta_{k,1}-2^{1-k}M_{k-1}[\mu].
\end{align}
\end{proposition}

\noindent
This proposition will be proved in the remainder of this subsection, relying on an appropriate substitution of the weight $\nu$.
To this end, we informally introduce a linear mapping $H_\mu$ on measures on the half line $[0,\infty)$ as follows.
If $\rho$ is a measure on $[0,\infty)$ we let $H_\mu \rho$ be the measure given by
\begin{align}
    \rho +  \left(\int H(L,K;\mu] \dd{\rho(L)}\right) K \dd{K}.
\end{align}
The effect of $H_\mu$ on the times can be computed using the series expansion \eqref{eq:Hexpression},  
\begin{align}
    t_k[H_\mu\rho] 
    &= t_k[\rho] + \frac{2}{4^k k!} \int_0^\infty  \left(\int H(L,K;\mu] \dd{\rho(L)}\right) \,K^{2k+1} \dd{K}\nonumber\\
    &= t_k[\rho] + \sum_{p=1}^\infty \frac{2(2R[\mu])^p}{p!4^{p+k}(p+k)!} \int  \dd{\rho(L)}\,L^{2p+2k} \nonumber\\
    & = \sum_{p=0}^\infty \frac{(2R[\mu])^p}{p!} t_{p+k}[\rho].
\end{align}
We observe that $H_\mu$ acts as an infinite upper-triangular matrix on the times.
This matrix is easily inverted to give
\begin{align}
    t_q[\rho] 
    &= \sum_{p=0}^\infty \frac{(-2R[\mu])^p}{p!} t_{p+q}[H_\mu \rho].\label{eq:inversetimesubs}
\end{align}
This means that knowledge of the generating function $\tilde{F}_g[H_\mu \rho,\mu]$ with substituted weight $H_\mu \rho$ is sufficient to recover the original generating function $\tilde{F}_g[\nu,\mu]$.
Luckily the former is within close reach.

\begin{lemma}\label{lem:halftight_F_rel} 
The generating functions for tight Weil--Petersson volumes and regular Weil--Petersson volumes are related by
\begin{align}
    \tilde F_g[H_\mu \rho,\mu]=F_g[\rho+\mu] - \delta_{g,0}F_{\textrm{corr}}[\rho,\mu],
\end{align}
where the correction term 
\begin{align}
    F_{\textrm{corr}}[\rho,\mu]=\sum_{n=0}^2\frac{1}{n!}\sum_{p=0}^\infty \frac{1}{p!}\int \rmd\rho(L_1)\cdots\rmd\rho(L_n)\,\rmd\mu(L_{n+1})\cdots\rmd\mu(L_{n+p})\,V_{0,n+p}(\mathbf{L}) 
\end{align}
is necessary to subtract the constant, linear and quadratic dependence on $\rho$ in the genus-$0$ case.
\end{lemma}
\begin{proof}
If $g,n\geq 0$ (such that $n\geq 3$ if $g=0$) and $L_1,\ldots,L_n \in [0,\infty)\cup i (0,\pi)$, then Proposition~\ref{prop:VHTrelation} allows us to compute
\begin{align}
    &\sum_{p=0}^\infty \frac{1}{p!} \int \dd{\mu(L_{n+1})}\cdots\dd{\mu(L_{n+p})} V_{g,n+p}(\mathbf{L}) \nonumber\\
    &= \sum_{p=0}^\infty \frac{1}{p!}\int \dd{\mu(L_{n+1})}\cdots\dd{\mu(L_{n+p})} \mkern-20mu
    \sum_{I_0 \sqcup \cdots \sqcup I_n = \{n+1,\ldots,n+p\}} \int T_{g,n,|I_0|}(\mathbf{K},\mathbf{L}_{I_0})\prod_{\substack{1\leq i\leq n\\I_i \neq \emptyset}} H_{|I_i|}(L_i,K_i,\mathbf{L}_{I_i})\,K_i\rmd K_i\nonumber\\
    &= \sum_{p_0,\ldots,p_n=0}^\infty \frac{1}{p_0!\dots p_n!}\int \dd{\mu(L_{n+1})}\cdots\dd{\mu(L_{n+p_0+\cdots+p_n})} T_{g,n,p_0}(\mathbf{K},\mathbf{L}^{(0)})\prod_{\substack{1\leq i\leq n\\p_i \neq 0}} H_{p_i}(L_i,K_i,\mathbf{L}^{(i)})\,K_i\rmd K_i,
\end{align}
where we use the notation $\mathbf{L}^{(j)} = (L_{n+p_0+\cdots+p_{j-1}+1},\,\ldots\,,L_{n+p_0+\cdots+p_{j}})$.
In terms of the tight Weil--Petersson volume generating function \eqref{eq:tightgenfun} and the half-tight cylinder generating function \eqref{eq:htcylindergenfun} this evaluates to
\begin{align}
    \sum_{p=0}^\infty \frac{1}{p!} \int \dd{\mu(L_{n+1})}\cdots\dd{\mu(L_{n+p})} V_{g,n+p}(\mathbf{L})=\sum_{J \subset \{1,\ldots,n\}} \int T_{g,n}(\mathbf{K};\mu] \prod_{i\in J} K_i\,H(L_i,K_i;\mu]\rmd K_i, \label{eq:tightdecomp}
\end{align}
where it is understood that in the argument of $T_{g,n}(\mathbf{K};\mu]$ we take $K_i = L_i$ for $i\notin J$.

Expanding $F_g[\rho+\mu]$ from its definition \eqref{eq:wppartitionfunction} we find 
\begin{align}
    F_g[\rho+\mu] -F_{\textrm{corr}}[\rho,\mu]\delta_{g,0} &= 
    \sum_{n,p=0}^\infty \ind_{g\geq1 \textrm{ or } n\geq3} \frac{1}{n!p!}\int \rmd\rho(L_1)\cdots\rmd\rho(L_n)\,\rmd\mu(L_{n+1})\cdots\rmd\mu(L_{n+p})V_{g,n+p}(\mathbf{L})\nonumber\\
    &= \sum_{n=0}^\infty \frac{1}{n!} \ind_{g\geq1 \textrm{ or } n\geq3} \int \left(\prod_{i=1}^n\dd{\rho(L_i)}\right)\sum_{p=0}^\infty \frac{1}{p!} \int \left(\prod_{i=n+1}^{n+p}\dd{\mu(L_i)}\right)V_{g,n+p}(\mathbf{L})\label{eq:Fgexpansion}
\end{align}
Plugging in \eqref{eq:tightdecomp} and using that $T_{0,n}(\mathbf{K};\mu]=0$ for $n<3$, yields
\begin{align*}
    F_g[\rho+\mu] -F_{\textrm{corr}}[\rho,\mu]\delta_{g,0} &= \sum_{n=0}^\infty \frac{1}{n!} \int \rmd\rho(L_1)\cdots\rmd\rho(L_n)\sum_{J \subset \{1,\ldots,n\}} \int T_{g,n}(\mathbf{K};\mu] \prod_{i\in J} K_i\,H(L_i,K_i;\mu]\rmd K_i \\
    &= \sum_{n=0}^\infty \frac{1}{n!} \int \rmd(H_\mu\rho)(L_1)\cdots\rmd(H_\mu\rho)(L_n)\, T_{g,n}(\mathbf{L};\mu] \\
    &= \tilde{F}_g[H_\mu \rho , \mu]
\end{align*}
as claimed.
\end{proof}

\noindent
Lemma~\ref{lem:halftight_F_rel} and \eqref{eq:wpintersectionnum} together lead to the relation
\begin{align}
    \sum_{g=0}^\infty \lambda^{2g-2}2^{3-3g}\tilde F_g[H_\mu\rho,\mu]=G(\pi^2;t_0[\rho+\mu],  t_1[\rho+\mu],\ldots)-8\lambda^{-2}F_{\textrm{corr}}.
\end{align}
The right-hand side can be specialized, making use of a variety of identities between intersection numbers.
Firstly, a relation between intersection numbers involving $\kappa_1$ and pure $\psi$-class intersection numbers \cite{Witten_Two_1991,Faber_conjectural_1999} leads to the identity \cite{Kaufmann_Higher_1996}
\begin{align}
    G(s;x_0,x_1,\dots)=G(0;x_0,x_1,x_2+\gamma_2(s),\ldots), 
\end{align}
where the shifts are
\begin{align}
    \gamma_k(s)=\frac{(-1)^k}{(k-1)!}s^{k-1}\ind_{k\geq2}.
\end{align}
For us this gives
\begin{align}
    \sum_{g=0}^\infty \lambda^{2g-2}2^{3-3g}\tilde F_g[H_\mu\rho,\mu]=G(0;\mathbf{t}[\rho+\mu]+\boldsymbol{\gamma}(\pi^2))-8\lambda^{-2}F_{\textrm{corr}},\label{eq:Ftildeshift}
 \end{align}
where we use the notation $G(0;\mathbf{x})=G(0;x_0,x_1,x_2,\dots)$.

This can be further refined using Witten's observation \cite{Witten_Two_1991}, proved by Kontsevich \cite{Kontsevich1992}, that $G(0;\mathbf{x})$ satisfies the \emph{string equation}
\begin{align}
    \left(- \pdv{x_0} + \sum_{p=0}^\infty x_{p+1}\pdv{x_p} + \frac{x_0^2}{2\lambda^2}\right) e^{G(0;\mathbf{x})}=0.    \label{eq:stringequation}
\end{align}
Following a computation of Itzykson and Zuber \cite{Itzykson_Combinatorics_1992}, it implies the following identity.

\begin{lemma}\label{lem:G_renorm_x0} 
The solution to the string equation \eqref{eq:stringequation} satisfies a formal power series identity in the parameter $r$,
\begin{align}
    G(0;x_0,x_1,\ldots) &= G\left(0; 
    r + \sum_{k=0}^\infty \frac{(-r)^k}{k!}x_k,
    \sum_{k=0}^\infty \frac{(-r)^k}{k!}x_{k+1},
    \sum_{k=0}^\infty \frac{(-r)^k}{k!}x_{k+2}, \ldots \right)\breakline
    - \frac{1}{2\lambda^2} \int_0^r \dd{s} \qty(s + \sum_{k=0}^\infty \frac{(-s)^k}{k!} x_k)^2 .
\end{align}
\end{lemma}
\begin{proof}
For $x_0,x_1,\ldots$ fixed, let us consider the sequence of functions
\begin{align}
    \mathbf{y}(s) = (y_0(s),y_1(s),\ldots), \qquad y_i(s)=\delta_{i,0}\,s+\sum_{k=0}^\infty \frac{(-s)^k}{k!}x_{k+i},
\end{align}
such that $y_p'(s) = \delta_{p,0} - y_{p+1}(s)$.
The string equation \eqref{eq:stringequation} then implies 
\begin{align}
    \dv{s} G\left(0;\mathbf{y}(s)\right) &= \frac{\partial G}{\partial x_0} (0;\mathbf{y}(s)) - \sum_{p=0}^\infty y_{p+1}(s)\,\frac{\partial G}{\partial x_p}(\mathbf{y}(s))= \frac{y_0(s)^2}{2\lambda^2}.
\end{align}
Integrating from $s=0$ to $s=r$ gives the claimed identity.
\end{proof}

\noindent
Before we can use this lemma, we establish a relation between $\tau_k[H_\mu\rho,\mu]$ and $t_k[\rho+\mu]$.
\begin{lemma} \label{lem:ttilde_to_tau}
We can rewrite 
\begin{align}\label{eq:ttilde_to_tau}
    t_q[\rho+\mu]+\gamma_q(\pi^2)=\delta_{q,0}\,2R[\mu]+\sum_{p=0}^\infty\frac{(-2R[\mu])^p}{p!}\tau_{p+q}[H_\mu\rho,\mu] ,
\end{align}
where the shifted times $\tau_k[\nu,\mu]$ are defined in \eqref{eq:taushiftedtimes}.
\end{lemma}
\begin{proof}
We first relate the moments $M_i[\mu]$ defined in \eqref{eq:momentdef} to the times $t_i[\mu]$. 
Note that $Z(u;\mu]$ defined in \eqref{eq:Zdef} can be expressed in the times as
\begin{align}
    Z(u;\mu]&=u-\sum_{k=0}^\infty \frac{(2u)^k}{2k!}(t_k[\mu]+\gamma_k(\pi^2))
\end{align}
By taking $p$ derivatives with respect to $u$, we get
\begin{align}
    \sum_{k=0}^\infty \frac{(2R[\mu])^k}{k!}(t_{k+p}[\mu]+\gamma_{k+p}(\pi^2))&=
    \begin{dcases}
    2R[\mu] &\qq{if} p=0\\
    1-M_0[\mu] &\qq{if} p=1\\
    -2^{1-p}M_{p-1}[\mu] &\qq{if} p\geq 2
    \end{dcases}.
\end{align}
Just like before in obtaining \eqref{eq:inversetimesubs}, this can be inverted to
\begin{align}
    t_q[\mu]+\gamma_q(\pi^2)&=\delta_{q,1}-\sum_{p=0}^\infty\frac{(-2R[\mu])^p}{p!}2^{1-p-q}M_{p+q-1}[\mu]
\end{align}
The right-hand side of \eqref{eq:ttilde_to_tau} can thus be expressed as
\begin{align}
    \delta_{q,0}\,2R[\mu]+\sum_{p=0}^\infty\frac{(-2R[\mu])^p}{p!}\tau_{p+q}[H_\mu\rho,\mu]&= \delta_{q,1}+\sum_{p=0}^\infty\frac{(-2R[\mu])^p}{p!}t_{p+q}[H_\mu\rho] -\sum_{p=0}^\infty\frac{(-2R[\mu])^p}{p!}2^{1-p-q}M_{p+q-1}[\mu] \nonumber\\
&=t_q[\mu]+\gamma_q(\pi^2)+\sum_{p=0}^\infty \frac{(-2R[\mu])^p}{p!} t_{p+q}[H_\mu\rho].
\end{align}
From \eqref{eq:inversetimesubs} the last term is just $t_q[\rho]$, so we have reproduced the left-hand side of \eqref{eq:ttilde_to_tau}, since $t_q[\rho+\mu] = t_q[\rho]+t_q[\mu]$.
\end{proof}

\noindent
The last two lemmas allow us to express \eqref{eq:Ftildeshift} as
\begin{align}
    \sum_{g=0}^\infty \lambda^{2g-2}2^{3-3g}\tilde F_g[H_\mu\rho,\mu]
    &=G(0;\boldsymbol{\tau}[H_\mu\rho,\mu])\nonumber\\
    &\quad+\frac{1}{2\lambda^2} \int_0^{2R[\mu]} \dd{s} \qty(s + \sum_{k=0}^\infty \frac{(-s)^k}{k!} \tau_k[H_\mu\rho,\mu])^2 -\frac{8}{\lambda^{2}}F_{\textrm{corr}}[\rho,\mu].
\end{align}
To finish the proof of Proposition~\ref{prop:GB_link_to_tildeT} we thus only need to check that the last two terms cancel.

\begin{lemma}\label{lem:Q_analysis} We have
\begin{align}
    F_{\textrm{corr}}[\rho,\mu]= \frac{1}{16}\int_0^{2R[\mu]} \dd{s} \qty(s + \sum_{k=0}^\infty \frac{(-s)^k}{k!} \tau_k[H_\mu\rho,\mu])^2.
\end{align}
\end{lemma}
\begin{proof}
    Let us denote the right-hand side by $G_{\textrm{corr}}$.
    By the definition \eqref{eq:taushiftedtimes},
\begin{align}
    G_{\textrm{corr}}=\frac{1}{16}\int_0^{2R[\mu]} \dd{s} \qty(\sum_{k=0}^\infty \frac{(-s)^k}{k!} (t_k[H_\mu\rho]-2^{1-k}M_{k-1}[\mu]))^2
\end{align}
Changing integration variables to $r=R[\mu]-s/2$ gives
\begin{align}
    G_{\textrm{corr}}&=\frac{1}{8} \int_0^{R[\mu]} \dd{r} \left(\sum_{k=0}^\infty \frac{(2r-2R[\mu])^k}{k!} (t_k[H_\mu\rho]-2^{1-k}M_{k-1}[\mu])\right)^2\nonumber\\
    &=\frac{1}{8} \int_0^{R[\mu]} \dd{r} \left(\sum_{k=0}^\infty \frac{(2r-2R[\mu])^k}{k!} t_k[H_\mu\rho] -2Z(r;\mu]\right)^2\nonumber\\
    &=\frac{1}{8} \int_0^{R[\mu]} \dd{r} \left(\sum_{k=0}^\infty \frac{(2r)^k}{k!} t_k[\rho] -2Z(r)\right)^2,
\end{align}
where in the second equality we use the series expansion of $Z(r)=Z(r;\mu]$ around $r=R[\mu]$ (recall from \eqref{eq:Mexplicit} that $M_k[\mu] = Z^{(k+1)}(R[\mu];\mu]$), and in the third equality we expanded $(2r-2R[\mu])^k$ as a polynomial in $r$ and made use of \eqref{eq:inversetimesubs}.

In terms of the weight $\rho$ this can be written as
\begin{align}
    G_{\textrm{corr}} &= \frac{1}{2} \int_0^{R[\mu]} \dd{r} \left(Z(r)-\int \dd{\rho(L)} I_0(L\sqrt{2r}) \right)^2.
\end{align}
In the constant, linear and quadratic term in $\rho$ we then recognize exactly the expressions \eqref{eq:F0expr}, \eqref{eq:diskfunction} and \eqref{eq:cylinderfunction},
\begin{align}
    G_{\textrm{corr}} &= F_0[\mu] + \int \rmd\rho(L_1) \frac{\delta F_0[\mu]}{\delta\mu(L_1)}+ \frac{1}{2}\int \rmd\rho(L_1)\rmd\rho(L_2) \frac{\delta F_0[\mu]}{\delta\mu(L_1)\delta\mu(L_2)}\nonumber\\
    &= F_{\textrm{corr}}[\rho,\mu].
\end{align}
\end{proof}

\subsection{Properties of the new kernel}

Recall that the new kernel is given by
\begin{align}
    K(x,t,\mu]&= \int_{-\infty}^\infty K_0(x+z,t) \dd{X(z)},\qquad
K_0(x,t)=\frac{1}{1+\exp(\frac{x+t}{2})}+\frac{1}{1+\exp(\frac{x-t}{2})},
\end{align}
where $X(z)=X(z;\mu]$ is determined by its two-sided Laplace transform $\hat X(u;\mu]$,
\begin{align}
\hat X(u;\mu]=\int_{-\infty}^\infty\dd{X(z)} e^{-uz} = \frac{ \sin(2\pi u)}{2\pi u \,\eta(u;\mu]}
\end{align}
and 
\begin{align}
    \eta(u;\mu]=\sum_{m=0}^\infty \frac{M_m[\mu]}{(2m+1)!!}u^{2m}.
\end{align}
To prove Theorem~\ref{the:tight_top_recursion}, we need to relate  $K(x,t;\mu]$ to the moments $M_k[\mu]$, since they appear in the shifted times.
We define the \emph{reverse moments} $\beta_m[\mu]$ as the coefficients of the reciprocal series
\begin{align} \label{eq:beta_def}
\frac{1}{\eta(u;\mu]} =\sum_{m=0}^\infty \beta_m[\mu] u^{2m}.   
\end{align}
Multiplying both series shows that the moments and reverse moments obey
\begin{align}\label{eq:betaMrel}
&\sum_{m=0}^p \frac{M_m[\mu]}{(2m+1)!!} \beta_{p-m}[\mu]=\delta_{p,0}
\end{align}
for each $p\geq 0$.
Note in particular that
\begin{align}
    \beta_0[\mu] = \frac{1}{M_0[\mu]}, \qquad \beta_1[\mu] = - \frac{M_1[\mu]}{3M_0[\mu]^2}.\label{eq:betaexamples}
\end{align}

\begin{proposition}
\label{prop:tildeK}
For $i,j\geq1$, the new kernel satisfies
\begin{align} 
    \int_0^\infty \dd{x}\, \frac{x^{2i-1}}{(2i-1)!}\,K(x,t;\mu]=\sum_{m=0}^{i} \beta_m[\mu]\,\frac{t^{2i-2m}}{(2i-2m)!}  
\end{align}
and
\begin{align}
    \int_0^\infty \dd{x}\int_0^\infty \dd{y}\, \frac{x^{2i-1}y^{2j-1}}{(2i-1)!(2j-1)!}\, K(x+y,t;\mu]=\sum_{m=0}^{i+j} \beta_m[\mu]\,\frac{t^{2i+2j-2m}}{(2i+2j-2m)!}.
\end{align}
\end{proposition}

\noindent
We need two lemmas to prove this proposition. First we examine the one-sided Laplace transforms
\begin{align}
\hat K_0(u,t) &\coloneqq \int_0^\infty \dd{x} e^{-ux} K_0(x,t), \label{eq:hatK0} \\
\hat K(u,t;\mu] &\coloneqq \int_0^\infty \dd{x} e^{-ux} K(x,t;\mu].\label{eq:hatK}
\end{align}

\begin{lemma} \label{lem:hatK0}
    \begin{align}
    \hat K_0(u,t)- \hat K_0(-u,t)= \frac{-4\pi \cosh(tu)}{\sin(2\pi u)}+\frac{2}{u}
    \end{align}
\end{lemma}
\begin{proof}
    To compute the integral \eqref{eq:hatK0}, we only need positive values for $x$, so we assume $x > 0$. Since $K_0(x,-t)=K_0(x,t)$ we can also assume $t \geq 0$.
    For $x<t$ we may expand
    \begin{align}
    K_0(x,t)&=\frac{\exp(\frac{-x-t}{2})}{1+\exp(\frac{-x-t}{2})}+\frac{1}{1+\exp(\frac{x-t}{2})}\nonumber\\
    &= -\sum_{p=1}^\infty \qty(-e^{\frac{-x-t}{2}})^p + \sum_{p=0}^\infty \qty(-e^{\frac{x-t}{2}})^p,
    \end{align}
    while for $x>t$ we may use
    \begin{align}
    K_0(x,t)&=\frac{\exp(\frac{-x-t}{2})}{1+\exp(\frac{-x-t}{2})}+\frac{\exp(\frac{-x+t}{2})}{1+\exp(\frac{-x+t}{2})}\nonumber\\
    &= -\sum_{p=1}^\infty \qty(-e^{\frac{-x-t}{2}})^p - \sum_{p=1}^\infty \qty(-e^{\frac{-x+t}{2}})^p.
    \end{align}
    This gives
    \begin{align}
    \hat K_0(u,t) &= -\sum_{p=1}^\infty \int_0^\infty \dd{x} e^{-ux}\qty(-e^{\frac{-x-t}{2}})^p + \sum_{p=0}^\infty \int_0^t \dd{x} e^{-ux} \qty(-e^{\frac{x-t}{2}})^p - \sum_{p=1}^\infty \int_t^\infty \dd{x} e^{-ux} \qty(-e^{\frac{-x+t}{2}})^p \nonumber\\
    &= -e^{-ut} \sum_{p=-\infty}^\infty \frac{(-1)^p}{u-p/2} -\sum_{p=1}^\infty (-1)^p \frac{\exp(\frac{-tp}{2})}{u+p/2} + \sum_{p=0}^\infty (-1)^p \frac{\exp(\frac{-tp}{2})}{u-p/2}\nonumber\\
    &= \frac{-2\pi e^{-ut}}{\sin(2\pi u)}+\frac{1}{u} + \sum_{p=0}^\infty (-1)^p e^{\frac{-tp}{2}}\qty(\frac{1}{u-p/2}-\frac{1}{u+p/2}).
    \end{align}
    When subtracting $\hat{K}_0(-u,t)$ it should be clear that the sum cancels and we easily obtain the claimed formula.
\end{proof}

\begin{lemma} \label{lem:hattildeK}
    \begin{align}
    \hat K(u,t;\mu]-\hat K(-u,t;\mu] = \hat X(u;\mu]\qty(\hat K_0(u,t) - \hat K_0(-u,t)-\frac{2}{u}) + \frac{2\hat X(0;\mu]}{u}
    \end{align}
\end{lemma}
\begin{proof}
    From the definition \eqref{eq:hatK} we obtain
    \begin{align}
    \hat K(u,t;\mu] -\hat K(-u,t;\mu] &= \int_{-\infty}^\infty \dd{X(z)} \int_0^\infty \dd{x} (e^{-ux}-e^{ux}) K_0(x+z,t)\nonumber\\
	&=\int_{-\infty}^\infty \dd{X(z)} \int_{z}^\infty \dd{x} (e^{(z-x)u}-e^{(x-z)u}) K_0(x,t)\nonumber\\
	&=\int_{-\infty}^\infty \dd{X(z)} (e^{zu} \hat{K}_0(u,t) - e^{-zu} \hat{K}_0(-u,t)) \breakline
	-\int_{-\infty}^\infty \dd{X(z)} \int_0^z \dd{x} (e^{(z-x)u}-e^{(x-z)u}) K_0(x,t).
    \end{align}
    The first integral evaluates to $\hat X(u;\mu](\hat K_0(u,t)-\hat K_0(-u,t))$.
    By changing variables $(x,z)\to(-x,-z)$ and using the symmetry of $\rmd X(z)$, we observe that the second integral is unchanged when $K_0(x,t)$ is replaced by $K_0(-x,t)$.
    Since also $K_0(x,t)+K_0(-x,t)=2$, the second integral can be calculated to give
    \begin{align}
        \int_{-\infty}^\infty \dd{X(z)} \int_0^z \dd{x} (e^{(z-x)u}-e^{(x-z)u}) K_0(x,t)&= \int_{-\infty}^\infty \dd{X(z)} \int_0^z \dd{x} (e^{(z-x)u}-e^{(x-z)u})\nonumber\\
	&=- \frac{2}{u} \int_{-\infty}^\infty \dd{X(z)} (1-e^{uz})\nonumber\\
	&=\frac{2}{u}\hat X(u;\mu]- \frac{2\hat X(0;\mu]}{u}.
    \end{align}
    Subtracting both integrals gives the desired result.
\end{proof}

\begin{proof}[Proof of Proposition \ref{prop:tildeK}]
We start by noting
\begin{align} 
    \int_0^\infty \dd{x}\, \frac{x^{2i-1}}{(2i-1)!}\,K(x,t;\mu]=-\frac{1}{2}[u^{2i-1}]  \left(\hat K(u,t;\mu]-\hat K(-u,t;\mu]\right) 
\end{align}
Using Lemma \ref{lem:hattildeK} and Lemma \ref{lem:hatK0}, we get for $i\geq1$
\begin{align} 
    \int_0^\infty \dd{x}\, \frac{x^{2i-1}}{(2i-1)!}\,K(x,t;\mu]&=-\frac{1}{2}[u^{2i-1}] \left(\hat X(u;\mu]\qty(\frac{-4\pi \cosh(tu)}{\sin(2\pi u)}) + \frac{2\hat X(0;\mu]}{u}\right) \nonumber\\
    &= [u^{2i}] \frac{ \cosh(tu)}{\eta(u;\mu]}\nonumber
    \\&= \sum_{m=0}^i \beta_m[\mu] \frac{t^{2i-2m}}{(2i-2m)!}
\end{align}
The second identity follows easily from the first by performing the integration at constant $x+y$, since 
\begin{align*}
    \int_0^{z} \frac{x^{2i-1}(z-x)^{2j-1}}{(2i-1)!(2j-1)!} \rmd x = \frac{z^{2i+2j-1}}{(2i+2j-1)!}.
\end{align*}
\end{proof}

\subsection{Proof of Theorem~\ref{the:tight_top_recursion}}
We will prove the tight topological recursion by retracing Mirzakhani's proof \cite{Mirzakhani2007a} of Witten's conjecture, which relies on the observation that her recursion formula \eqref{eq:Mirzakhanirecursion}, expressed as an identity on the coefficients of the volume polynomials $V_{g,n}(\mathbf{L})$, is equivalent to certain differential equations for the generating function $G(s;x_0,x_1,\ldots)$ of intersection numbers (see also \cite{Mulase_Mirzakhanis_2008}).
These differential equations can be expressed as the Virasoro constraints \cite{Dijkgraaf_Loop_1991,Witten_Two_1991,Mulase_Mirzakhanis_2008}
\begin{align}
    \tilde V_p e^{G(0;x_0,x_1,\ldots)}=0.
\end{align}
Here the Virasoro operators $\tilde V_{-1},\tilde V_0,\tilde V_1,\tilde V_2,\ldots$ are the differential operators acting on the ring of formal power series in $x_0,x_1,x_2,\ldots$ via
\begin{align}
    \tilde V_p
    &=-\frac{(2p+3)!!}{2^{p+1}} \pdv{x_{p+1}}+\frac{1}{2^{p+1}}\sum_{n=0}^\infty \frac{(2n+2p+1)!!}{(2n-1)!!} x_n \pdv{x_{n+p}}\breakline
    +\frac{\lambda^2}{2^{p+2}}\sum_{i+j=p-1}(2i+1)!!(2j+1)!!\pdv{x_i}\pdv{x_j}\breakline
    +\delta_{p,-1}\qty(\frac{\lambda^{-2}x_0^2}{2})+\frac{\delta_{p,0}}{16}.
\end{align}
They satisfy the Virasoro relations
\begin{align}
&\comm{\tilde V_m}{\tilde V_n}=(m-n)\tilde V_{m+n}.
\end{align}

Proposition~\ref{prop:GB_link_to_tildeT} suggests introducing the shift $x_k \to x_k + \tilde{\gamma}_k$ in $G$ with $\tilde{\gamma}_k = \delta_{k,1} - 2^{1-k}M_{k-1}[\mu]$, which satisfies 
\begin{align}
    \left( \tilde{V}_p + \frac{(2p+3)!!}{2^{p+1}} \pdv{x_{p+1}} - \frac{1}{2^{p+1}}\sum_{n=0}^\infty \frac{2^{-n}M_n[\mu]}{(2n+1)!!}(2p+2n+3)!! \pdv{x_{p+n+1}} \right)e^{G(0;x_0,x_1+\tilde{\gamma}_1,x_2+\tilde{\gamma_2},\ldots)}=0.\label{eq:stringshift}
\end{align}
We use the reverse moments $\beta_m[\mu]$ of \eqref{eq:beta_def} to introduce linear combinations
\begin{align}
    V_p&=\sum_{m=0}^\infty \beta_m[\mu] 2^{p} \left( \tilde{V}_{p+m} + \frac{(2p+2m+3)!!}{2^{p+m+1}} \pdv{x_{p+m+1}} - \frac{1}{2^{p+m+1}}\sum_{n=0}^\infty \frac{2^{-n}M_n[\mu]}{(2n+1)!!}(2p+2m+2n+3)!! \pdv{x_{p+m+n+1}} \right)\label{eq:Vpdef}
\end{align}
of these operators for all $p\geq -1$, which therefore obey
\begin{align}
    V_p \exp\left(G(0;x_0,x_1+\tilde{\gamma}_1,x_2+\tilde{\gamma_2},\ldots)\right)=0.
\end{align}
Using \eqref{eq:betaMrel} the operators $V_p$ can be expressed as 
\begin{align}
V_p&=\sum_{m=0}^\infty \beta_m[\mu] 2^{p} \left( \tilde{V}_{p+m} + \frac{(2p+2m+3)!!}{2^{p+m+1}} \pdv{x_{p+m+1}} - \frac{1}{2^{p+m+1}}\sum_{n=0}^\infty \frac{2^{-n}M_n[\mu]}{(2n+1)!!}(2p+2m+2n+3)!! \pdv{x_{p+m+n+1}} \right)\nonumber\\
&=-\frac{1}{2} (2p+3)!! \pdv{x_{p+1}}
+\frac{\lambda^2}{4}\sum_{m=0}^\infty\sum_{i+j=p+m-1}2^{-m}\beta_m[\mu](2i+1)!!(2j+1)!!\pdv{x_{i}}\pdv{x_{j}}\breakline
+\frac{1}{2}\sum_{n,m=0}^\infty 2^{-m} \beta_m[\mu] \frac{(2n+2p+2m+1)!!}{(2n-1)!!} x_{n} \pdv{x_{n+p+m}}\breakline
+\delta_{p,-1}\qty(\frac{\lambda^{-2}x_0^2\beta_0[\mu]}{4}+\frac{\beta_1[\mu]}{32})+\delta_{p,0}\frac{\beta_0[\mu]}{16}
\end{align}
In particular, after some rearranging (and shifting $p\to p-1$) we observe the identity
\begin{align}
\frac{1}{2} (2p+1)!! \frac{\partial G}{\partial x_p}
&=\frac{\lambda^2}{4}\sum_{m=0}^\infty\sum_{i+j=p+m-2}2^{-m}\beta_m[\mu](2i+1)!!(2j+1)!!\left(\frac{\partial^2 G}{\partial x_i\partial x_j}+\frac{\partial G}{\partial x_i}\frac{\partial G}{\partial x_j}\right)\breakline
+\frac{1}{2}\sum_{n,m=0}^\infty 2^{-m} \beta_m[\mu] \frac{(2n+2p+2m-1)!!}{(2n-1)!!} x_{n} \frac{\partial G}{\partial x_{n+p+m-1}}\breakline
+\delta_{p,0}\qty(\frac{\lambda^{-2}x_0^2\beta_0[\mu]}{4}+\frac{\beta_1[\mu]}{32})+\delta_{p,1}\frac{\beta_0[\mu]}{16},\label{eq:Gdiffeq}
\end{align}
where $G$ is understood to be evaluated at $G=G(0;x_0,x_1+\tilde{\gamma}_1,x_2+\tilde{\gamma_2},\ldots)$.

Substituting $x_k = t_k[\nu]$ such that $x_k + \tilde{\gamma}_k = \tau_k[\nu,\mu]$, Proposition~\ref{prop:GB_link_to_tildeT} links $G$ to the generating function 
\begin{align}
 G(0;\tau_0[\nu,\mu],\tau_1[\nu,\mu],\ldots)&=\sum_{g=0}^\infty \lambda^{2g-2}2^{3-3g}\tilde{F}_g[\nu,\mu]\nonumber\\
 &=\sum_{g=0}^\infty \lambda^{2g-2}2^{3-3g}\sum_{n=0}^\infty\frac{1}{n!}\int \rmd\nu(L_1)\cdots\rmd\nu(L_n) T_{g,n}(\mathbf{L};\mu]\label{eq:fulltightgenfun}
\end{align}
of tight Weil--Petersson volumes.
The differential equations \eqref{eq:Gdiffeq} can then be reformulated as the functional differential equation (here all partial derivates of $G$ are evaluated at $0, \tau_0[\nu,\mu],\tau_1[\nu,\mu],\ldots$)
\begin{align*}
    &\frac{\delta}{\delta \nu(L_1)}G(0;\tau_0[\nu,\mu],\tau_1[\nu,\mu],\ldots)\\
    &\quad=    \sum_{p=0}^\infty \frac{2L_1^{2p}}{4^p p!}\frac{\partial G}{\partial x_p} = \sum_{p=0}^\infty \frac{1}{L_1}\int_0^{L_1} \rmd t \frac{t^{2p}}{(2p)!}2^{1-p}(2p+1)!!\frac{\partial G}{\partial x_p}\\
    &\quad=
        \sum_{p=0}^{\infty}\frac{\lambda^2}{L_1} \int_0^{L_1}\dd{t}\,\frac{t^{2p}}{(2p)!}\sum_{m=0}^\infty\sum_{i+j=p+m-2}2^{-i-j-2}\beta_m[\mu](2i+1)!!(2j+1)!! \left(\frac{\partial^2 G}{\partial x_i\partial x_j}+\frac{\partial G}{\partial x_i}\frac{\partial G}{\partial x_j}\right)\\
    &\quad\quad+\sum_{p=0}^{\infty}\frac{1}{L_1} \int_0^{L_1}\dd{t}\,\frac{t^{2p}}{(2p)!}\int\dd{\nu(P)}\,\sum_{n,m=0}^\infty 2^{-m-p-n+2}\beta_m[\mu]
        \frac{(2n+2p+2m-1)!!}{(2n)!} P^{2n} \frac{\partial G}{\partial x_{n+p+m-1}}
    \\
    &\quad\quad+\lambda^{-2}t_0^2[\nu]\beta_0[\mu]+\frac{\beta_1[\mu]}{8}+\frac{\beta_0[\mu]L_1^2}{48}
\end{align*}
Inserting the integral identities of Proposition~\ref{prop:tildeK} this can also be expressed in terms of the kernel $K(x,t;\mu]$ as
\begin{align*}
    &\frac{\delta}{\delta \nu(L_1)}G(0;\tau_0[\nu,\mu],\tau_1[\nu,\mu],\ldots)\\
    &\quad = \frac{\lambda^2 }{4L_1} \int_0^{L_1}\dd{t}\int_0^\infty\dd{x}\int_0^\infty\dd{y}\sum_{i,j=0}^\infty\, K(x+y,t;\mu]\frac{x^{2i+1}y^{2j+1}}{4^i i!4^j j!} \left(\frac{\partial^2 G}{\partial x_i\partial x_j}+\frac{\partial G}{\partial x_i}\frac{\partial G}{\partial x_j}\right)\\
    &\qquad + \frac{1}{L_1}\int_0^{L_1}\dd{t}\int_0^\infty\dd{x} \int\dd{\nu(P)}\,\sum_{q=0}^\infty\, \left( K(x,t+P;\mu] +  K(x,t-P;\mu]\right)   \frac{x^{2q+1}}{4^q q!}       \frac{\partial G}{\partial x_q}
    \\
    &\qquad +\lambda^{-2}t_0^2[\nu]\beta_0[\mu]+\frac{\beta_1[\mu]}{8}+\frac{\beta_0[\mu]L_1^2}{48}\\
    &= \frac{\lambda^2 }{16L_1} \int_0^{L_1}\dd{t}\int_0^\infty\dd{x}\int_0^\infty\dd{y}\,xy K(x+y,t;\mu] \qty(\frac{\delta^2 G}{\delta\nu(x)\delta\nu(y)}+\frac{\delta G}{\delta\nu(x)}\frac{\delta G}{\delta\nu(y)})\breakline
    +\frac{1}{2L_1}\int_0^{L_1}\dd{t}\int_0^\infty\dd{x} \int\dd{\nu(P)}\,x \left( K(x,t+P;\mu] +  K(x,t-P;\mu]\right)   \frac{\delta G}{\delta\nu(x)}\breakline
    +\fdv{\nu(L_1)}\qty(\frac{\lambda^{-2}t_0^3[\nu]}{6M_0[\mu]}-\frac{M_1[\mu]t_0[\nu]}{48M_0[\mu]^2}+\frac{t_1[\nu]}{24M_0[\mu]}),
    \end{align*}
where in the last line we used \eqref{eq:betaexamples}.
This equation at the level of the generating function \eqref{eq:fulltightgenfun} is precisely equivalent to the recursion equation on its polynomial coefficients
\begin{align}
    T_{g,n}(\mathbf{L})&=
    \frac{1}{2L_1} \int_0^{L_1}\dd{t}\int_0^\infty\dd{x}\int_0^\infty\dd{y}\,xy K(x+y,t;\mu] \Bigg[ T_{g-1,n+1}(x,y,\mathbf{L}_{\widehat{\{1\}}})\nonumber\\
    &\mkern370mu +\!\!\!\!\sum_{\substack{g_1+g_2=g\\I\amalg J=\{2,\ldots,n\}}} T_{g_1,1+|I|}(x,\mathbf{L}_{I}) T_{g_2,1+|J|}(y,\mathbf{L}_{J})\Bigg]\nonumber\\
    &+\frac{1}{2L_1}\int_0^{L_1}\dd{t}\int_0^\infty\dd{x}\sum_{j=2}^n\,x \left(K(x,t+L_j;\mu] + K(x,t-L_j;\mu]\right) T_{g,n-1}(x,\mathbf{L}_{\widehat{\{1,j\}}}),
    \end{align}
for $(g,n)\notin \{(0,3),(1,1)\}$ combined with the initial data
\begin{align}
T_{0,3}(L_1,L_2,L_3;\mu]&=\frac{1}{M_0[\mu]}\\
T_{1,1}(L;\mu]&=-\frac{M_1[\mu]}{24M_0[\mu]^2}+\frac{L^2}{48M_0[\mu]}.
\end{align}
This completes the proof of Theorem~\ref{the:tight_top_recursion}.

\subsection{Proof of Theorem~\ref{the:recursion_n}}

We follow a strategy along the lines of the proof of Theorem~\ref{the:tight_top_recursion}.
Recall the relation \eqref{eq:fulltightgenfun} between the intersection number generating function $G$ and the tight Weil--Petersson volumes $T_{g,n}$.
Let us denote by $G_{g,n}(x_0,x_1,\ldots;\mathsf{M}_0,\mathsf{M}_1,\ldots)$ the homogeneous part of degree $n$ in $x_0,x_1,\ldots$ of  $G_g(0;x_0,x_1 + 1 - \mathsf{M}_0,x_2-\tfrac12 \mathsf{M}_1, x_3 - \tfrac{1}{4} \mathsf{M}_2, \ldots)$.
In other words, they are homogeneous polynomials of degree $n$ in $x_0,x_1,\ldots$ with coefficients that are formal power series in $\mathsf{M}_0,\mathsf{M}_1,\ldots$, such that
\begin{align}
    G_g(0;x_0,x_1 + 1 - \mathsf{M}_0,x_2-\tfrac12 \mathsf{M}_1, \ldots) = \sum_{n=0}^\infty G_{g,n}(x_0,x_1,\ldots;\mathsf{M}_0,\mathsf{M}_1,\ldots).
\end{align} 
We will prove that there exist polynomials $\bar{\mathcal{P}}_{g,n}(x_0,x_1,\ldots;m_1,m_2\ldots)$ such that
\begin{align}
    G_{g,n}(x_0,x_1,\ldots;\mathsf{M}_0,\mathsf{M}_1,\ldots) = \frac{1}{n!}\frac{1}{\mathsf{M}_0^{2g-2+n}} \bar{\mathcal{P}}_{g,n}\left(x_0,x_1,\ldots;\frac{\mathsf{M}_1}{\mathsf{M}_0},\frac{\mathsf{M}_2}{\mathsf{M}_0},\ldots\right)\label{eq:Pbarpolynomial}
\end{align}
and deduce a recurrence in $n$.

For $g\geq 2$ and $n=0$, the existence of a polynomial $\bar{\mathcal{P}}_{g,0}(m_1,m_2,\ldots)$ follows from \cite[Lemma~12]{budd2020irreducible}, since
\begin{align}
    G_{g,0}(x_0,x_1,\ldots;\mathsf{M}_0,\mathsf{M}_1,\ldots) &= G_{g}(0;0,1 - \mathsf{M}_0,-\tfrac12 \mathsf{M}_1,\ldots)\nonumber\\
    &= \mathsf{M}_0^{2-2g}G_{g}\left(0;0,0,-\tfrac12 \frac{\mathsf{M}_1}{\mathsf{M}_0},-\tfrac14 \frac{\mathsf{M}_2}{\mathsf{M}_0},\ldots\right)
\end{align}
and $G_g(0;0,0,x_2,x_3,\ldots)$ is polynomial by construction.
Also
\begin{equation}
    G_{0,3} = \frac{x_0^3}{6} \frac{1}{\mathsf{M}_0}, \qquad G_{1,1} = \frac{x_1}{24}\frac{1}{\mathsf{M}_0}-\frac{x_0}{48}\frac{\mathsf{M}_1}{\mathsf{M}_0^2}
\end{equation}

Let us now assume $2g-2+n \geq 2$ and aim to express $G_{g,n}$ in tems of $G_{g,n-1}$.
By construction the series $G_{g,n}$ obeys for $k\geq 1$,
\begin{align}
    \frac{\partial G_{g,n}}{\partial x_k} = -2^{k-1}\frac{\partial G_{g,n-1}}{\partial \mathsf{M}_{k-1}}.\label{eq:Ggnpdvrelation}
\end{align}
The string equation, i.e.\ \eqref{eq:stringshift} at $p=-1$, written in terms of $G_{g,n}$ reads
\begin{align}
    \sum_{k=1}^\infty x_k \frac{\partial G_{g,n-1}}{\partial x_{k-1}} - \sum_{k=0}^\infty 2^{-k} \mathsf{M}_k \frac{\partial G_{g,n}}{\partial x_{k}}= 0,
\end{align}
which after rearranging gives the relation
\begin{align}
     \frac{\partial G_{g,n}}{\partial x_0} = \sum_{k=1}^\infty \left(\frac{x_k}{\mathsf{M}_0} \frac{\partial G_{g,n-1}}{\partial x_{k-1}} - 2^{-k} \frac{\mathsf{M}_k}{\mathsf{M}_0} \frac{\partial G_{g,n}}{\partial x_k}\right)
\end{align}
Together with \eqref{eq:Ggnpdvrelation} this is sufficient to identify the recursion relation
\begin{align}
    G_{g,n}(x_0,x_1,\ldots;\mathsf{M}_0,\mathsf{M}_1,\ldots) &= \frac{1}{n} \sum_{k=0}^\infty x_k \frac{\partial G_{g,n}}{\partial x_k}\nonumber\\
    &= \frac{1}{n} \sum_{k=1}^\infty \left(\frac{x_0x_k}{\mathsf{M}_0}\frac{\partial G_{g,n-1}}{\partial x_{k-1}} + \frac{x_0\mathsf{M}_k}{2\mathsf{M}_0} \frac{\partial G_{g,n-1}}{\partial \mathsf{M}_{k-1}} - 2^{k-1} x_k \frac{\partial G_{g,n-1}}{\partial \mathsf{M}_{k-1}}\right).
\end{align}
By induction, we now verify that $G_{g,n}$ is of the form \eqref{eq:Pbarpolynomial}.
If \eqref{eq:Pbarpolynomial} is granted for $G_{g,n-1}$, then 
\begin{align}
    G_{g,n}(x_0,x_1,\ldots;\mathsf{M}_0,\mathsf{M}_1,\ldots) &= \frac{1}{n!}\frac{1}{\mathsf{M}_0^{2g-2+n}}\Bigg[\sum_{k=1}^\infty x_0x_k\frac{\partial \bar{\mathcal{P}}_{g,n-1}}{\partial x_{k-1}} - \sum_{k=2}^\infty \left(-x_0\frac{\mathsf{M}_k}{2\mathsf{M}_0} +2^{k-1}x_k\right)\frac{\partial \bar{\mathcal{P}}_{g,n-1}}{\partial m_{k-1}} \breakline
    + \left(-x_0\frac{\mathsf{M}_1}{2\mathsf{M}_0} +x_1\right)\left((2g+n-3)\bar{\mathcal{P}}_{g,n-1} - \sum_{k=1}^\infty m_k \frac{\partial \bar{\mathcal{P}}_{g,n-1}}{\partial m_{k}}\right)\Bigg]
\end{align}
is indeed of the form \eqref{eq:Pbarpolynomial} provided
\begin{align}
    \bar{\mathcal{P}}_{g,n}(x_0,x_1,\ldots;m_1,m_2,\ldots) &= \sum_{k=1}^\infty x_0x_k\frac{\partial \bar{\mathcal{P}}_{g,n-1}}{\partial x_{k-1}} - \sum_{p=1}^\infty \left(-\frac{x_0m_{p+1}}{2} +2^p x_{p+1}-\frac{x_0 m_1 m_k}{2}+x_1 m_k\right)\frac{\partial \bar{\mathcal{P}}_{g,n-1}}{\partial m_{p}} \breakline
    + \left(-\frac{x_0m_1}{2} +x_1\right)(2g+n-3)\bar{\mathcal{P}}_{g,n-1}
\end{align}

According to \eqref{eq:fulltightgenfun} the series $G_{g,n}$ and the tight Weil--Petersson volume $T_{g,n}$ are related via 
\begin{align}
    G_{g,n}(t_0[\nu],t_1[\nu],\ldots;M_0[\mu],M_1[\mu],\ldots) = \frac{2^{3-3g}}{n!} \int \rmd\nu(L_1)\cdots\rmd\nu(L_n) T_{g,n}(\mathbf{L};\mu].
\end{align}
This naturally leads to the existence of polynomials $\mathcal{P}_{g,n}(\mathbf{L},m_1,m_2,\ldots)$ such that 
\begin{align}\label{eq:TPrel}
    T_{g,n}(\mathbf{L};\mu] = \frac{1}{M_0^{2g-2+n}} \mathcal{P}_{g,n}\left(\mathbf{L},\frac{M_1}{M_0},\ldots,\frac{M_{3g-3+n}}{M_0}\right),
\end{align}
to get
\begin{align}
    \mathcal{P}_{g,n}(\mathbf{L},\mathbf{m}) &= \sum_{p=1}^\infty \left(m_{p+1} - \frac{L_1^{2p+2}}{2^{p+1}(p+1)!}-m_1 m_p + \frac{1}{2}L_1^2 m_p\right) \frac{\partial \mathcal{P}_{g,n-1}}{\partial m_p}(\mathbf{L}_{\widehat{\{1\}}},\mathbf{m}) \nonumber\\
    &\quad  + (2g-3+n)(-m_1+\tfrac12 L_1^2) \mathcal{P}_{g,n-1}(\mathbf{L}_{\widehat{\{1\}}},\mathbf{m}) + \ind_{\{g=0,n=3\}}
    + \sum_{i=2}^n \int \dd{L_i} L_i \, \mathcal{P}_{g,n-1}(\mathbf{L}_{\widehat{\{1\}}},\mathbf{m}).
\end{align}
The claims about the degree of the polynomials $\mathcal{P}_{g,n}$ are easily checked to be valid for the initial conditions and to be preserved by the recursion formula \eqref{eq:Pgnrecursion}.
This proves theorem \ref{the:recursion_n}.

\begin{proof}[Proof of Corollary~\ref{cor:stringdilaton}]
We note that
\begin{align}
    [L_1^{2p}]\mathcal{P}_{g,n}(\mathbf{L},\mathbf{m}) &=
    \delta_{p,0}\Bigg[\sum_{q=1}^\infty \left(m_{q+1}-m_1 m_q\right) \frac{\partial \mathcal{P}_{g,n-1}}{\partial m_q}(\mathbf{L}_{\widehat{\{1\}}},\mathbf{m})-(2g-3+n)m_1 \mathcal{P}_{g,n-1}(\mathbf{L}_{\widehat{\{1\}}},\mathbf{m}) \breakline
    + \ind_{\{g=0,n=3\}}
    + \sum_{i=2}^n \int \dd{L_i} L_i \, \mathcal{P}_{g,n-1}(\mathbf{L}_{\widehat{\{1\}}},\mathbf{m})\Bigg]\\\nonumber
    &+\delta_{p,1}\Bigg[
    \sum_{q=1}^\infty \left(\frac{1}{2}m_q\right) \frac{\partial \mathcal{P}_{g,n-1}}{\partial m_q}(\mathbf{L}_{\widehat{\{1\}}},\mathbf{m})+ \tfrac12(2g-3+n) \mathcal{P}_{g,n-1}(\mathbf{L}_{\widehat{\{1\}}},\mathbf{m})
    \Bigg]\\\nonumber
    &+\ind_{p>1}\Bigg[
    -\frac{1}{2^{p}p!} \frac{\partial \mathcal{P}_{g,n-1}}{\partial m_{p-1}}(\mathbf{L}_{\widehat{\{1\}}},\mathbf{m})
    \Bigg]
\end{align}
Setting $m_0=1$ this gives
\begin{align}
\sum_{p=0}^\infty 2^p p! m_p [L_1^{2p}]\mathcal{P}_{g,n}(\mathbf{L},\mathbf{m}) &= \ind_{\{g=0,n=3\}}
    + \sum_{i=2}^n \int \dd{L_i} L_i \, \mathcal{P}_{g,n-1}(\mathbf{L}_{\widehat{\{1\}}},\mathbf{m})
\end{align}
and
\begin{align}
\sum_{p=1}^\infty 2^p p! m_{p-1} [L_1^{2p}]\mathcal{P}_{g,n}(\mathbf{L},\mathbf{m}) &= (2g-3+n) \mathcal{P}_{g,n-1}(\mathbf{L}_{\widehat{\{1\}}},\mathbf{m}).
\end{align}
Using equation \eqref{eq:TPrel}, we get the desired result.
\end{proof}

\newpage
\section{Laplace transform, spectral curve and disk function}\label{sec:Laplace}
\subsection{Proof of Theorem \ref{the:tight_spectral_curve}}

Let us consider the partial derivative operator
\begin{align}\label{eq:bnd_insertion}
\Delta(z)=4\sum_{p=0}^\infty (2z^2)^{-1-p} (2p+1)!! \pdv{x_p}
\end{align}
on the ring of formal power series in $x_0,x_1,\ldots$ and $1/z$.
For later purposes we record several identities for the power series coefficients around $z=\infty$, valid for $a\geq0$,
\begin{align}
\qty[u^{-2-2a}]\Delta(u)&=2^{1-a}(2a+1)!!\pdv{x_a},\\
\qty[u^{-4-2a}]\Delta(u)\Delta(-u)&=2^{2-a}\sum_{i+j=a}(2i+1)!!(2j+1)!!\frac{\partial^2}{\partial x_i\partial x_j},\label{eq:Deltaidentities}\\
\qty[u^{2a-1}]\frac{1}{u\qty(z^2-u^2)\eta(u;\mu]}&=\sum_{m=0}^a z^{2m-2a-2} \beta_m[\mu],
\end{align}
where the reverse moments $\beta[\mu]$ were introduced in \eqref{eq:beta_def}.

From the definition \eqref{eq:tightlaplace} and the relation \eqref{eq:fulltightgenfun} we deduce that
for $g\geq1$ or $n\geq3$
\begin{align}
\omega_{g,n}(\mathbf{z})&= \int_0^\infty \left[\prod_{1\leq i\leq n} L_i e^{-z_i L_i}\right] \frac{\delta^n \tilde{F}_{g}[\nu,\mu]}{\delta\nu(L_1)\cdots\delta\nu(L_n)}\Bigg|_{\nu=0} \rmd L_1\cdots\rmd L_n \nonumber\\
&=2^{3g-3}\eval{\Delta(z_1)\ldots\Delta(z_n)G_g(0;x_0,x_1+\tilde{\gamma}_1,x_2+\tilde{\gamma}_2,\ldots)}_{x_0=x_1=\cdots=0},\label{eq:omegafromG}
\end{align}
where $\tilde{\gamma}_{k} = \delta_{k,1} - 2^{1-k} M_{k-1}[\mu]$ as before.
Recall the differential equation \eqref{eq:Gdiffeq} satisfied by this (shifted) intersection number generating function $G(0;x_0,x_1+\tilde{\gamma}_1,x_2+\tilde{\gamma}_2,\ldots)$,
\begin{align*}
\frac{1}{2} (2p+1)!! \frac{\partial G}{\partial x_p}
    &=\frac{\lambda^2}{4}\sum_{m=0}^\infty\sum_{i+j=p+m-2}2^{-m}\beta_m[\mu](2i+1)!!(2j+1)!!\left(\frac{\partial^2 G}{\partial x_i\partial x_j}+\frac{\partial G}{\partial x_i}\frac{\partial G}{\partial x_j}\right)\breakline
    +\frac{1}{2}\sum_{n,m=0}^\infty 2^{-m} \beta_m[\mu] \frac{(2n+2p+2m-1)!!}{(2n-1)!!} x_{n} \frac{\partial G}{\partial x_{n+p+m-1}}\breakline
    +\delta_{p,0}\qty(\frac{\lambda^{-2}x_0^2\beta_0[\mu]}{4}+\frac{\beta_1[\mu]}{32})+\delta_{p,1}\frac{\beta_0[\mu]}{16}.
\end{align*}
With the help of the identities \eqref{eq:Deltaidentities} it can be recast in terms of the operator $\Delta(z)$ as
\begin{align*} 
\Delta(z_1)G
    &=2\lambda^2\sum_{a=0}^\infty\sum_{m=0}^{a+2}\sum_{i+j=a}(2z_1^2)^{m-a-3}2^{-m}\beta_m[\mu](2i+1)!!(2j+1)!!\left(\frac{\partial^2 G}{\partial x_i\partial x_j}+\frac{\partial G}{\partial x_i}\frac{\partial G}{\partial x_j}\right)\breakline
    +4\sum_{n,m,p=0}^\infty (2z_1^2)^{-1-p} 2^{-m} \beta_m[\mu] \frac{(2n+2p+2m-1)!!}{(2n-1)!!} x_{n} \frac{\partial G}{\partial x_{n+p+m-1}}\breakline
    +\frac{4}{z_1^2}\qty(\frac{\lambda^{-2}x_0^2\beta_0[\mu]}{4}+\frac{\beta_1[\mu]}{32})+\frac{\beta_0[\mu]}{8z_1^4} \displaybreak[0]\\
    &=\frac{\lambda^2}{16}\sum_{a=0}^\infty\qty(\qty[u^{2a+3}]\frac{1}{u\qty(z_1^2-u^2)\eta(u;\mu]})\qty[u^{-4-2a}]\Big(\Delta(u)\Delta(-u)G+(\Delta(u)G)(\Delta(-u)G)\Big)\breakline
    +\frac{1}{2}\sum_{q=0}^\infty\sum_{n=0}^{q+1}\frac{2^n x_{n}}{(2n-1)!!}\qty(\qty[u^{2q-2n+1}]\frac{1}{u\qty(z_1^2-u^2)\eta(u;\mu]}) \qty[u^{-2q-2}]\Delta(u)G\breakline
    +\frac{4}{z_1^2}\qty(\frac{\lambda^{-2}x_0^2\beta_0[\mu]}{4}+\frac{\beta_1[\mu]}{32})+\frac{\beta_0[\mu]}{8z_1^4}\displaybreak[0]\\
    &=\frac{\lambda^2}{16}\residue_{u\to0}\frac{1}{u\qty(z_1^2-u^2)\eta(u;\mu]}\Big(\Delta(u)\Delta(-u)G+(\Delta(u)G)(\Delta(-u)G)\Big)\breakline
    +\frac{1}{2}\residue_{u\to0}\sum_{n=0}^\infty\frac{(2u^2)^n x_{n}}{(2n-1)!!}\frac{1}{u\qty(z_1^2-u^2)\eta(u;\mu]}\Delta(u)G\breakline
    +\frac{4}{z_1^2}\qty(\frac{\lambda^{-2}x_0^2\beta_0[\mu]}{4}+\frac{\beta_1[\mu]}{32})+\frac{\beta_0[\mu]}{8z_1^4}. 
\end{align*}
Extracting the genus-$g$ contribution, which appears as the coefficient of $\lambda^{2g-2}$, the relation \eqref{eq:omegafromG} allows us to turn this into a recursion for $\omega_{g,n}$,
\begin{align*}
\omega_{g,n}(\mathbf{z})&
    =\frac{1}{2}\residue_{u\to0}\frac{1}{u\qty(z_1^2-u^2)\eta(u;\mu]}\big[\omega_{g-1,n+1}(u,-u,\mathbf{z}_{\widehat{\{1\}}})+
    \sum_{\substack{g_1+g_2=g\\I\amalg J=\{2,\ldots,n\}}}\omega_{g_1,|I|}(u,\mathbf{z}_{I})\omega_{g_2,|J|}(-u,\mathbf{z}_{J})\big]\breakline
    +\residue_{u\to0}\sum_{j=2}^n\sum_{p=0}^\infty u^{2p}  z_j^{-2-2p}(2p+1)\frac{1}{u\qty(z_1^2-u^2)\eta(u;\mu]}\omega_{g,n-1}(u,\mathbf{z}_{\widehat{\{1,j\}}})\breakline
    +\frac{\delta_{g,0}\delta_{n,3}}{M_0[\mu]z_1^2z_2^2z_3^2}+\delta_{g,1}\delta_{n,1}\qty(-\frac{M_1[\mu]}{24M_0[\mu]^2z_1^2}+\frac{1}{8M_0[\mu]z_1^4}) \displaybreak[0]\\
    &=\residue_{u\to0}\frac{1}{2u\qty(z_1^2-u^2)\eta(u;\mu]}\Bigg[\omega_{g-1,n+1}(u,-u,\mathbf{z}_{\widehat{\{1\}}})+
    \sum_{\substack{g_1+g_2=g\\I\amalg J=\{2,\ldots,n\}}}\omega_{g_1,|I|}(u,\mathbf{z}_{I})\omega_{g_2,|J|}(-u,\mathbf{z}_{J})+\breakline\quad\quad\quad
    +\sum_{j=2}^n\qty(\frac{1}{(z_j-u)^2}+\frac{1}{(z_j+u)^2})\omega_{g,n-1}(u,\mathbf{z}_{\widehat{\{1,j\}}})\Bigg]\breakline
    +\frac{\delta_{g,0}\delta_{n,3}}{M_0[\mu]z_1^2z_2^2z_3^2}+\delta_{g,1}\delta_{n,1}\qty(-\frac{M_1[\mu]}{24M_0[\mu]^2z_1^2}+\frac{1}{8M_0[\mu]z_1^4}).
\end{align*}
Finally, if we set $\omega_{0,2}(\mathbf{z})=(z_1-z_2)^{-2}$ and $\omega_{0,0}(\mathbf{z})=\omega_{0,1}(\mathbf{z})=0$, this reduces to
\begin{align}
    \omega_{g,n}(\mathbf{z})
    &=\residue_{u\to0}\frac{1}{2u\qty(z_1^2-u^2)\eta(u;\mu]}\Bigg[\omega_{g-1,n+1}(u,-u,\mathbf{z}_{\widehat{\{1\}}})+
    \sum_{\substack{g_1+g_2=g\\I\amalg J=\{2,\ldots,n\}}}\omega_{g_1,|I|}(u,\mathbf{z}_{I})\omega_{g_2,|J|}(-u,\mathbf{z}_{J})\Bigg].
\end{align}

\subsection{Disk function}	

Due to proposition \ref{prop:VHTrelation}, there is a relation between the regular and tight Weil--Petersson volumes. In this subsection we will look at this relation in the Laplace transformed setting.

In particular, we are interested in the \emph{Laplace transformed generating functions of (regular) Weil--Petersson volumes}
\begin{align*}
    \mathcal{W}_{g,n}(\mathbf{z}) &= \int_0^\infty \dd{L_1} L_1e^{-z_1L_1} \cdots \int_0^\infty \dd{L_n}L_ne^{-z_nL_n} \frac{\delta^n F_g[\mu]}{\delta\mu(L_1)\cdots \delta \mu(L_n)},
\end{align*}
where we recall that
\begin{align*}
    \frac{\delta^n F_g[\mu]}{\delta\mu(L_1)\cdots \delta \mu(L_n)} = \sum_{p=0}^\infty \frac{1}{p!}\int  V_{g,n+p}(\mathbf{L},\mathbf{K})\,\rmd\mu(K_1)\cdots\rmd\mu(K_p).
\end{align*}

\begin{lemma}\label{lem:Hlaplace}
    We define $x_i=x_i(z_i;\mu]=\sqrt{z_i^2-2R[\mu]}$. For $g\geq 1$ or $n\geq 3$ we have
    \begin{align}
        \mathcal{W}_{g,n}(\mathbf{z}) = \omega_{g,n}(\mathbf{x}) \prod_{i=1}^n \frac{z_i}{x_i},
    \end{align}
    while for $g=0$ and $n=1,2$,
    \begin{align}
        \mathcal{W}_{0,1}(\mathbf{z}) &= - \int_0^R \frac{z_1}{(z_1^2-2r)^{3/2}} Z(r)\,\rmd r,\label{eq:disklaplace}\\
        \mathcal{W}_{0,2}(\mathbf{z}) &=\frac{z_1}{x_1}\frac{z_2}{x_2} \omega_{0,2}(\mathbf{x})- \frac{1}{(z_1-z_2)^2}.
    \end{align}
\end{lemma}
\begin{proof}
    For the first identity we wish to combine \eqref{eq:FTrelation} and \eqref{eq:tightlaplace}.
    It requires an expression for the Laplace transform of the half-tight cylinder.
    Using
    \begin{equation}
        \int_0^\infty \dd{y} \frac{z}{\sqrt{4\pi y^3}} e^{-\frac{z^2}{4y}-yL^2} = e^{-zL}, \qquad \int_0^\infty \dd{p}e^{-y\frac{p^2}{2R}}I_1(p) = e^{\frac{R}{2y}}-1,
    \end{equation}
    allows us to compute
    \begin{align*}
        \int_K^\infty \dd{L} L e^{-z L} (H(L,K) K + \delta(L-K)) 
        &=K e^{-zK}+\begin{multlined}[t] 
            \int_K^\infty \dd{L} L \qty(\int_0^\infty \dd{y} \frac{zK}{\sqrt{4\pi y^3}} e^{-\frac{z^2}{4y}-yL^2})\\
            \sqrt{\frac{2R}{L^2-K^2}}I_1\qty(\sqrt{2R(L^2-K^2)})
        \end{multlined}\\
        &=K e^{-zK}+\int_0^\infty \dd{y} \frac{zK}{\sqrt{4\pi y^3}} e^{-\frac{z^2}{4y}-yK^2} \int_0^\infty \dd{p}e^{-y\frac{p^2}{2R}}I_1(p)  \\
        &=K e^{-zK}+\int_0^\infty \dd{y} \frac{zK}{\sqrt{4\pi y^3}} e^{-\frac{z^2}{4y}-yK^2}\qty(e^{\frac{R}{2y}}-1) \\
        &=K\frac{z}{\sqrt{z^2 -2R}}e^{-K\sqrt{z^2-2R}}.
    \end{align*}
    Therefore,
    \begin{align*}
        \mathcal{W}_{g,n}(\mathbf{z}) &= \int_0^\infty T_{g,n}(\mathbf{K};\mu] \prod_{i=1}^n K_i \frac{z_i}{x_i} e^{-x_i K_i} \rmd K_i,
    \end{align*}
    which by \eqref{eq:tightlaplace} gives the first stated identity.

    For the last two identities we use that the Laplace transform of the modified Bessel function $I_0$ is given by
    \begin{align*}
        \int_0^\infty I_0(L\sqrt{2r}) L e^{-z L}\rmd L = \frac{z}{(z^2-2r)^{3/2}}.
    \end{align*}
    Then \eqref{eq:disklaplace} follows directly from \eqref{eq:diskfunction}, while for the cylinder case \eqref{eq:cylinderfunction} implies
    \begin{align*}
        \mathcal{W}_{0,2}(\mathbf{z}) &= \int_0^R \frac{z_1}{(z_1^2-2r)^{3/2}}\frac{z_2}{(z_2^2-2r)^{3/2}} \rmd r \\
        &= \frac{z_1}{\sqrt{z_1^2-2R}}\frac{z_2}{\sqrt{z_2^2-2R}} \frac{1}{\left(\sqrt{z_1^2-2R}-\sqrt{z_2^2-2R}\right)^2}-\frac{1}{(z_1-z_2)^2}\\
        &= \frac{z_1}{x_1}\frac{z_2}{x_2} \omega_{0,2}(\mathbf{x})- \frac{1}{(z_1-z_2)^2}.
    \end{align*}
\end{proof}

We finish this section by giving alternative expressions for the disk function and the series $\eta(u;\mu]$.

\begin{lemma}
    The disk function $\mathcal{W}_{0,1}(\mathbf{z})$ is related to $\eta$ via
    \begin{align*}
        \mathcal{W}_{0,1}(\mathbf{z}) = - z_1 \sqrt{z_1^2-2R}\,\eta\left(\sqrt{z_1^2-2R}\right) + \frac{z_1}{2\pi}\sin2\pi z_1 -\int  \dd{\mu(L)}\cosh(L z_1),
    \end{align*}
    valid when $4|R| < |z_1|^2$.
\end{lemma}
Note that $\mu=0$ gives $\mathcal{W}_{0,1}(\mathbf{z})=0$ as expected.

\begin{proof}
The starting point is the standard generating function \cite[Equation 10.1.39]{Abramowitz1964}
\begin{align}
    \frac{1}{u}\sin(\sqrt{u^2-2ut}) = \sum_{n=0}^\infty \frac{t^n}{n!}y_{n-1}(u)    
\end{align}
for the spherical Bessel functions $y_k(u)$ valid when $2|t| < |u|$.
Restricting to $2|R| < |x|^2$ and using the series expansion of the ordinary and spherical Bessel functions we find 
\begin{align}
	\frac{x}{2\pi}\sin(2\pi\lambda\sqrt{x^2+2R}) &= \sum_{k=0}^\infty \frac{\,(-\pi)^k}{k!} 2^kx^{2-k}\lambda^{k+1} y_{k-1}(2\pi\lambda x) R^k\nonumber\\
	&= \sum_{k,m=0}^\infty \frac{(-1)^{m}\pi^{2m+\frac12}}{k!m!} \frac{2^{k-1}x^{2(m-k+1)}\lambda^{2m+1}R^k}{\Gamma(m-k+\tfrac32)} \nonumber\\
	&= \sum_{p=-\infty}^\infty x^{2p} \sum_{m=\max(0,p-1)}^\infty \frac{(-1)^m\pi^{2m+\frac12}}{(m+1-p)!m!} \frac{2^{m-p}\lambda^{2m+1}R^{m+1-p}}{\Gamma(\tfrac12+p)} \nonumber\\
	&= \sum_{p=-\infty}^\infty x^{2p} \lambda^p \frac{\Gamma(1/2)}{\Gamma(p+1/2)2^p}\left(\frac{-\sqrt{2}\pi}{\sqrt{R}}\right)^{p-1}J_{p-1}(2\pi\lambda\sqrt{2R}).
\end{align}
Setting $\lambda=1$ now gives
\begin{equation}
    \frac{x}{2\pi}\sin(2\pi\sqrt{x^2+2R})=\sum_{p=-\infty}^\infty x^{2p}  \frac{\Gamma(1/2)}{\Gamma(p+1/2)2^p}\left(\frac{-\sqrt{2}\pi}{\sqrt{R}}\right)^{p-1}J_{p-1}(2\pi\sqrt{2R}).
\end{equation}
On the other hand we can show that
\begin{align*}
    \dv{\lambda}\qty[\frac{x}{\sqrt{x^2+2R}}\cos(2\pi\lambda\sqrt{x^2+2R})]&=-2\pi x \sin(2\pi\lambda\sqrt{x^2+2R})\nonumber\\
    &=-4\pi^2 \sum_{p=-\infty}^\infty x^{2p} \lambda^p \frac{\Gamma(1/2)}{\Gamma(p+1/2)2^p}\left(\frac{-\sqrt{2}\pi}{\sqrt{R}}\right)^{p-1}J_{p-1}(2\pi\lambda\sqrt{2R})\nonumber\\
    &=\dv{\lambda}\qty[\sum_{p=-\infty}^\infty x^{2p} \lambda^p\frac{\Gamma(1/2)}{\Gamma(p+1/2)2^p}\left(\frac{-2\pi}{\sqrt{2R}}\right)^p J_{p}(2\pi\lambda\sqrt{2R})].
\end{align*}
Integrating and setting $\lambda=(iL)/(2\pi)$ gives
\begin{equation}
    \frac{x}{\sqrt{x^2+2R}}\cosh(L\sqrt{x^2+2R}) = \sum_{p=-\infty}^\infty x^{2p}\frac{\Gamma(1/2)}{\Gamma(p+1/2)2^p} \left(\frac{L}{\sqrt{2R}}\right)^p I_p(L\sqrt{2R}),
\end{equation}
valid for $2|R| < |x|^2$.

Starting from \eqref{eq:disklaplace} and restricting to $2|R| < |x_1|^2$ we find the series expansion
\begin{align}
    \mathcal{W}_{0,1}\left(\sqrt{x_1^2+2R}\right)\frac{x_1}{\sqrt{x_1^2+2R}} &= - \int_0^R \frac{x_1}{(x_1^2 + 2R - 2r)^{3/2}} Z(r) \rmd r\nonumber\\
    &=- x_1^{-2} \int_0^R \frac{1}{\qty(1 + \frac{2R - 2r}{x_1^2})^{3/2}} Z(r) \rmd r\nonumber\\
    &= - \sum_{p=1}^\infty \frac{(-2)^p\Gamma(1/2)}{(p-1)!\Gamma(1/2-p)} x_1^{-2p} \int_{0}^R (r-R)^{p-1} Z(r)\rmd r.
\end{align}
We can use \cite[6.567.1 and 6.567.18]{Gradshteyn_Table_2015}, 
\begin{align}
    \int_{0}^R (r-R)^{p-1} \frac{\sqrt{r}}{\sqrt{2}\pi}J_1(2\pi\sqrt{2r}) \rmd r &= \frac{(-1)^{p-1}\sqrt{2}R^{p+1/2}}{\pi} \int_0^1 x^2 (1-x^2)^{p-1} J_1(2\pi\sqrt{2R}x)\dd{x}\nonumber\\
    &= (p-1)! \qty(\frac{-\sqrt{R}}{\sqrt{2}\pi})^{p+1}J_{p+1}(2\pi\sqrt{2R}),\\
    \int_{0}^R (r-R)^{p-1} I_0(L\sqrt{2r}) \rmd r &= (-R)^{p-1}2R \int_0^1 x (1-x^2)^{p-1} I_0(L\sqrt{2R}x)\dd{x}\nonumber\\
    &= -(p-1)! \qty(\frac{-\sqrt{2R}}{L})^{p}I_{p}(L\sqrt{2R}).
\end{align}
This yields
\begin{align}    
    \mathcal{W}_{0,1}\left(\sqrt{x_1^2+2R}\right)\frac{x_1}{\sqrt{x_1^2+2R}} &= 
    -\sum_{p=1}^\infty x_1^{-2p} \frac{(-2)^p\Gamma(1/2)}{\Gamma(1/2-p)}\Bigg[\left(\frac{-\sqrt{R}}{\sqrt{2}\pi}\right)^{p+1} J_{p+1}(2\pi\sqrt{2R})\breakline
    + \int  \dd{\mu(L)} \left(\frac{-\sqrt{2R}}{L}\right)^{p} I_{p}(L\sqrt{2R})\Bigg]\nonumber\\
    &= \sum_{p=-\infty}^{-1} x_1^{2p} \frac{\Gamma(1/2)}{\Gamma(1/2+p)2^p}\Bigg[\left(\frac{-\sqrt{2}\pi}{\sqrt{R}}\right)^{p-1} J_{p-1}(2\pi\sqrt{2R})\breakline
    -\int  \dd{\mu(L)} \left(\frac{L}{\sqrt{2R}}\right)^{p} I_{p}(L\sqrt{2R})\Bigg].
\end{align}
On the other hand, \eqref{eq:Mexplicit} and \eqref{eq:etadef} imply
\begin{align}
    x_1^2 \eta(x_1)&= \sum_{p=1}^\infty x_1^{2p} \frac{\Gamma(1/2)}{\Gamma(1/2+p)2^p}\Bigg[\left(\frac{-\sqrt{2}\pi}{\sqrt{R}}\right)^{p-1} J_{p-1}(2\pi\sqrt{2R})
    -\int  \dd{\mu(L)} \left(\frac{L}{\sqrt{2R}}\right)^{p} I_{p}(L\sqrt{2R})\Bigg]
\end{align}
Together with $Z(R)=0$ we may now conclude that
\begin{align}
    &\mathcal{W}_{0,1}\left(\sqrt{x_1^2+2R}\right)\frac{x_1}{\sqrt{x_1^2+2R}} + x_1^2 \eta(x_1)\nonumber\\
    &\qquad\qquad = \frac{x_1}{2\pi}\sin(2\pi\sqrt{x_1^2+2R}) - \int  \dd{\mu(L)} \frac{x_1}{\sqrt{x_1^2+2R}}\cosh(L\sqrt{x_1^2+2R}),
\end{align}
valid for $2|R| < |x_1|^2$.
Substituting $x_1 = \sqrt{z_1^2 - 2R}$ gives the desired expression.
\end{proof}

For convenience, we record an explicit expression for $\eta(u;\mu]$ that follows from this proof.
For a formal power series $F(r,u)$ in $r$ with coefficients that are Laurent polynomials in $u$, we denote by $[u^{\geq 0}]F(r,u)$ the formal power series obtained by dropping the negative powers of $u$ in the coefficients of $F(r,u)$.
Then we can write $\eta(u;\mu]$ as
\begin{equation}
    \eta(u;\mu] = [u^{\geq 0}]\left(\frac{u}{2\pi}\sin(2\pi\sqrt{u^2+2r}) - \int  \dd{\mu(L)} \frac{u}{\sqrt{u^2+2r}}\cosh(L\sqrt{u^2+2r})\right)\Bigg|_{r=R[\mu]}.
\end{equation}

\newpage
\section{JT gravity}\label{sec:JT}
The Weil--Petersson volumes play an important role in Jackiw-Teitelboim (JT) gravity \cite{jackiw1985,teitelboim1983,Saad_JT_2019}, a two-dimensional toy model of quantum gravity.
JT gravity has received significant attention in recent years because of the holographic perspective on the double-scaled matrix model it is dual to \cite{Saad_JT_2019}. 
In this section we point to some opportunities to use our results in the context of JT gravity and its extensions in which hyperbolic surfaces with defects play a role \cite{Maxfield_2021,Blommaert_2021,Okuyama2021FZZT}.
But we start with a brief introduction to the JT gravity partition function in Euclidean signature.

JT gravity is governed by the (Euclidean) action
\begin{align}
I_{\textrm{JT, bulk}}[g_{\mu\nu},\phi]=-\frac{1}{2}\int_{\mathcal{M}}\sqrt{g}\phi \qty(R+2),
\end{align}
where $\phi$ is the scalar dilaton field, $g_{\mu\nu}$ is a two-dimensional Riemannian metric and $R$ the corresponding Ricci scalar curvature.
Since we want this action to make sense when the manifold has boundaries, the boundary term
\begin{align}
I_{\textrm{JT, boundary}}[g_{\mu\nu},\phi]=-\int_{\partial\mathcal{M}}\sqrt{h}\phi \qty(K-1)
\end{align}
is included, where $h_{\mu\nu}$ is the induced metric on the boundary and $K$ is the extrinsic curvature at the boundary.
Including the topological Einstein--Hilbert term, proportional to parameter $S_0$, gives the full (Euclidean) JT action
\begin{align}
I_{\textrm{JT}}[g_{\mu\nu},\phi]=-S_0\chi+I_{\textrm{JT,bulk}}[g_{\mu\nu},\phi]+I_{\textrm{JT,boundary}}[g_{\mu\nu},\phi],
\end{align}
where $\chi$ is the Euler characteristic of the manifold.

The JT gravity partition function on a manifold $\mathcal{M}$ with $n$ boundaries of lengths $\boldsymbol{\beta} = (\beta_1,\ldots,\beta_n)$ can formally be written as
\begin{align}
    Z_n\qty(\boldsymbol{\beta})=\int_{\mathcal{M}}\mathcal{D}g\,\mathcal{D}\phi \exp(I_{\textrm{JT}}).
\end{align}
In the partition function, the dilation field $\phi$ acts as a Lagrange multiplier on $(R+2)$, therefore enforcing a constant negative curvature $R=-2$ in the bulk. 
This is why the relevant manifolds will be hyperbolic surfaces. 

Due to the Einstein--Hilbert term, we can do a topological expansion by a formal power series expansion in $e^{-S_0}$,
\begin{align}
    Z_n\qty(\boldsymbol{\beta})=\sum_{g=0}^\infty \qty(e^{-S_0})^{2g+n-2} Z_{g,n}(\boldsymbol{\beta}).
\end{align}
It has been shown by Saad, Shenker \& Stanford \cite{Saad_JT_2019} that the JT partition functions $Z_{g,n}$ for $2g+n-2>0$ can be further decomposed by splitting the surfaces into $n$ \emph{trumpets} and a hyperbolic surface of genus $g$ and $n$ geodesic boundaries with lengths $\mathbf{b}=(b_1,\ldots,b_n)$, and that the partition function measure is closely related to the Weil--Petersson measure. 
To be precise, it satisfies the identity
\begin{align}
Z_{g,n}(\boldsymbol{\beta})=\int_0^\infty  \qty(\prod_{i=1}^n b_i \dd{b_i} Z^{\textrm{Trumpet}}(\beta_i,b_i)) V_{g,n}(\mathbf{b}),
\end{align}
where $V_{g,n}(\mathbf{b})$ are the Weil--Petersson volumes and the trumpet contributions are given by
\begin{align}
Z^{\textrm{Trumpet}}(\beta,b)=\frac{1}{2\sqrt{\pi\beta}}e^{-\frac{b^2}{4\beta}}.
\end{align}
This formula is the link between JT gravity and Weil--Petersson volumes.

\begin{figure}
    \centering
    \includegraphics[height=.2\textheight]{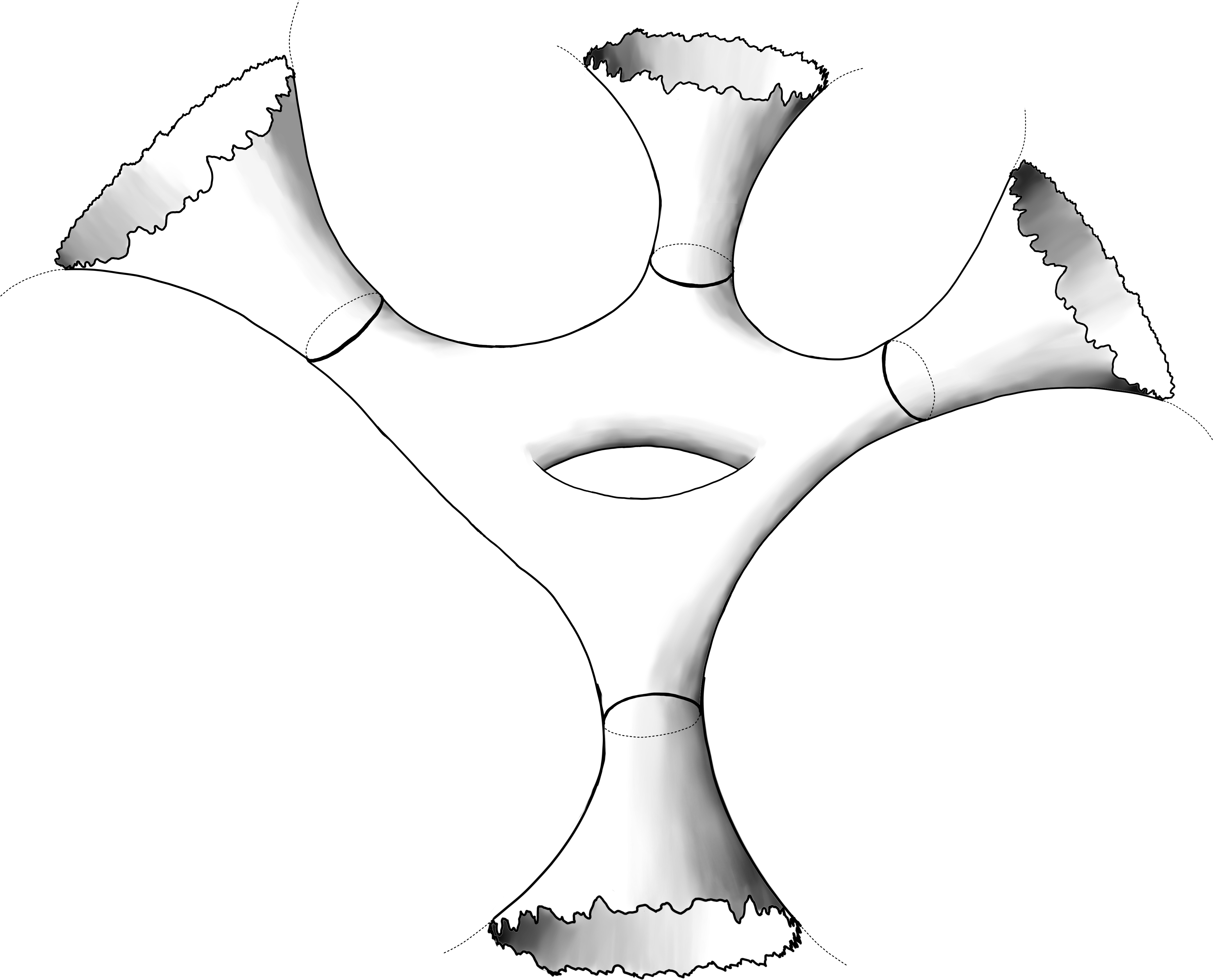}
    \colorbox{lightgray}{\includegraphics[height=.2\textheight]{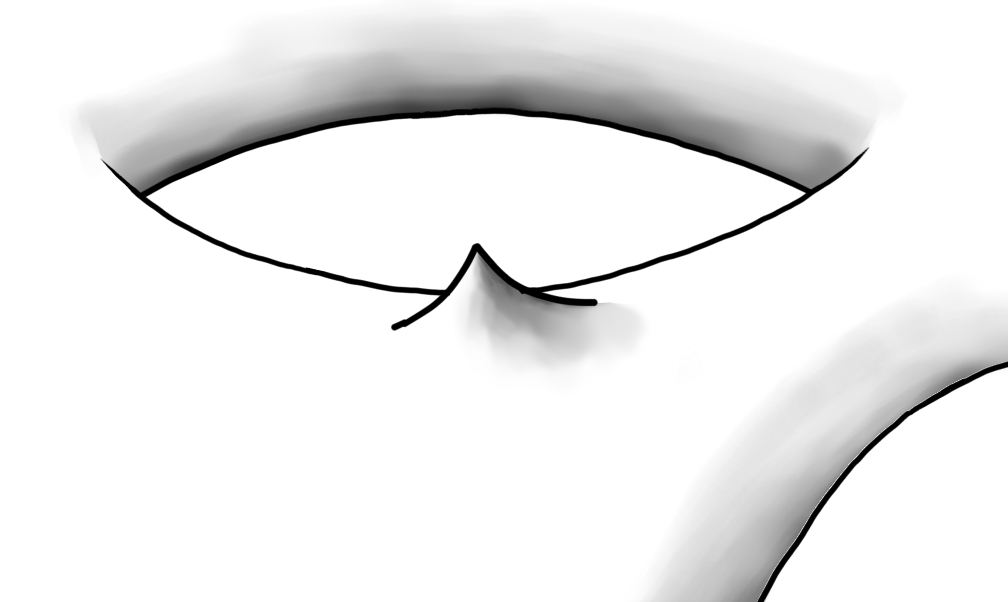}}
    \caption{Left: Visualization of decomposing the manifold into $n$ JT trumpets and a hyperbolic surface with $n$ geodesic boundaries. In this case $g=1$ and $n=4$. Right: A conical defect on a hyperbolic surface.}
\end{figure}

There are several natural extensions of the JT action. If we only allow up to two derivatives, the most general action can be transformed to \cite{Witten_2020}
\begin{align}
I_{\textrm{bulk}}[g_{\mu\nu},\phi]=-\frac{1}{2}\int_{\mathcal{M}}\sqrt{g} [\phi\qty(R+2)+U(\phi)].
\end{align}
In the next subsection we will discuss a natural choice of the dilaton potential $U(\phi)$, which gives rise to defects in the hyperbolic surfaces. 

\subsection{Conical defects}

One of the most natural dilaton potentials is
\begin{align}
U(\phi)=\mu e^{-2\pi(1-\alpha)\phi},
\end{align}
which adds a gas of conical defects of cone angle $2\pi\alpha$ carrying  weight $\mu$ each.
It naturally arises \cite{Maxfield_2021} from Kaluza--Klein instantons when performing dimensional reduction on three-dimensional black holes.

More generally, one can allow multiple types of defects by considering a measure $\mu$ on $i[0,2\pi)$ and setting  
\begin{align}
U(\phi)=\int_0^1 \rmd\mu(2\pi i\alpha) \,e^{-2\pi(1-\alpha)\phi}.
\end{align}
For instance, the choice $\mu = \sum_{j=1}^k \mu_j \delta_{i \gamma_j}$  gives $k$ types of defects with cone angles $\gamma_1,\ldots,\gamma_k \in [0,2\pi]$,
\begin{align}
    U(\phi)=\sum_i \mu_i e^{-2\pi(1-\alpha_i)\phi}.
\end{align}
The choice to consider the measure on the imaginary interval will be convenient later.

It can be shown \cite{Maxfield_2021} that these potentials indeed lead to conical defects. 
For example, one can look at the term linear in $\mu$ in the integrand of the partition function for a single type of gas:
\begin{align}
&[\mu^1]\exp(-I_{\textrm{bulk}})=\frac{1}{2}\exp(-I_{\textrm{JT, bulk}})\int_{\mathcal{M}}\dd{x_1} \sqrt{g(x_1)} \exp(-2\pi(1-\alpha)\phi(x_1))\nonumber\\
&\qquad=\frac{1}{2}\int_{\mathcal{M}}\dd{x_1} \sqrt{g(x_1)}\exp(\frac{1}{2}\int_{\mathcal{M}}\dd{x} \sqrt{g(x)}\phi(x) \qty(R(x)+2-4\pi(1-\alpha)\delta^2(x-x_1))).
\end{align}
It follows that the surface has curvature $R=-2$ everywhere, except at point $x_1$, where we have a conical defect with cone angle $2\pi\alpha$.
If one includes all orders of $\mu$, any number of defects may appear and each defect carries a weight $\mu$ \cite{Maxfield_2021}.

As already mentioned in the introduction, the Weil--Petersson volumes for surfaces with sharp cone points (cone angle $2\pi\alpha< \pi$) are obtained \cite{Tan_Generalizations_2006,Mirzakhani2007,Do_Weil_2009,do2011moduli} from the usual Weil--Petersson volume polynomials by treating the defect angle $2\pi\alpha$ as a geodesic boundary with imaginary boundary length $2\pi i\alpha$.
This that partition function $Z_{g,n}(\mathbf{\beta})$ is closely related to the generating function $F_g[\mu]$ of Weyl--Petersson volumes considered in this paper.
To be precise, using \eqref{eq:wppartitionfunction} and \eqref{eq:FTrelation},
\begin{align}
Z_{g,n}(\boldsymbol{\beta})&=\int_0^\infty  \qty(\prod_{i=1}^n b_i \dd{b_i} Z^{\textrm{Trumpet}}(\beta_i,b_i))\sum_{p=0}^\infty \frac{1}{p!}\int \rmd\mu(b_{n+1})\cdots\rmd\mu(b_{n+p}) V_{g,n+p}(\mathbf{b})\nonumber\\
&=\int_0^\infty  \qty(\prod_{i=1}^n b_i \dd{b_i} Z^{\textrm{Trumpet}}(\beta_i,b_i))\frac{\delta^n F_g[\mu]}{\delta\mu(b_{1})\cdots\delta\mu(b_{n})}\nonumber\\
&=\int_0^\infty  \qty(\prod_{i=1}^n b_i \dd{b_i} Z^{\textrm{Trumpet}}(\beta_i,b_i))\left(\prod_{i=1}^n \dd{K_i}\left(K_i\,H(b_i,K_i;\mu]+\delta(K_i-b_i)\right)\right) T_{g,n}(\mathbf{K};\mu],\label{eq:Zgntight}
\end{align}
which we can compute using the recursions described in this paper. 
In particular, its topological recursion can in principle be derived from that of $T_{g,n}$ in Theorem~\ref{the:tight_top_recursion}.

\begin{figure}
    \centering
    \includegraphics[height=.2\textheight]{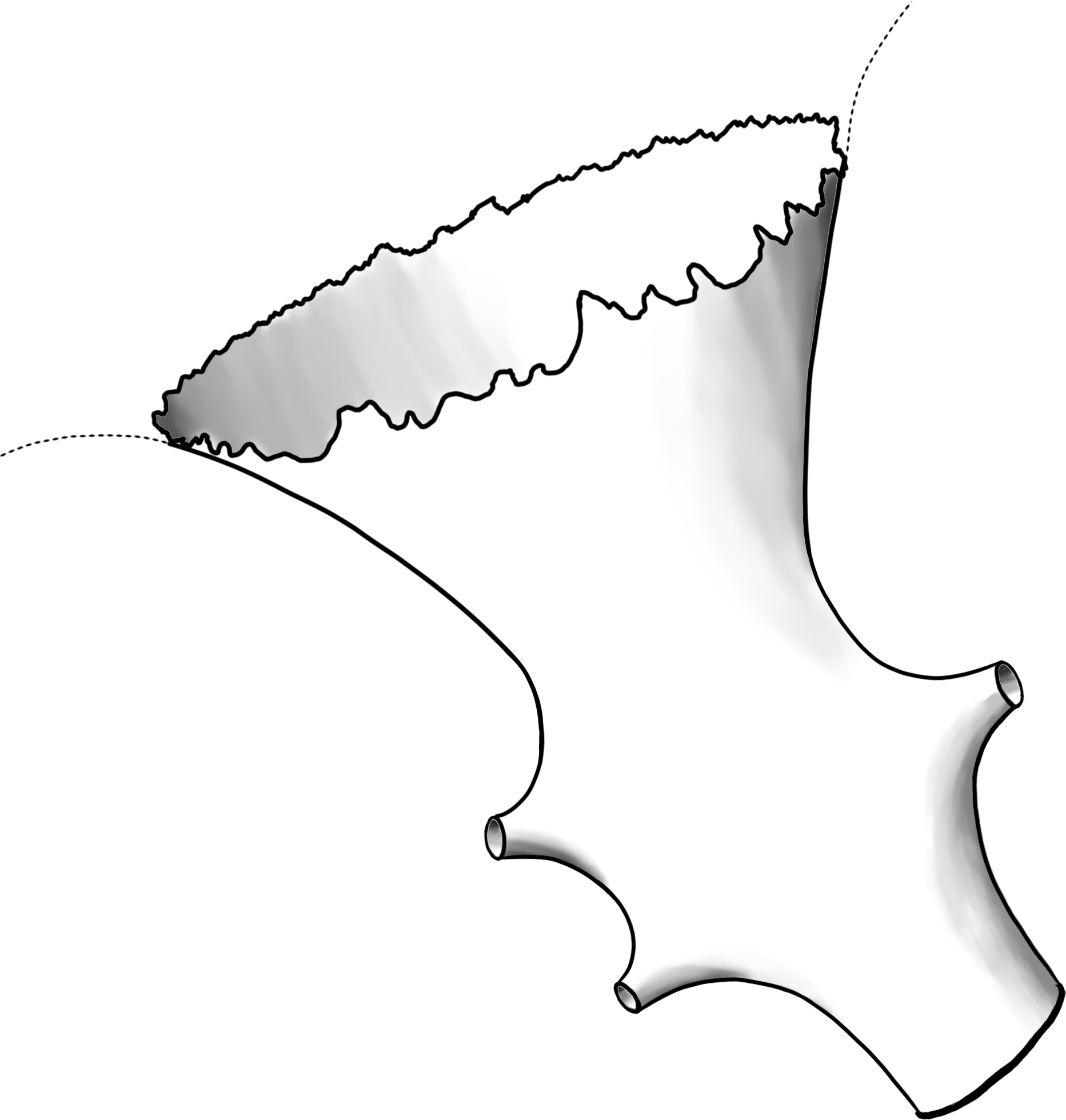}
    \caption{Visualization of the tight trumpet.\label{fig:tighttrumpet}}
\end{figure}

We can simplify \eqref{eq:Zgntight} by considering the \emph{tight trumpet}, which is a genus-$0$ hyperbolic surface with an asymptotic boundary of length $\beta$, a tight boundary of length $K$ and an arbitrary number of extra geodesic boundaries, with the constraint that the tight boundary cannot be separated from the asymptotic one by a curve of length $\beta$.
See Figure~\ref{fig:tighttrumpet}.
Since it can be obtained by gluing a trumpet to a half-tight cylinder, with the help of Lemma~\ref{lem:halftight} we find that the partition function associated to a tight trumpet is given by
\begin{align}
    Z^{\textrm{TT}}(\beta,K)&=\int_K^\infty \frac{b}{K} \dd{b} Z^{\textrm{Trumpet}}(\beta,b)\left(K\,H(b,K;\mu]+\delta(K-b)\right)\nonumber\\
    &= \frac{1}{2\sqrt{\pi\beta}}e^{-\frac{K^2}{4\beta}} + \int_K^\infty b\dd{b} \frac{1}{2\sqrt{\pi\beta}}e^{-\frac{b^2}{4\beta}} \sqrt{\frac{2R[\mu]}{b^2-K^2}} I_1\left( \sqrt{b^2-K^2}\sqrt{2R[\mu]}\right)\nonumber\\
    &=\frac{1}{2\sqrt{\pi\beta}}e^{-\frac{K^2}{4\beta}+2R[\mu]\beta} = e^{2R[\mu]\beta}Z^{\textrm{Trumpet}}(\beta,K).
\end{align}
Remarkably it differs from the JT trumpet only in a factor exponential in the boundary length $\beta$. 
We conclude that for $g\geq 1$ or $n\geq 3$,
\begin{align}
    Z_{g,n}(\boldsymbol{\beta})=  \int_0^\infty \left(\prod_{i=1}^n K_i \dd{K_i}Z^{\textrm{TT}}(\beta_i,K_i)\right) T_{g,n}(\mathbf{K};\mu],
\end{align}
which we understand as a gluing of tight trumpets to tight hyperbolic surfaces.
In the case $g=0$ and $n=2$, we only need to glue two tight trumpets together to find the universal two-boundary correlator
\begin{align}
    Z_{0,2}(\beta_1,\beta_2) &= \int_0^\infty Z^{\textrm{TT}}(\beta_1,K) Z^{\textrm{TT}}(\beta_2,K) \,K\,\rmd K\nonumber\\
    &= \frac{1}{2\pi}\frac{\sqrt{\beta_1\beta_2}}{\beta_1+\beta_2} e^{2(\beta_1+\beta_2)R[\mu]}.
\end{align}

We note that these expressions do not apply to the case of blunt cone points (cone angle $2\pi\alpha \in [\pi,2\pi]$).
The problem is that in the presence of such defects it is no longer true that every free homotopy class of closed curves necessarily contains a geodesic, because, informally, when shortening a closed curve it can be pulled across a blunt cone point, while that never happens for a short one.
However, this is not an issue when considering tight cycles, because in that setting one is considering larger homotopy classes, namely of the manifold with its defects closed off.
Such homotopy classes will always contain a shortest geodesic, which generically is unique.
Whereas the JT trumpet cannot always be removed from a surface with blunt defects in a well-defined manner, the removal of a tight trumpet should pose no problem.
It is natural to ask whether such reasoning can be used to connect to the recent works \cite{Turiaci_2021,Eberhardt_2D_2023} in JT gravity dealing with blunt cone points.

\subsection{FZZT-branes}
Another well-studied extension of JT gravity, is the introduction of FZZT branes. With this extension the hyperbolic surfaces can end on a FZZT brane. In the random matrix model description of JT gravity, this corresponds to fixing some eigenvalues of the random matrix \cite{Blommaert_2021}.

In the partition function, this leads to the addition of an arbitrary number of geodesic boundaries as defects with a certain weight $\mathcal{M}(L)=-e^{-zL}$, where $L$ is the length of the boundary,
\begin{align} \label{eq:FZZT}
    Z_{g,n}(\boldsymbol{\beta})_{\textrm{FZZT}}&= \sum_{p=0}^\infty \frac{e^{-S_0p}}{p!}\int_0^\infty  \qty(\prod_{i=1}^n \dd{b_i} b_i Z^{\textrm{Trumpet}}(\beta_i,b_i)) \qty(\prod_{i=n+1}^{n+p} \dd{b_i} \mathcal{M}(b_i)) V_{g,n+p}(\mathbf{b}).
\end{align}
Such weights have been interpreted \cite{Okuyama2021FZZT}\footnote{Please note that $z$ in our work corresponds to $z/(\sqrt{2}\pi) $ in \cite{Okuyama2021FZZT}} as the action of a fermion with mass $z$.

Using our setup, we can rewrite this to:
\begin{align} 
    Z_{g,n}(\boldsymbol{\beta})_{\textrm{FZZT}}&= \int_0^\infty \qty(\prod_{i=1}^n \dd{b_i} b_i Z^{\textrm{Trumpet}}(\beta_i,b_i))\breakline \left(\prod_{i=1}^n \dd{K_i} (K_i H(b_i,K_i;\mu_{\textrm{FZZT}}] +\delta(b_i-K_i))\right) T_{g,n}(\mathbf{K};\mu_{\textrm{FZZT}}] ,
\end{align}
with
\begin{align}
\mu_{\textrm{FZZT}}=-e^{-S_0-zL}\rmd L,
\end{align}
or again using the tight trumpet
\begin{align}
    Z_{g,n}(\boldsymbol{\beta})_{\textrm{FZZT}}=  \int_0^\infty \left(\prod_{i=1}^n K_i \dd{K_i}Z^{\textrm{TT}}(\beta_i,K_i;\mu_{\textrm{FZZT}}]\right) T_{g,n}(\mathbf{K};\mu_{\textrm{FZZT}}],
\end{align}
with
\begin{align}
    Z^{\textrm{TT}}(\beta,K;\mu]
    &=\frac{1}{2\sqrt{\pi\beta}}e^{-\frac{K^2}{4\beta}+2R[\mu]\beta}.
\end{align}
The behaviour of $R[\mu_{\textrm{FZZT}}]$ depends on $z$ and $S_0$ and its critical points should give insight into critical phenomena of the partition function, see \cite{castro2023critical}.

\end{document}